\documentclass[11pt,reqno]{article}

\RequirePackage{amsthm,amsmath,amsfonts,amssymb}
\RequirePackage[authoryear]{natbib}
\RequirePackage[colorlinks,citecolor=blue,urlcolor=blue]{hyperref}
\RequirePackage{graphicx}

\usepackage{amsmath, amssymb, amscd, amsthm, amsfonts}
\usepackage{graphicx}
\usepackage[perpage,symbol*]{footmisc}
\usepackage{float}
\usepackage{hyperref}
\usepackage{pgfplots}
\usepackage{color}  
\usepackage{tikz}

\usetikzlibrary{shapes, arrows, calc, arrows.meta, fit, positioning} 
\tikzset{  
    -Latex,auto,node distance =1.5 cm and 1.3 cm, thick,
    state/.style ={ellipse, draw, minimum width = 0.9 cm}, 
    point/.style = {circle, draw, inner sep=0.18cm, fill, node contents={}},  
    bidirected/.style={Latex-Latex,dashed}, 
    el/.style = {inner sep=2.5pt, align=right, sloped}  
}  

\usepackage{natbib}
\usepackage{authblk}

\usepackage{algorithm,algcompatible,amsmath}
\DeclareMathOperator*{\argmax}{\arg\!\max}%

\algnewcommand\INPUT{\item[\textbf{Input:}]}%
\algnewcommand\OUTPUT{\item[\textbf{Output:}]}%

\usepackage{enumitem}

\usepackage[english]{babel}
\setcitestyle{authoryear,open={(},close={)}}

\usepackage{colortbl,dcolumn}

\usepackage{listings} 
\usepackage{xcolor} 

\renewcommand{\raggedright}{\leftskip=0pt \rightskip=0pt plus 0cm}
\raggedright

\def\hilite<#1>{%
	\temporal<#1>{\color{blue!25}}{\color{magenta}}%
	{\color{blue!55}}}

\newcolumntype{H}{>{\columncolor{blue!20}}c!{\vrule}}
\newcolumntype{H}{>{\columncolor{blue!20}}c}

\newtheorem{theorem}{Theorem}

\newtheorem{definition}{Definition}

\newtheorem{lemma}{Lemma}
\newtheorem{assumption}{Assumption}

\numberwithin{equation}{section} 
\numberwithin{theorem}{section}
\numberwithin{lemma}{section} 
\numberwithin{corollary}{section}
\numberwithin{definition}{section}
\numberwithin{proposition}{section} 
\numberwithin{remark}{section}
\numberwithin{example}{section}

\renewcommand{\oddsidemargin}{0mm}

\usepackage{setspace}
\def\R{\mathbb R}

\definecolor{mydarkgreen}{rgb}{0,0.4,0}

\makeatletter
\def\@makefnmark{}

\makeatother
\usepackage{amsfonts,amssymb}
\usepackage{mathrsfs}

\begin{document}
\vspace{-15in}

	\title{\Large {\textbf{Population-level Task-evoked Functional Connectivity\\ via Fourier Analysis}}}
	
	\author[1,*]{Kun Meng}
	\author[2]{Ani Eloyan}
		
	\affil[1]{\small Division of Applied Mathematics, Brown University, Providence, Rhode Island, USA}
\affil[2]{Department of Biostatistics, Brown University School of Public Health, Providence, Rhode Island, USA}
	\affil[*]{Correspondence: Kun Meng, e-mail: \texttt{kun\_meng@brown.edu}.}
	
	\maketitle

\begin{abstract}
Functional magnetic resonance imaging (fMRI) is a non-invasive and in-vivo imaging technique essential for measuring brain activity. Functional connectivity is used to study associations between brain regions, either while study subjects perform tasks or during periods of rest. In this paper, we propose a rigorous definition of task-evoked functional connectivity at the population level (ptFC). Importantly, our proposed ptFC is interpretable in the context of task-fMRI studies. An algorithm for estimating the ptFC is provided. We present the performance of the proposed algorithm compared to existing functional connectivity frameworks using simulations. Lastly, we apply the proposed algorithm to estimate the ptFC in a motor-task study from the Human Connectome Project. 

\end{abstract}

\noindent%
{\it Keywords:}\footnote{
\begin{itemize}
\item The project herein was supported by Grant Number 5P20GM103645 from the National Institute of General Medical Sciences.
    \item \textbf{Abbreviations:} BOLD, blood-oxygenation level-dependent; DC, dynamic connectivity; fMRI, functional magnetic resonance imaging; FC, functional connectivity; HCP, Human Connectome Project; HRF, hemodynamic response function; MLM, mixed linear model; ptFC, task-evoked functional connectivity at the population level; ptFCE, ptFC estimation; PH, persistent homology; ROI, region of interest; SBHM, spatial Bayesian hierarchical model; WSMZ, weakly stationary with mean zero.
\end{itemize}
} AMUSE algorithm; Human Connectome Project; motor-task; weakly stationary with mean zero.

\maketitle

\section{Introduction}\label{section: Introduction}

\textit{Functional magnetic resonance imaging} (fMRI) is a non-invasive brain imaging technique used to estimate both brain regional activity and interactions between brain regions. In fMRI studies, brain signals are measured on 3-dimensional volume elements (\textit{voxels}) during a certain period of time, and each signal is observed at discrete time points, e.g., to investigate neural activity. However, since neural activity occurs in milliseconds, it is impossible to directly observe it using fMRI technology. The neural activity is implicitly captured by \textit{blood-oxygenation level-dependent} (BOLD) signals. When neural activity in a brain area occurs, it is followed by localized changes in metabolism, where the corresponding local oxygen consumption increases, and then oxygen-rich blood flows to this area. This process results in an increase in oxyhemoglobin and a decrease in deoxyhemoglobin. The BOLD signal value at each time point is the difference between the oxyhemoglobin and deoxyhemoglobin levels. The signal consisting of BOLD values across all time points measures the localized metabolic activity influenced by the local brain vasculature, and it indirectly measures the localized neural activity. 

\subsection{Goal and Main Contribution}

During an fMRI experimental study, subjects either rest (resting-state fMRI) or perform tasks (task-fMRI). We focus on task-fMRI throughout this paper. While the experiment may involve one or multiple tasks, we model the effect of one of the tasks, which is considered the task of interest. In this paper, we propose a framework for modeling the functional connectivity evoked by the single task of interest, referred to as task-evoked functional connectivity, while ignoring the effect of the nuisance tasks.


Dependencies between activation in any two regions of the brain are referred to as {\it functional connectivity} \citep[FC,][]{cribben2016functional}. These dependencies have been widely studied, and several approaches have been proposed for the estimation of FC, especially resting-state FC (i.e., the FC derived from resting-state fMRI data). Task-evoked FC is fundamentally different from resting-state FC. \cite{lynch2018task} show that existing estimates of task-evoked FC cannot explain differences between the resting-state FC and the task-evoked FC. While the brain FC in an individual subject is often of interest (referred to as subject-level FC), population-level estimates of FC have been proposed in the literature and used for obtaining more robust subject-level FC estimates \citep[e.g.,][]{bowman2008bayesian, mejia2018improved}. To overcome the limitations of existing task-evoked FC estimation approaches when investigating population-level FC,  we propose a novel definition of \textit{population-level task-evoked FC} (ptFC). This definition is based on a correlation of subject-specific effects that capture task-evoked neural activity. Importantly, our proposed ptFC framework takes into account the complex biological processes of the brain response to task stimuli that many existing estimators ignore. 

Under the model assumption of the Bayesian hierarchical framework proposed by \cite{bowman2008bayesian}, our proposed ptFC is equal to their ``task-related inter-region functional connectivity." However, our proposed ptFC and our proposed approach to estimating ptFC do not depend on the referred Bayesian hierarchical framework and work for more general structures (see Section \ref{section: Definition of ptFC} for details). Furthermore, with appropriate estimation adjustments, our proposed framework has the potential to be an extension of many existing models for fMRI, particularly the Bayesian models that include grouped random terms representing signals. \cite{bowman2014brain} and \cite{zhang2015bayesian} provide thorough reviews of existing Bayesian models for fMRI data. Exploration of applying our framework to the existing models is left for future research.

\subsection{Task-fMRI Signal Types}


We are interested in modeling neural activity and BOLD signals of a set of \textit{subjects} $\omega$ from a population of interest. At each time point $t$, we denote by $Y_k(\omega;t)$ the BOLD signal value at the $k^{th}$ node of the subject $\omega$'s brain. The word ``node" refers to either a voxel or a \textit{region of interest} (ROI). Aggregated BOLD signals at a macro-area level are often of interest. That is, for each subject, BOLD signals are spatially averaged within each pre-selected ROI, and the resulting ROI-specific signals are analyzed. The BOLD signals for subject $\omega$ are represented by a vector-valued function $\{\pmb{Y}(\omega;t)=\left(Y_1(\omega;t), \cdots, Y_K(\omega;t)\right)^T \vert \, t\in \mathcal{T}\}$ on $\mathcal{T}$,
where $\mathcal{T}$ is the collection of time indices and $K$ is the number of nodes. Furthermore, we assume that $\mathcal{T}$ is a compact subset of $\mathbb{R}$. 

In this paper, we model the following three types of signals in task-fMRI studies:
\begin{enumerate}[label=\roman*]
    \item \textit{Stimulus signals}, denoted by $N(t)$, representing the experimental designs of tasks. Specifically, $N(t)=1$ when the stimulus of interest is present, and $N(t)=0$ when the stimulus is absent (see Eq.~\eqref{eq: stimulus signal for the HCP data} for an example). Given a \textit{hemodynamic response function} \citep[HRF,][]{lindquist2008statistical}, a stimulus signal $N(t)$ determines the design matrix of a \textit{general linear model} \citep[GLM,][]{friston1995analysis}. Further details are provided in Section \ref{section: Models for Signals}.
    \item \textit{Task-evoked neural activity signals} at the $k^{th}$ node, denoted by $\Phi_k[\omega;\, N(t)]$, stemming solely from the task stimulus $N(t)$, where the ``stimulus-to-activity" maps $\Phi_k[\omega; \, \bullet]: \, N(t)\mapsto \Phi_k[\omega; \, N(t)]$ depend on subjects $\omega$. The map $\Phi_k[\omega;\bullet]$ characterizes how neurons at the $k^{th}$ node of the subject $\omega$ react to stimulus $N(t)$.
    \item Observed BOLD signals $Y_k(\omega;t)=\Psi_k\left\{\omega;\, \Phi_k\left[\omega; \, N(t)\right]\right\}$ that are associated with the task-evoked neural activity $\Phi_k[\omega;\, N(t)]$, where the ``activity-to-BOLD" maps $\Psi_k\{\omega; \bullet\}:\,\Phi_k[\omega; N(t)]\mapsto Y_k(\omega;t)$ are $\omega$-dependent. The map $\Psi_k\{\omega;\bullet\}$ describes how neural activity $\Phi_k[\omega;N(t)]$ induces BOLD signal $Y_k(\omega;t)$. A comprehensive discussion on the ``activity-to-BOLD" mechanism can be found in Chapter 3 of \cite{ashby2019statistical}.
\end{enumerate}
The relationship between the three types of signals is illustrated in Figure \ref{fig: three kinds of signals}. We are interested in modeling $\Phi_k[\omega; N(t)]$. However, only BOLD signals $Y_k(\omega;t)$ are observable. The goal of many fMRI studies, including our work, is to recover  $\Phi_k[\omega;\bullet]$ by analyzing $Y_k(\omega;t)$. In Section \ref{section: A New Definition of FC}, we provide nonparametric models for $\Phi_k[\omega;\bullet]$ and $\Psi_k\{\omega;\bullet\}$. 

\begin{figure}[ht]
    \centering
    \begin{tikzpicture}[main/.style = {draw, rectangle}, node distance = 8.1cm, auto] 
\node[main] (1) {$N(t)$}; 
\node[main] (2) [right of=1] {$\Phi_k[\omega; \, N(t)]$};
\node[main] (3) [right of=2] {$Y_k(\omega;\, t)$};

\draw[->] (1) --node{``stimulus-to-activity" map $\Phi_k[\omega;\,\bullet]$} (2);
\draw[->] (2) --node{``activity-to-BOLD" map $\Psi_k\{\omega;\, \bullet\}$} (3);
\end{tikzpicture}
    \caption{Three types of signals at the $k^{th}$ node of subject $\omega$: stimulus signal $N(t)$, task-evoked neural activity signal $\Phi_k[\omega; N(t)]$, and observed BOLD signal $Y_k(\omega;t)=\Psi_k\left\{\omega; \Phi_k\left[\omega; N(t)\right]\right\}$.}\label{fig: three kinds of signals}
\end{figure}


Theoretically, time $t$ is a continuous variable, and $\mathcal{T}=[0,t^*]$, where $t^*<\infty$ denotes the end of the experiment of interest. In applications, we obtain data only at discrete and finite time points in $\mathcal{T}=\{\tau \Delta\}_{\tau=0}^T$, where $\Delta$ is the predetermined \textit{repetition time} (TR) and $T$ indicates that BOLD signals are observed at $T+1$ time points. A BOLD signal at the $k^{th}$ node of subject $\omega$ in task-fMRI consists of three components: 
\begin{enumerate}
    \item $P_k(\omega;t)$ denotes the component that is evoked solely by the experimental task $N(t)$ of interest.
    \item $Q_k(\omega;t)$ denotes the component that stems from spontaneous brain activity, e.g., the activity coordinating respiration and heartbeat, and the neural activity responding to stimuli that we are not interested in.
    \item Random error $\epsilon_k(\omega;t)$.
\end{enumerate}
We assume the components have an additive structure. Observed BOLD signals $Y_k(\omega;t)$ are of the following form.
\begin{align}\label{eqn: BOLD signal decomposition}
    Y_k(\omega;t)=P_k(\omega;t) + Q_k(\omega;t) + \epsilon_k(\omega;t), \ \ k=1,2,\cdots,K,\ \ t\in\mathcal{T},
\end{align}
where $P_k(\omega;t)$ are called \textit{task-evoked terms}. $P_k(\omega;t)$ are of primary interest and identifiable under some probabilistic conditions. The proof of the identifiability of $P_k(\omega;t)$ is provided in Appendix C using a theorem proposed by \cite{tong1991indeterminacy}. The model in Eq.~\eqref{eqn: BOLD signal decomposition} is equivalent to many existing approaches for modeling BOLD signals in task-fMRI \citep{bowman2008bayesian, joel2011relationship, zhang2013semi, warnick2018bayesian}. The  relationship between Eq.~\eqref{eqn: BOLD signal decomposition} and these approaches is discussed in Section \ref{section: Definition of ptFC} and Appendix A. Similar to these existing methods, we assume no interaction between $P_k(\omega;t)$ and $Q_k(\omega;t)$ in Eq.~\eqref{eqn: BOLD signal decomposition}. The inclusion of an interaction between these terms is discussed in Appendix H.


In statistical analysis of task-fMRI data $\{\pmb{Y}(\omega;t)\}_{t\in\mathcal{T}}$ in Eq.~\eqref{eqn: BOLD signal decomposition}, two topics are primarily of interest: (i) identification of nodes presenting task-evoked neural activity, i.e., the indices $k$ such that $P_k(\omega;t)\ne 0$, and (ii) detection of task-evoked associations between brain nodes. GLMs are commonly implemented to detect task-evoked nodes \citep{lindquist2008statistical}. Each GLM is conducted at an individual node level and is not informative for investigating associations between nodes. FC characterizes the associations between nodes captured by BOLD signals \citep{friston1993functional}. FC observed during a task experiment tends to be different from that observed in a resting-state experiment \citep{lowe2000correlations}. If the subject $\omega$ performs only one task $N(t)$ of interest, the difference between task-fMRI and resting-state fMRI is represented by the task-evoked terms $P_k(\omega;t)$. In this paper, we investigate the associations between task-evoked terms $P_1(\omega;t),\ldots, P_K(\omega;t)$ (as opposed to the associations between BOLD signals $Y_1(\omega;t),\ldots, Y_K(\omega;t)$) corresponding to FC stemming solely from the task $N(t)$ of interest.


A considerable amount of work has been done to define and estimate FC. \cite{friston1993functional} defines FC as the temporal Pearson correlation between a pair of BOLD signals across time. Since we focus on task-evoked functional connectivity in this paper, we may tentatively consider applying the temporal Pearson correlation approach to task-evoked terms, i.e., we may measure the task-evoked FC between nodes $k$ and $l$ by the following
\begin{align}\label{eq: raw Pearson correlation}
    \left\vert corr(\omega;P_k, P_l)\right\vert:=\left\vert\frac{\int_{\mathcal{T}} P_k^*(\omega;t)\times P_l^*(\omega;t)\mu(dt)}{\sqrt{ \int_{\mathcal{T}}\left\vert P_k^*(\omega;t)\right\vert^2 \mu(dt) \times \int_{\mathcal{T}}\left\vert P_l^*(\omega;t)\right\vert^2 \mu(dt)}}\right\vert, 
\end{align}
where $P^*_{k'}(\omega;t)=P_{k'}(\omega;t)-\frac{1}{\mu(\mathcal{T})}\int_{\mathcal{T}}P_{k'}(\omega;s)\mu(ds)$ for $k'\in\{k,l\}$; if $\mathcal{T}=[0,T^*]$, then $\mu(dt)=dt$; if $\mathcal{T}=\{\tau \Delta\}_{\tau=0}^T$, then $\mu(dt)$ is the counting measure $\sum_{\tau\in\mathbb{Z}}\delta_{\tau \Delta}(dt)$, where $\delta_{\tau \Delta}$ is the point mass at $\tau \Delta$. There are many other approaches to defining and estimating FC, e.g., \textit{coherence analysis} \citep{muller2001multivariate} and \textit{beta-series regression} \citep{rissman2004measuring} are widely used in FC studies. These approaches, especially that in Eq.~\eqref{eq: raw Pearson correlation}, have several limitations. In Section \ref{section: advantages of the task-fc def}, we discuss the limitations of the approaches referred to above compared to our proposed ptFC.

\subsection{Human Connectome Project Data and Paper Organization}\label{section: Human Connectome Project Data and Paper Organization}


We apply our proposed framework to a study investigating the human motor system. We use a cohort of subjects from a task-evoked fMRI study publicly available at the Human Connectome Project (HCP). The \textit{block-design} motor task used in this study is adapted from experiments by \cite{buckner2011organization} and \cite{yeo2011organization}, while details on the HCP implementation are given by \cite{barch2013function}. During the experiment, the subjects are asked to perform five tasks when presented with a cue: tap left and right fingers, squeeze left and right toes, and move their tongue. BOLD signals $Y_k(\omega;t)$ are collected from 308 subjects, and each BOLD signal is obtained with TR $\Delta=0.72$ (seconds). Each experiment lasts for about $204$ seconds, i.e., $T=283$. We focus on the brain functional connectivity evoked by the task of \textit{squeezing right toes}. The onsets of these task blocks vary across subjects. Nevertheless, the corresponding onsets of any two subjects differ in less than 0.1 seconds ($\le\Delta$). Therefore, we assume that all subjects share the same stimulus signal 
\begin{align}\label{eq: stimulus signal for the HCP data}
    N(t)=\mathbf{1}_{[86.5, 98.5)}(t)+\mathbf{1}_{[162, 174)}(t).
\end{align} 

This paper is organized as follows. Section \ref{section: A New Definition of FC} proposes models for task-evoked neural activity and BOLD signals, respectively; based on these models, we propose a rigorous and interpretable definition of the ptFC; we particularly explain the relationship between our ptFC and the ``inter-region task-related functional connectivity" proposed in \cite{bowman2008bayesian}. In Section \ref{section: Periodic Extension and Distributional Assumptions}, we introduce the probabilistic assumptions for estimating the ptFC. Section \ref{section: Estimation} presents an algorithm for estimating ptFC. In Section \ref{section: simulations}, we compare the performance of our proposed algorithm with existing approaches using simulations. In Section \ref{section: applications}, we apply the proposed approach to estimate the ptFC during a motor task using the publicly available HCP data set. Section \ref{section: Conclusions and Further Discussions} concludes the paper.


\section{Population-level Task-evoked Funtional Connectivity (ptFC)}\label{section: A New Definition of FC}

We first propose the models for neural activity and BOLD signals at individual nodes. Using these models, we provide the definition of ptFC. Then, we explain the relationship between our proposed ptFC and the Bayesian hierarchical framework in \cite{bowman2008bayesian}. Lastly, we list the advantages of ptFC compared to existing approaches. 

Hereafter, $\Omega$ denotes the collection of subjects of interest. $\Omega$ is discrete and finite. Let $\mathbb{P}$ define a study-dependent probability measure on $\Omega$. Then, $(\Omega, 2^\Omega, \mathbb{P})$ is a probability space, where $2^\Omega$ is the power set of $\Omega$. Any function defined on $\Omega$ within the context of the probability space $(\Omega, 2^\Omega, \mathbb{P})$ is a random variable (see Section 1.5 of \cite{klenke2013probability} for the definition of random variables). $\mathbb{E}$ denotes the expectation with respect to $\mathbb{P}$. For example, in the HCP experiment described in Section \ref{section: Human Connectome Project Data and Paper Organization}, $\Omega$ represents the population of all healthy adults. FMRI data are collected in a sample of size $n$ drawn from $\Omega$ according to the underlying distribution $\mathbb{P}$. For example, the HCP data set described in Section \ref{section: Human Connectome Project Data and Paper Organization} has a sample size of $n=308$. We estimate $\mathbb{P}$ by the empirical distribution $\frac{1}{n}\sum_{\omega}\delta_\omega$. The expectation and correlation with respect to this empirical distribution are the sample average and sample correlation, respectively.

\subsection{Models for Signals}\label{section: Models for Signals}

We use $h_k(t)$ to denote the HRF at the $k^{th}$ node, which corresponds to the task $N(t)$ and is common to all subjects $\omega\in\Omega$. The HRF $h_k(t)$ is solely determined by the local vascular system at the $k^{th}$ node, irrelevant to the neural activation \citep[][Chapter 3]{ashby2019statistical}. For each $\omega$, we model the task-evoked terms $P_k(\omega;t)$ in Eq.~\eqref{eqn: BOLD signal decomposition} by 
\begin{align}\label{eq: convolution model}
P_k(\omega;t)= \beta_k(\omega)\times (N*h_k)\left(t-t_{0,k}\right),\ \ t\in\mathcal{T},\ \ \omega\in\Omega,\ \ k=1,2,\cdots, K. 
\end{align}
where $\beta_k(\omega)$ is a coefficient representing the subject-specific effect, $t_{0,k}$ is the latency shared by all subjects as the reaction of the $k^{th}$ node to $N(t)$ at time $t$ is not instantaneous, and $*$ denotes the convolution operation. 

Since $\beta_k:\omega\mapsto\beta_k(\omega)$ is a random variable defined on the probability space $(\Omega, 2^\Omega, \mathbb{P})$, the model in Eq.~\eqref{eq: convolution model} is a mixed linear model \citep[MLM,][]{mclean1991unified} with responses $\{P_k(\omega;t)\}_{t\in\mathcal{T}}$, independent variables $\{(N*h_k)(t-t_{0,k})\}_{t\in\mathcal{T}}$, and random effect $\beta_k(\omega)$; the error term of the MLM is absorbed by the $\{\epsilon_k(\omega;t)\}_{t\in\mathcal{T}}$ in Eq.~\eqref{eqn: BOLD signal decomposition}. In this paper, we employ the terms ``subject-specific effect" and ``random effect" interchangeably when referring to $\beta_k(\omega)$, based on the following considerations: (i) $\beta_k(\omega)$ is subject-specific due to its dependence on the subject $\omega$; (ii) $\beta_k:\omega\mapsto\beta_k(\omega)$ is a random variable as $\beta_k$ is a function of $\omega$, which is referred to by \cite{mclean1991unified} as a ``random effect" for the MLM presented in Eq.~\eqref{eq: convolution model} \citep[see Section 2 of][]{mclean1991unified}. It is important to note that we apply the MLM only to the task term $P_k(\omega;t)$ of the BOLD signal in Eq.~\eqref{eqn: BOLD signal decomposition}. The combination of the additive model in Eq.~\eqref{eqn: BOLD signal decomposition} and the MLM in Eq.~\eqref{eq: convolution model} is a ``semi-MLM." In Section \ref{section: Estimation}, we present the estimation strategy for the terms in this semi-MLM.

The task-evoked terms $ P_k(\omega;t)$ in Eq.~\eqref{eq: convolution model}
model the neural activity evoked by task $N(t)$, and the subject-specific $\beta_{k}(\omega)$ measures the magnitude of this component. Visualizations of Eq.~\eqref{eq: convolution model} are presented in Supplementary Figure 6. Theoretically, latency depends on subject $\omega$. However, in most studies, the latency is much shorter than the corresponding TR $\Delta$, and the difference between the latency times of any two subjects is negligible. Therefore, we model the latency as $t_{0,k}$, common to all subjects $\omega\in\Omega$. With the model in Eq.~\eqref{eq: convolution model}, we represent the model in Eq.~\eqref{eqn: BOLD signal decomposition} for BOLD signals as follows.
\begin{align}\label{eq: BOLD time-course model}
    Y_k(\omega;t)=\beta_k(\omega)\times \left(N*h_k\right)\left(t-t_{0,k}\right)+R_k(\omega;t),\ \ t\in\mathcal{T},\ \ k=1, \cdots, K,\ \ \omega\in\Omega,
\end{align}
where $R_k(\omega; t)=Q_k(\omega;t)+\epsilon_k(\omega;t)$ are referred to as \textit{reference terms} for succinctness.

Task-evoked terms $P_k(\omega;t)=\beta_k(\omega)\times (N*h_k)(t-t_{0,k})= \left[\{\beta_k(\omega)\times N(\cdot-t_{0,k})\}*h_k \right](t)$ correspond to neural activity evoked by $N(t)$, where $h_k$ represent metabolism and vasculature and do not characterize neural activity. Therefore, we model neural activity signals $\Phi_k[\omega; N(t)]$ responding to $N(t)$ as $\beta_k(\omega)\times N(t-t_{0,k})$. Using the model in Eq.~\eqref{eq: BOLD time-course model}, the ``stimulus-to-activity" map $\Phi_k[\omega; \bullet]$ and ``activity-to-BOLD" map $\Psi_k\{\omega; \bullet\}$ in Section \ref{section: Introduction} are
\begin{align}\label{eq: stimulus-BOLD maps}
\begin{aligned}
& \Phi_k\left[\omega;\, \bullet\right]: N(t)\mapsto \beta_k(\omega)\times N\left(t-t_{0,k}\right), \\
& \Psi_k\left\{\omega;\, \bullet \right\}: \Phi_k\left[\omega;\, N(t)\right]\mapsto \left\{\Phi_k[\omega;\, N(\cdot)]*h_k \right\} (t) +Q_k(\omega;\, t)+\epsilon_k(\omega;t)=Y_k(\omega;\, t).
\end{aligned}
\end{align}

\subsection{Definition of ptFC}\label{section: Definition of ptFC}

In this subsection, we provide the definition of the ptFC and its relationship with the ``inter-region task-related functional connectivity" framework proposed by \cite{bowman2008bayesian}.

FC is anticipated to characterize the mechanism of neural activity rather than metabolism or vasculature. Therefore, task-evoked FC should be defined using the task-evoked neural activity signals $\Phi_k[\omega; N(t)]$. In the model of $\Phi_k[\omega; N(t)]$ presented by Eq.~\eqref{eq: stimulus-BOLD maps}, tasks $N(t)$ are fully determined by experimental designs. Additionally, latency $t_{0,k}$ can be viewed as a parameter of HRF $h_k$ since task-evoked terms can be expressed as $\{[\beta_k(\omega)\times N]*[h_k(\cdot-t_{0,k})]\}(t)$. Recall that the $h_k$ only depends on metabolism and vasculature, and it does not characterize neural activity. Therefore, task-evoked FC is expected to be determined by $\beta_k(\omega)$. Since $\beta_k(\omega)$, for $k=1,\cdots,K$, measure the magnitude of the neural activity evoked by task $N(t)$, nodes $k$ and $l$ are functionally connected at the population level during the task if one of the following scenarios holds: (i) for subjects with strong reaction to $N(t)$ at their $k^{th}$ nodes, their $l^{th}$ nodes' reaction to $N(t)$ is strong as well and vice versa; (ii) for subjects with strong reaction to $N(t)$ at their $k^{th}$ nodes, their $l^{th}$ nodes' reaction to $N(t)$ is weak and vice versa, i.e., either positively or negatively correlated. Therefore, we make the following assumptions on the distribution of $\beta_k(\omega)$ across $(\Omega,\mathbb{P})$:
\begin{itemize}
    \item If there exists a functional connection evoked by $N(t)$ between nodes $k$ and $l$, the corresponding $\beta_k(\omega)$ and $\beta_l(\omega)$ are approximately linearly associated.
    \item Each variance $\mathbb{V}\beta_k=\mathbb{E}(\beta_k^2)-\left(\mathbb{E}\beta_k\right)^2>0$, i.e., the random variable $\beta_k(\omega)$ is not deterministic, where $\mathbb{E}\beta_k^\nu \overset{\operatorname{def}}{=} \int_{\Omega}\left\{\beta_k(\omega)\right\}^\nu \,\mathbb{P}(d\omega)$ for $\nu=1,2$.
\end{itemize}
Since the correlation  $corr\left(\beta_k, \beta_l\right) \overset{\operatorname{def}}{=}\frac{\mathbb{E}\left(\beta_k\beta_l\right)-\left(\mathbb{E}\beta_k\right)\left(\mathbb{E}\beta_l\right)}{\sqrt{\mathbb{V}\beta_k\times\mathbb{V}\beta_l}}$ across the population $\Omega$ measures the linear correlation between $\beta_k(\omega)$ and $\beta_l(\omega)$, we define ptFCs as follows:
\begin{definition}\label{def: group-level task-evoked FC and EC}
Suppose the task-evoked BOLD signals $\{Y_k(\omega; t)\}_{k=1}^K$ are of the form in Eq.~\eqref{eq: BOLD time-course model}. The population-level task-evoked functional connectivity (ptFC) between the $k^{th}$ and $l^{th}$ nodes is defined as $\vert corr(\beta_k,\beta_l)\vert$.
\end{definition}
The Pearson correlation approach in Eq.~\eqref{eq: raw Pearson correlation} defines the FC between two brain nodes through the correlation across the time index space $\mathcal{T}$. In contrast, the ptFC takes the form of a correlation defined across the population space $\Omega$. While the $\mathcal{T}$-correlation and $\Omega$-correlation forms differ mathematically, they model the same brain activity mechanism and, hence, these correlations are comparable.

An advantage of ptFC is its scale-invariance. Since $\beta_k(\omega) (N*h_k)=\frac{\beta_k(\omega)}{c_1\times c_2} [(c_2 N)*(c_2 h_k)]$, the scale of $\beta_k(\omega)$ changes if the scale of $h_k$ or $N(t)$ changes. But $\vert corr(\beta_k,\beta_l)\vert$ is invariant to the transform $\beta_{k'}(\omega)\mapsto c\beta_{k'}(\omega)$, for $k'\in\{k,l\}$ and any $c\ne0$. Furthermore, using this scale-invariance, we show in Appendix C that ptFC is identifiable under some probabilistic conditions. Additionally, $\vert corr(\beta_k,\beta_l)\vert$ is invariant to the transform $\beta_{k'}(\omega)\mapsto \beta_{k'}(\omega)+c$ for $k'\in\{k,l\}$ and any $c\in\mathbb{R}$, hence we may assume $\mathbb{E}\beta_{k'}=0$ for $k'\in\{k,l\}$ in Section \ref{section: Estimation}. In addition, we have the following interpretations:
\begin{itemize}
    \item \textbf{Interpretation of $\beta_k(\omega)$}: Each signal $\notag\Phi_k\left[\omega; N(t)\right] = \beta_k(\omega)\times N\left(t-t_{0,k}\right)$ describes the neural activity evoked by task $N(t)$, excluding effects of local vasculature or metabolism. If the $k^{th}$ node of $\omega$ does not react to task $N(t)$, then $\beta_k(\omega)=0$. For given $N(t)$ and $h_k$, the magnitude of $\vert\beta_k(\omega)\vert$ indicates the strength of the reaction in the $k^{th}$ node of $\omega$ to $N(t)$.
    
    \item \textbf{Interpretation of the ptFC $\vert corr(\beta_k, \beta_l)\vert$}: The pair $\{\beta_k(\omega), \beta_l(\omega)\}_{\omega\in\Omega}$ quantifies the magnitude of neural activity in response to $N(t)$ for nodes $k$ and $l$. These two nodes have functional connectivity evoked by $N(t)$ if $\{\beta_k(\omega)\}_{\omega\in\Omega}$ and $\{\beta_l(\omega)\}_{\omega\in\Omega}$ are linearly associated across $(\Omega,\mathbb{P})$. More explicitly, a perfectly linear relationship between $\{\beta_k(\omega)\}_{\omega\in\Omega}$ and $\{\beta_l(\omega)\}_{\omega\in\Omega}$ implies strongest functional connectivity between nodes $k$ and $l$ evoked by $N(t)$. Finally, $\vert corr(\beta_k, \beta_l)\vert$ quantifies the strength of $N(t)$-induced connectivity.
\end{itemize}

\cite{bowman2008bayesian} proposed a spatial Bayesian hierarchical model (SBHM) to characterize the task-evoked (referred to as ``task-related" therein) functional connectivity at a population level. Here, we explore the relationship between our proposed ptFC and the SBHM in \cite{bowman2008bayesian}. The model for BOLD signals in the SBHM is the same as our model in Eq.~\eqref{eq: BOLD time-course model}, except the SBHM assumes $t_{0,k}=0$ and $Q_k(\omega;t)=H(t)\cdot\eta_k(\omega)$ for all $k=1,2,\ldots,K$, where $H(t)$ contains the covariates that are not of interest. In addition, each $k$ in the SBHM corresponds to a region of interest instead of a single voxel. For this BOLD signal model, \cite{bowman2008bayesian} assumed the following parametric and hierarchical structures
\begin{align}\label{eq: hierarchical structure}
    \begin{aligned}
    & \beta_k(\omega)\,\vert\, \mu_k,\alpha_{k}(\omega), \sigma_{k}^2 \sim \operatorname{Normal}(\mu_k+\alpha_k(\omega), \sigma_{k}^2), \\
    & cov(\beta_k,\beta_l\,\vert\, \mu_k,\alpha_{k}(\omega), \sigma_{k}^2, \mu_l,\alpha_{l}(\omega), \sigma_{l}^2)=0,\ \ \ \text{ if }k\ne l,\\
    & \mu_k\,\vert\,\lambda_{k}^2\sim \operatorname{Normal}(\mu_{0,k},\lambda_k^2), \ \ \ \sigma_{k}^{-2}\sim \operatorname{Gamma}(a_0,b_0), \ \ \ \lambda_k^{-2}\sim \operatorname{Gamma}(c_0,d_0),\\
    & (\alpha_1(\omega),\alpha_2(\omega),\ldots, \alpha_K(\omega))^T \,\vert\,\boldsymbol{\Gamma} \sim  \operatorname{Normal}(\boldsymbol{0}, \boldsymbol{\Gamma}),\ \ \ \boldsymbol{\Gamma}^{-1}\sim \operatorname{Wishart}\{(h_0\boldsymbol{H}_0)^{-1}, h_0\},
    \end{aligned}
\end{align}
where $\boldsymbol{\Gamma}=(\gamma_{kl})_{1\le k,l\le K}$ is a $K$-by-$K$ positive definite matrix. Finally, functional connectivity between regions $k$ and $l$ was defined as the following ratio
\begin{align}\label{eq: ptFC in Bowman 2008}
    B_{k,l} \overset{\operatorname{def}}{=} \frac{\gamma_{kl}}{\sqrt{(\gamma_{kk}+\sigma^2_k)\cdot(\gamma_{ll}+\sigma^2_l)}},\ \ \ \text{ for }k\ne l.
\end{align}
Since the ratio in Eq.~\eqref{eq: ptFC in Bowman 2008} does not depend on subjects $\omega$, the inter-region functional connectivity $B_{k,l}$ is a population-level quantity. If we assume the hierarchical structure in Eq.~\eqref{eq: hierarchical structure}, while setting $\mu_k=0$ for all $k=1,\ldots,K$, then $\vert corr(\beta_k,\beta_l)\vert=\vert B_{k,l}\vert$ when $k\ne l$. That is, up to at most a negative sign, our proposed ptFC is equal to the inter-region functional connectivity in Eq.~\eqref{eq: ptFC in Bowman 2008} if we assume the hierarchical structures in Eq.~\eqref{eq: hierarchical structure}. However, neither does the definition nor the estimation of our proposed ptFC depend on the parametric and hierarchical structures in Eq.~\eqref{eq: hierarchical structure}. Hence, our proposed ptFC is a generalization of the framework proposed by \cite{bowman2008bayesian}. 

Recently, there is an emerging consensus in the literature suggesting that brain networks undergo temporal fluctuations corresponding to experimental tasks. \textit{Dynamic connectivity} (DC) is a collection of approaches investigating these fluctuations. Although ptFC in Definition \ref{def: group-level task-evoked FC and EC} does not change over time, the ptFC framework can be directly incorporated to obtain \textit{sliding-window} based DC estimates \citep{hutchison2013dynamic}. Specifically, for a preselected window of time, BOLD signals are extracted for that time interval, and ptFC is estimated using the algorithm proposed in this paper. Next, the window is shifted in time for a certain number of time points, and ptFC is estimated for BOLD signals within the time interval of the same length as in the first step for the shifted time interval. As a result, we obtain a sequence of ptFCs illustrating the dynamic structure of FC.

\subsection{Advantages of the ptFC Compared to Existing Approaches}\label{section: advantages of the task-fc def}

In this subsection, we present the advantages of our proposed ptFC compared to some existing approaches. We first discuss the limitations of the Pearson correlation approach in Eq.~\eqref{eq: raw Pearson correlation} using the model in Eq.~\eqref{eq: convolution model}. Plugging Eq.~\eqref{eq: convolution model} into Eq.~\eqref{eq: raw Pearson correlation}, we obtain the following association between $P_k(\omega;t)$ and $P_l(\omega;t)$.
\begin{align}\label{eq: Pearson correlation with the convolution model}
    \left\vert corr(\omega;P_k, P_l)\right\vert=\left\vert \, \int_{\mathcal{T}} \phi_k(t) \times \phi_l(t)\mu(dt)\Bigg/\sqrt{ \int_{\mathcal{T}}\left\vert \phi_k(t)\right\vert^2 \mu(dt)\times \int_{\mathcal{T}}\left\vert \phi_l(t)\right\vert^2 \mu(dt)} \, \right\vert,
\end{align}
where $\phi_{k'}(t)=N*h_{k'}(t-t_{0,k'})-\frac{1}{\mu(\mathcal{T})}\int_{\mathcal{T}}N*h_{k'}(s-t_{0,k'})\mu(ds)$ for $k'\in\{k,l\}$. Eq.~\eqref{eq: Pearson correlation with the convolution model} reveals the following limitations of the Pearson correlation approach.
\begin{itemize}
    \item \noindent\textbf{Only nuisance parameters:} Based on the reasoning in Section \ref{section: Definition of ptFC}, task-evoked neural activity is modeled by $\beta_k(\omega)$. Hence, it is counterintuitive that the quantity in Eq.~\eqref{eq: Pearson correlation with the convolution model} does not depend on $\beta_k(\omega)$ and depends only on the nuisance parameters $N(t)$ and $h_k(t)$ when considering neural activity. For example, if the $k^{th}$ node does not react to $N(t)$, then the task-evoked term $\beta_k(\omega)\times (N*h_k)(t-t_{0,k})$ is expected to be zero, i.e., $\beta_k(\omega)=0$, and there should be no task-evoked interaction between the $k^{th}$ and other nodes. However, (\ref{eq: Pearson correlation with the convolution model}) can still be very large as the non-reaction information presented by $\beta_k(\omega)= 0$ vanishes. In contrast with Eq.~\eqref{eq: Pearson correlation with the convolution model}, Definition \ref{def: group-level task-evoked FC and EC} is based on $\{(\beta_k(\omega), \beta_l(\omega))\vert\omega\in\Omega\}$.
    
    \item \noindent\textbf{Variation in HRFs:} HRFs can heavily vary across brain nodes \citep{miezin2000characterizing}. Since task-evoked FC is not expected to depend on HRFs, variation of HRFs in different brain regions should not influence task-evoked FC. However, the measurement in Eq.~\eqref{eq: Pearson correlation with the convolution model} can be small if $h_k\ne h_l$ as illustrated in Supplementary Figure 6 (c). Given that the correlation form $ corr(\beta_k,\beta_l)$ does not depend on HRFs $h_k$ and $h_l$, ptFCs $\vert corr(\beta_k,\beta_l) \vert$ are invariant to the variation in HRFs across brain nodes.  
    
    \item \noindent\textbf{Variation in latency:} $t_{0,k}$ may vary across nodes, e.g., see \cite{miezin2000characterizing} for an investigation of the left and right visual and motor cortices and the corresponding difference between response onsets. If $t_{0,k}\ne t_{0,l}$, for example $t_{0,k}< t_{0,l}$, then it is possible that node $k$ reacts to task $N(t)$ first, and then the neural activity at node $k$ causes that at node $l$. Because of this potential causality represented by $t_{0,k}< t_{0,l}$, it is natural to expect that nodes $k$ and $l$ are likely functionally connected. However, the measurement in Eq.~\eqref{eq: Pearson correlation with the convolution model} can be very small if $t_{0,k} \ne t_{0,l}$ and may not reveal the true interaction between the neural activity of two nodes. An example of this issue is illustrated in Appendix F. Since our proposed ptFCs $\vert corr(\beta_k,\beta_l)\vert$ do not involve the latency $t_{0,k}$ and $t_{0,l}$, the variation in latency does not influence our proposed ptFCs.
\end{itemize}

Coherence analysis for FC \citep{muller2001multivariate} is not influenced by any of the issues discussed above. In the coherence analysis approach, FC between BOLD signals $Y_k(\omega;t)$ and $Y_l(\omega;t)$ is measured by the \textit{coherence} evaluating the extent of the \textit{linear time-invariant relationship} between these two signals via Fourier frequencies. However, there is no guarantee that two BOLD signals are linear time-invariant if the corresponding two nodes are functionally connected. Last but not least, the Pearson correlation and coherence analysis approaches are designed to measure FC evoked by all stimuli --- both the tasks of interest and nuisance stimuli. Hence, they are not interpretable from the task-evoked FC viewpoint. On the contrary, beta-series regression \citep{rissman2004measuring} and ptFCs are designed to measure the FC  evoked by an experiment's specific task of interest. Additionally, our simulation studies in Section \ref{section: simulations} (see Table \ref{table: identification rate}) show that beta-series regression performs worse than our proposed ptFC approach in many cases.

\section{Periodic Extension and Distributional Assumptions}\label{section: Periodic Extension and Distributional Assumptions}

In order to develop the estimator of the ptFC defined in Section \ref{section: Definition of ptFC}, it is essential to make certain assumptions regarding the terms in Eq.~\eqref{eqn: BOLD signal decomposition}. In this section, we discuss general assumptions and potential verification strategies for these assumptions. Since BOLD signals are observed at discrete and finite time points, we assume throughout this paper $\mathcal{T}=\{\tau\Delta\}_{\tau=0}^T$ and $t=\tau\Delta$ for $\tau\in\{0,1,\cdots,T\}$. For any function $f:\mathcal{T} \rightarrow\R$, we extend $f$ as follows to be a periodic function on $\{\tau'\Delta\}_{\tau'\in\mathbb{Z}}$.
\begin{align}\label{eq: periodic extension}
    f(\tau'\Delta)=f(\tau\Delta),\ \ \ \tau'\equiv\tau\mbox{ (mod $T+1$)}\ \ \mbox{for }\tau=0,1,\cdots,T\mbox{ and all }\tau'\in\mathbb{Z}.
\end{align}
Hereafter, all functions on $\mathcal{T}$ are implicitly extended using Eq.~\eqref{eq: periodic extension} to be periodic functions on $\{\tau'\Delta \}_{\tau' \in \mathbb{Z}}$.

The identifiability of the task-evoked terms (see Appendix C) and our approach to estimating ptFC (see Section \ref{section: Estimation}) depend on two assumptions on the distribution of BOLD signals --- ``independence" and ``weak stationarity with zero mean." We discuss them in Sections \ref{section: Independence} and \ref{section: Weak Stationarity with Zero Mean}, respectively, and show the following: for commonly used implementations of our model, the ``independence" assumption can be checked, and the ``weak stationarity with zero mean" assumption is satisfied. 
    
\subsection{Independence}\label{section: Independence}

We assume the following assumption holds.
\begin{assumption}
The terms $\{P_k(\omega;t)\}_{t\in\mathcal{T}}$, $\{Q_k(\omega;t)\}_{t\in\mathcal{T}}$, and $\{\epsilon_k(\omega;t)\}_{t\in\mathcal{T}}$ in Eq.~\eqref{eqn: BOLD signal decomposition} are independent.
\end{assumption}
The assumption of independence between the two model terms and the random error $\{\epsilon_k(\omega;t)\}_{t\in\mathcal{T}}$ is based on the fact that $\{\epsilon_k(\omega;t)\}_{t\in\mathcal{T}}$ models the random additive noise due to measurement error and other random occurrences unrelated to the brain functional organization during the task and at rest. This assumption is common in many statistical models including functional principal component analysis \citep{di2009multilevel}, the hierarchical functional connectivity modeling \citep{bowman2008bayesian}, and others.  

Generally, testing for independence between the task term and rest terms is difficult due to the fact that most fMRI studies are either resting-state or task fMRI and we rarely obtain both resting-state and task fMRI data at a single experiment. However, for commonly used implementations of our model, the independence between $\{P_k(\omega;t)\}_{t\in\mathcal{T}}$ and $\{Q_k(\omega;t)\}_{t\in\mathcal{T}}$ can be tested using data. For example, suppose the term $Q(\omega;t)$ in Eq.~\eqref{eqn: BOLD signal decomposition} is of the following form
\begin{align}\label{eq: an example of Qk}
    Q_k(\omega;t)=\left\{\sum_{\gamma=1}^\Gamma \beta_{k,\gamma}(\omega)\times \tilde{N}_{\gamma}*\tilde{h}_{k,\gamma}(t)\right\},
\end{align}
where each $\tilde{N}_{\gamma}$ denotes either a stimulus that is not of interest or the stimulus from the spontaneous brain activity, and $\tilde{h}_{k,\gamma}$ denotes the HRF at the $k^{th}$ node corresponding to the task $\tilde{N}_{\gamma}$. The combination of Eq.~\eqref{eqn: BOLD signal decomposition} and the $Q(\omega;t)$ defined in Eq.~\eqref{eq: an example of Qk} is compatible with the model of BOLD signals implemented in \cite{bowman2008bayesian} (see Eq. (1) therein). In addition, this combination is the following MLM
\begin{align*}
    Y_k(\omega;t)=\beta_k(\omega)\times \left(N*h_k\right)\left(t-t_{0,k}\right)+\left\{\sum_{\gamma=1}^\Gamma \beta_{k,\gamma}(\omega)\times \tilde{N}_{\gamma}*\tilde{h}_{k,\gamma}(t)\right\} + \epsilon_k(\omega;t),\ \ t\in\mathcal{T},\ \ k=1, \cdots, K,\ \ \omega\in\Omega,
\end{align*}
which was also discussed in Section 8.4 of \cite{ashby2019statistical}. 
One may apply the estimation approach for MLM to estimate the random coefficients $\beta_k(\omega)$ and $\{\beta_{k,\gamma}(\omega)\}_{\gamma=1}^\Gamma$ for all $\omega\in\Omega$ \citep[e.g.,][]{mclean1991unified} and test the independence between $\beta_k(\omega)$ and $\{\beta_{k,\gamma}(\omega)\}_{\gamma=1}^\Gamma$. The independence between $\beta_k(\omega)$ and $\{\beta_{k,\gamma}(\omega)\}_{\gamma=1}^\Gamma$ indicates the independence between $P_k(\omega;t)=\beta_k(\omega)\times \left(N*h_k\right)\left(t-t_{0,k}\right)$ and $Q_k(\omega;t)$.

\subsection{Weak Stationarity with Zero Mean}\label{section: Weak Stationarity with Zero Mean}

We introduce the following concept: a stochastic process $\pmb{G}(\omega;t)=(G_1(\omega;t), \cdots, G_K(\omega;t))^T$ is called \textit{weakly stationary with mean zero} (WSMZ) if $\mathbb{E}G_k(t)=0$, for all $t\in\mathcal{T}$, and $\mathbb{E}\{G_k(t)G_l(t+s)\}$ depends only on $s$, rather than $t$, for all $k,l\in\{1,\cdots,K\}$. Suppose we artificially generate an auxiliary random variable $U:\Omega\rightarrow\mathcal{T}$ such that $U$ is uniformly distributed on $\mathcal{T}$. Since it is artificially generated, it is independent of all other random variables referred to in our model. Hereafter, we will be under the following assumption on the reference terms $\{R_k(\omega;t)\}_{t\in\mathcal{T}}$ (see Eq.~\eqref{eq: BOLD time-course model})
\begin{assumption}\label{assumption: WSMZ}
The stochastic process $\{R_k(\omega;t-U(\omega))\}_{t\in\mathcal{T}}$ is weakly stationary with mean zero.
\end{assumption}
Assumption \ref{assumption: WSMZ} does not require $\{R_k(\omega;t)\}_{t\in\mathcal{T}}$ itself to be WSMZ, and it is satisfied in most commonly used models. For example, Theorem B.3 in Appendix B shows that reference terms $\{R_k(\omega;t)=Q_k(\omega;t)+\epsilon_k(\omega;t)\}_{t\in\mathcal{T}}$ satisfy Assumption \ref{assumption: WSMZ} if $Q_k(\omega;t)$ are defined by Eq.~\eqref{eq: an example of Qk} and $\{\epsilon_k(\omega;t)\}_{t\in\mathcal{T}}$ are white noise.


\section{ptFC Estimation (ptFCE) Algorithm}\label{section: Estimation}

To estimate $\vert corr(\beta_k, \beta_l)\vert$, one needs the task-evoked terms $P_k(\omega;t)$. However, these terms are not observed in most applications. In Section \ref{section: The Estimation of Reference Signals}, we will propose an estimator (denoted as $\tilde{R}_k(\omega;t)$) for $R_k(\omega;t)$. As a result, we can obtain estimators for task-evoked terms as $\tilde{P}_k(\omega;t) \overset{\operatorname{def}}{=}Y_k(\omega;t)-\tilde{R}_k(\omega;t)$. Suppose $h_k$ and $t_{0,k}$ are known \citep[otherwise, they can be estimated using semi-parametric methods, e.g., see][]{zhang2013semi}, then a straightforward approach to estimate $\vert corr(\beta_k, \beta_l)\vert$ is the MLM-based approach as follows: we implement the following MLM
\begin{align}\label{eq: MLM-based approach}
    \tilde{P}_k(\omega;\tau\Delta)= (N*h_k)(\tau\Delta-t_{0,k}) \cdot \tilde{\beta}(\omega)+\tilde{\epsilon}(\omega; \tau\Delta),\ \ \ \tau=0,1,\ldots,T,
\end{align}
where $\{(N*h_k)(\tau\Delta-t_{0,k})\}_{\tau=0}^T$ is the design matrix (viewed as a $(T+1)$-by-1 matrix) relating $\{\tilde{P}_k(\omega;\tau\Delta)\}_{\tau=0}^T$ to the random effect $\tilde{\beta}(\omega)$, and $\{\tilde{\epsilon}_k(\omega;\tau\Delta)=\tilde{P}_k(\omega;\tau\Delta)-P_k(\omega;\tau\Delta)\}_{\tau=0}^T$ are unknown random errors. Then, $\vert corr(\beta_k, \beta_l)\vert$ is estimated by the absolute value of the sample correlation between $\tilde{\beta}_k(\omega)$ and $\tilde{\beta}_l(\omega)$ across all subjects $\omega$. While this is a simple approach for the estimation of ptFC, using simulation studies, we show that the MLM-based approach is not robust to the violation of assumptions in the model in Eq.~\eqref{eq: BOLD time-course model} and leads to worse performance compared to our proposed method. The comparisons between the MLM-based approach and our proposed method are presented in Section \ref{section: simulations}.

In this section, we propose the ptFCE algorithm based on the Fourier transform and the AMUSE algorithm \citep{tong1991indeterminacy}. First, we define the convolution as follows
\begin{align*}
    \left(N*h_k\right)\left(\tau\Delta\right) \overset{\operatorname{def}}{=} \frac{1}{T+1}\sum_{\tau'=0}^T N(\tau'\Delta)h_k\left((\tau-\tau')\Delta\right),\ \ \mbox{ for }\tau=0,1,\cdots,T.
\end{align*}
Then, we define the Fourier transform of $f$ as
\begin{align}\label{eq: def of the Fourier transform}
    \widehat{f}(\xi) \overset{\operatorname{def}}{=} \frac{1}{T+1}\sum_{\tau=0}^{T} f(\tau\Delta) e^{-2\pi i \xi (\tau\Delta)},\ \ \ \text{ for }\xi\in\mathbb{R}.
\end{align}
where $i=\sqrt{-1}$ is the imaginary unit. The $\widehat{f}(\xi)$ is a periodic function of $\xi$ with period $1/\Delta$. Throughout this paper, $\widehat{(\cdot)}$ denotes the Fourier transform. Additionally, we assume $\mathbb{E}\beta_k=\mathbb{E}R_k(t)=0$, for all $k$ and $t$, motivated by the centralization 
\begin{align*}
    Y_k(\omega;t)-\mathbb{E}Y_k(t)=\left\{\beta_k(\omega)-\mathbb{E}\beta_k\right\}(N*h_k)(t-t_{0,k}) + \left\{R_k(\omega;t)-\mathbb{E}R_k(t)\right\}.
\end{align*}
Using $\left\{\beta_k(\omega)-\mathbb{E}\beta_k\right\}$ and $\left\{R_k(\omega;t)-\mathbb{E}R_k(t)\right\}$ as $\beta_k(\omega)$ and $R_k(\omega;t)$, respectively, we model the demeaned signals $Y_k(\omega;t)-\mathbb{E}Y_k(t)$. Because of the invariance $corr(\beta_k, \beta_l)=corr((\beta_k-\mathbb{E}\beta_k), (\beta_l-\mathbb{E}\beta_l))$, the assumption $\mathbb{E}\beta_k=0$ does not prevent the detection of ptFC $\vert corr(\beta_k, \beta_l) \vert$.

Our proposed ptFC depends on neither latency $t_{0,k}$ nor reference terms $R_k(\omega;t)$. To remove the dependence of the signals on the latency $t_{0,k}$, we artificially generate an auxiliary random variable $U:\Omega\rightarrow\mathcal{T}$, which is uniformly distributed on $\mathcal{T}=\{\tau\Delta\}_{\tau=0}^T$. Specifically, applying periodic extension in Eq.~\eqref{eq: periodic extension} to $(N*h_k)$, one may verify that the two stochastic processes $(N*h_k)(t-t_{0,k}-U(\omega))$ and $(N*h_k)(t-U(\omega))$ are identically distributed. Then the distribution of $(N*h_k)(t-U(\omega))$, and hence the distribution of $(N*h_k)(t-t_{0,k}-U(\omega))$, does not depend on $t_{0,k}$. Therefore, we investigate the following time-shifted signals 
\begin{align*}
    Y_k\left(\omega;t-U(\omega)\right)=\beta_k(\omega) (N*h_k)\left(t-t_{0,k}-U(\omega)\right)+R_k\left(\omega;t-U(\omega)\right),\ \ \ \text{ for }k=1,\cdots,K.
\end{align*}
Based on the discussion above, the distributions of $Y_k\left(\omega;t-U(\omega)\right)$ do not depend on $t_{0,k}$. Then, we consider the autocovariance $\mathbb{E}\left\{ Y_k(t-U) Y_l(t+s-U) \right\}$ when estimating $\vert corr(\beta_k,\beta_l)\vert$, where $s=\underline{s}\Delta$ with $\underline{s}\in\mathbb{Z}$. Additionally, $\vert corr(\beta_k,\beta_l)\vert$ depends only on task-evoked terms $\{Y_{k}(\omega;t-U(\omega))-R_k(\omega;t-U(\omega))\}$ (see Eq.~\eqref{eq: BOLD time-course model}). Hence, we investigate the autocovariance difference $\mathbb{E}\left\{ Y_k(t-U) Y_l(t+s-U) \right\} -\mathbb{E}\left\{R_k(t-U)R_l(t+s-U)\right\}$. If $U(\omega)$, $\{\beta_k(\omega)\}_{k=1}^K$, and $\{R_k(\omega;t)\vert t\in\mathcal{T}\}_{k=1}^K$ are independent, Supplementary Theorem B.1 implies that this difference depends only on $s$. Therefore, we denote it as follows
\begin{align}\label{eq: expected difference}
  \mathcal{A}_{kl}(s) =  \mathbb{E}\left\{ Y_k(t-U) Y_l(t+s-U) \right\} -\mathbb{E}\left\{ R_k(t-U)R_l(t+s-U)\right\}, \ \ k,l=1,\cdots,K. 
\end{align}
The Fourier transform of $\mathcal{A}_{kl}(s)$ is denoted as $\widehat{\mathcal{A}_{kl}}(\xi)$ (see Eq.~\eqref{eq: def of the Fourier transform}). We propose the following function of $\xi$
\begin{align}\label{eq: def of C_kl}
    \mathcal{C}_{kl}(\xi) \overset{\operatorname{def}}{=}
    \frac{\left\vert \widehat{\mathcal{A}_{kl}}(\xi)\right\vert}{\sqrt{  \left\vert \, \widehat{\mathcal{A}_{kk}}(\xi) \,\,\widehat{\mathcal{A}_{ll}}(\xi)  \,\right\vert  }},\ \ \ \text{ for }\xi\in\mathbb{R},
\end{align}
by incorporating a normalization multiplier in Eq.~\eqref{eq: expected difference}. The property of Fourier transform indicates that $\mathcal{C}_{kl}(\xi)$ is a periodic function of $\xi$ with period $1/\Delta$. Without loss of generality, we view $\mathcal{C}_{kl}(\xi)$ as a function defined on the compact interval $[0,1/\Delta]$. Supplementary Lemma B.1 implies that $\widehat{\mathcal{A}_{kk}}(\xi) \widehat{\mathcal{A}_{ll}}(\xi)\ne 0$ except for finitely many $\xi$ in $[0,1/\Delta]$. Hence, the denominator in Eq. \eqref{eq: def of C_kl} is strictly positive except for finitely many $\xi$. We may ignore these finitely many points in the continuous interval $[0,1/\Delta]$. Supplementary Theorem B.1 provides the derivation of the Fourier transform
\begin{align*}
    \widehat{\mathcal{A}_{kl}}(\xi)=\mathbb{E}\left( \beta_k \beta_l\right) \cdot \vert \widehat{N}(\xi) \vert^2\cdot  \overline{\widehat{h}_k(\xi)}\, \widehat{h}_l(\xi) \cdot e^{2\pi i (t_{0,k}-t_{0,l})\xi},
\end{align*}
which implies the following representation of ptFC by canceling our the nuisance information contained in factors $\widehat{N}(\xi), h_k, h_l, t_{0,k}$, and $t_{0,l}$
\begin{align}\label{Eq: substitute equation}
    \mathcal{C}_{kl}(\xi)= \frac{\left\vert \widehat{\mathcal{A}_{kl}}(\xi)\right\vert}{\sqrt{  \left\vert \, \widehat{\mathcal{A}_{kk}}(\xi) \,\,\widehat{\mathcal{A}_{ll}}(\xi)  \,\right\vert  }} =\left\vert corr(\beta_k, \beta_l)\right\vert,
\end{align}
except for the finitely many $\xi\in[0,1/\Delta]$ such that $\widehat{\mathcal{A}_{kk}}(\xi) \widehat{\mathcal{A}_{ll}}(\xi) = 0$. That is, the function $\mathcal{C}_{kl}(\xi)$ of $\xi$ defined in Eq.~\eqref{eq: def of C_kl} is a constant function, and the constant is our proposed ptFC $\left\vert corr(\beta_k, \beta_l)\right\vert$. The representation in Eq.~\eqref{Eq: substitute equation} motivates the ptFCE algorithm, i.e., instead of directly estimating $\vert corr(\beta_k, \beta_l)\vert$, we propose an approach to estimating $\mathcal{C}_{kl}(\xi)$ via the Fourier transform. The main steps of our ptFCE algorithm provide  an estimation of factors $\mathcal{A}_{kl}(s)$.

BOLD signals $Y_k(\omega;t)$ are observed in experiments, and $U(\omega)$ is an auxiliary variable artificially generated in our estimation procedure (step 1 of Algorithm \ref{algorithm: ptFCE algorithm}). Hence, the term $\mathbb{E}\left\{ Y_k(t-U) Y_l(t+s-U) \right\}$ in Eq.~\eqref{eq: expected difference} can be estimated using the method of moments. However, the reference signals $R_k(\omega;t)$ implicitly contained in $Y_k(\omega;t)=P_k(\omega;t)+R_k(\omega;t)$ are not observable. Section \ref{section: The Estimation of Reference Signals} provides the estimation of $R_k(\omega;t)$.

\subsection{The Estimation of Reference Signals}\label{section: The Estimation of Reference Signals}

In this subsection, we derive an approximation $\tilde{R}_k(\omega;t)\approx R_k(\omega;t)$ from observed $Y_k(\omega;t)$. First, we note that $C_k \overset{\operatorname{def}}{=}\mathbb{E}\{(N*h_k)(t-t_{0,k}-U)\}$ is a constant depending on neither $t$ nor $t_{0,k}$. Define stochastic processes 
\begin{align*}
    J_k(\omega;t) := \beta_k(\omega) \{(N*h_k)(t-t_{0,k}-U(\omega))-C_k\},\ \ \ \text{ for }k=1,\cdots,K.
\end{align*}
For each fixed $k$, one can verify that scalar-valued stochastic process $J_k(\omega;t)$ is WSMZ conditioning on $\beta_k$, then $\mathbb{E}\left\{J_k(t) J_k(t+s) \,\big\vert\, \beta_k\right\}$ depends only on $s$. The following theorem gives the foundation for the estimation of reference terms $R_k(\omega;t)$.

\begin{theorem}\label{thm: 1st thm for AMUSE-ptFC}
For each $k\in\{1,\cdots,K\}$, suppose the BOLD signals $Y_k(\omega;t)$ are of the form in Eq.~\eqref{eq: BOLD time-course model}, and the random variable $U:\Omega\rightarrow\mathcal{T}$ is uniformly distributed. Additionally, for each $k$, suppose we are under the following conditions
\begin{itemize}
    \item the $\beta_k(\omega)$, $\{R_k(\omega;t)\}_{t\in\mathcal{T}}$, and $U(\omega)$ are independent;
    
    \item the scalar-valued stochastic process $\{R_k(\omega;t-U(\omega))\}_{t\in\mathcal{T}}$ is WSMZ;
    
    \item there exist $t^*\in\mathcal{T}$ and $\beta^*\in\mathbb{R}$ such that $ \frac{\mathbb{E}\{J_k(t) J_k(t-t^*)\vert\beta_k=\beta^*\}}{\mathbb{E}\left\{J_k(t)^2\vert\beta_k=\beta^*\right\}} \ne
    \frac{\mathbb{E}\{R_k(t-U) R_k(t-t^*-U)\}}{\mathbb{E}\left\{R_k(t-U)^2\right\}}$.
\end{itemize}
If nonsingular matrix $\pmb{A}=(a_{ij})_{1\le i, j \le 2}$ and stochastic process $\{\pmb{s}(\omega;t)=(s_1(\omega;t), s_2(\omega;t))^T\}_{t\in\mathcal{T}}$ satisfy the following 
    \begin{enumerate}[label=\roman*]
        \item given $\beta_k(\omega)=\beta^*$, $\pmb{s}(\omega;t)$ is WSMZ,
        \item $s_1(\omega;t_1)$ and $s_2(\omega;t_2)$ are uncorrelated for $t_1,t_2\in\mathcal{T}$,
        \item $\frac{\mathbb{E}\{s_1(t)s_1(t-t^*)\vert\beta_k=\beta^*\}}{\mathbb{E}\{s_1(t)\vert\beta_k=\beta^*\}} \ne \frac{\mathbb{E}\{s_2(t)s_2(t-t^*)\vert\beta_k=\beta^*\}}{\mathbb{E}\{s_2(t)\vert\beta_k=\beta^*\}}$, and \item $\pmb{A}(s_1(\omega;t), s_2(\omega;t))^T=(Y_k(\omega;t-U(\omega)) - \beta^* C_k, (N*h_k)(t-t_{0,k}-U(\omega))-C_k)^T$ for $t\in\mathcal{T}$,
    \end{enumerate}
    then there exist a non-singular diagonal matrix $\pmb{\Lambda}$ and a permutation matrix $\pmb{P}$, such that 
\begin{align}\label{AMUSE thm formula 2}
    \begin{pmatrix}
    1 & 1 \\
    \frac{1}{\beta^*} & 0
\end{pmatrix} = \pmb{A} \pmb{\Lambda}^{-1} \pmb{P}^{-1} \mbox{ and }
\begin{pmatrix}
J_k(\omega;t) \\
R_k(\omega;t-U(\omega))
\end{pmatrix} = \pmb{P}\pmb{\Lambda}\pmb{s}(\omega;t), \mbox{ for all }t\in\mathcal{T}.
\end{align}
\end{theorem}
\noindent Theorem \ref{thm: 1st thm for AMUSE-ptFC} is a straightforward result of Theorem 2 in \cite{tong1991indeterminacy}. For each fixed $k$ and $\omega$, the pair $(\pmb{A}, \pmb{s}(\omega;t))$ in Theorem \ref{thm: 1st thm for AMUSE-ptFC} is derived by the AMUSE algorithm with the following 2D signal as its input
\begin{align*}
    \left\{\Big( Y_k(\omega;t-U(\omega)) - \beta_k(\omega)\cdot C_k,\,\, (N*h_k)(t-t_{0,k}-U(\omega))-C_k\Big)^T \,\Big\vert\, t\in\mathcal{T}\right\}.
\end{align*}
However, coefficient $\beta_k(\omega)$ is unknown. Since $\mathbb{E}\beta_k=0$, we ignore the term $\beta_k(\omega)\cdot C_k$ and apply the following vector-valued signal as input of the AMUSE algorithm.
\begin{align}\label{eq: AMUSE input}
    \left\{\Big( Y_k(\omega;t-U(\omega)), \,\, (N*h_k)(t-t_{0,k}-U(\omega))-C_k\Big)^T \,\Big\vert\, t\in\mathcal{T}\right\}.
\end{align}

The choice of latency $t_{0,k}$ and HRFs $h_k$ in Eq.~\eqref{eq: AMUSE input} is of importance. $t_{0,k}$ and $h_k$ can be estimated using semi-parametric approaches, e.g., the spline-based method by \cite{zhang2013semi}. In the meantime, our proposed estimator of ptFC defined by the ptFCE algorithm includes the HRFs and the latency variables reflecting the relationship between ptFC and the observed BOLD signals. Hence, our ptFCE algorithm allows the inclusion of information on HRFs and latency terms when computing the estimator. Suppose the estimators of $t_{0,k}$ and $h_k$ are denoted by $t_{0,k}^{est}$ and $h_k^{est}$, respectively. Then, we replace the unknown $t_{0,k}$ and $h_k$ in the input in Eq.~\eqref{eq: AMUSE input} with the estimated $t_{0,k}^{est}$ and $h_k^{est}$. To simplify our implementation of Eq.~\eqref{eq: AMUSE input} when analyzing the HCP data, we choose $t_{0,k}=0$ and $h_k=$ the canonical HRF (the \texttt{R} function \texttt{canonicalHRF} with default parameters in package \texttt{neuRosim}), instead of using the estimated $t_{0,k}^{est}$ and $h_k^{est}$ via some semi-parametric approaches, based on the following considerations: 1) latency times are usually much smaller than their corresponding experimental time units $\Delta$ \citep[e.g.,][p 138]{zhang2013semi} the HCP data are collected using a block design; it is known that $t_{0,k}$ are negligible if latency $t_{0,k}$ are much shorter than the length of task blocks; 3) latency $t_{0,k}$ can be incorporated into the corresponding HRF $h_k$ as its parameter, i.e., we view $h_k(\cdot-t_{0,k})$ as an HRF, then the choice of $t_{0,k}$ is equivalent to the choice of corresponding HRF (see the discussions below); 4) in block-design experiments, the influence of biased HRFs on the ptFCE algorithm is moderate as we illustrate using simulations in Section \ref{section: simulations}; this holds even if the support of biased HRFs is longer than task blocks.

Since the distribution of $U(\omega)$ is known, constants $C_k = \mathbb{E}\{(N*h_k)(t-t_{0,k}-U)\}$ can be estimated. To approximately recover $R_k(\omega;t)$ from $(\pmb{A}, \pmb{s}(\omega;t))$, we show the following result.
\begin{theorem}\label{thm: 2nd thm for AMUSE-ptFC}
For each $k\in\{1,\cdots,K\}$ and $\omega\in\Omega$, suppose the pair $(\pmb{A}, \pmb{s}(\omega;t))$ satisfies Eq.~\eqref{AMUSE thm formula 2}. Then there exists $i'\in\{1,2\}$ such that $a_{1i'}\, s_{i'}(\omega;t)=J_k(\omega;t)$.
\end{theorem}
\noindent The proof of Theorem \ref{thm: 2nd thm for AMUSE-ptFC} is in Appendix B. The index $i'$ in Theorem \ref{thm: 2nd thm for AMUSE-ptFC} can be computed as 
\begin{align*}
    i'=\argmax_i\left\{corr(s_i(\omega;t), (N*h_k)(t-t_{0,k}-U(\omega_j))-C_k), \mbox{ across }t \,\big\vert\, i=1,2\right\}.
\end{align*}
The AMUSE algorithm and Theorem \ref{thm: 2nd thm for AMUSE-ptFC} recover $J_k(\omega;t)=\beta_k(\omega)(N*h_k)(t-t_{0,k}-U(\omega))-\beta_k(\omega)C_k$. Again, we ignore $\beta_k(\omega)C_k$ because $\mathbb{E}\beta_k=0$, that is, $J_k(\omega;t)\approx \beta_k(\omega)(N*h_k)(\omega;t-t_{0,k}-U(\omega))$. Then we estimate $R_k(\omega;t-U(\omega))$ by $\tilde{R}_k(\omega;t-U(\omega))\overset{\operatorname{def}}{=} Y_k(\omega;t-U(\omega))-J_k(\omega;t)\approx R_k(\omega;t-U(\omega))$, for all $k$ and $\omega$. In applications, $U(\omega)$ are artificially generated and known. Then, we have the approximation $\tilde{R}_k(\omega;t)\approx R_k(\omega;t)$. We conclude the derivation of $\tilde{R}_k(\omega;t)$ from $Y_k(\omega;t)$ as follows: \\
Step 1, each observed signal $Y_k(\omega;t)$ provides the input in Eq.~\eqref{eq: AMUSE input} for the AMUSE algorithm; \\
Step 2, the AMUSE algorithm computes $(\pmb{A}, \pmb{s}(\omega;t))$; \\
Step 3, Theorem \ref{thm: 2nd thm for AMUSE-ptFC} indicates $J_k(\omega;t)=a_{1i'}s_{i'}(\omega;t)$; \\
Step 4, compute $\tilde{R}_k(\omega;t)$ using $J_k(\omega;t)$ and the artificially generated $U(\omega)$. 

Theorem \ref{thm: 1st thm for AMUSE-ptFC} is based on the weak stationarity of $\{R_k(\omega;t-U(\omega))\}_{t\in\mathcal{T}}$, and the AMUSE algorithm is based on Theorem \ref{thm: 1st thm for AMUSE-ptFC}. If the weak stationarity assumption is violated, any source of non-weak stationarity in $\{R_k(\omega;t-U(\omega))\}_{t\in\mathcal{T}}$ is incorporated in the estimation of the task-evoked terms.

\subsection{The Estimator for $\mathcal{C}_{kl}(\xi)$ and The ptFCE Algorithm}\label{section: estimator and the ptFCE algorithm}

With the observed signals $\{Y_{k'}(\omega;t)\vert k'=k,l\}_{\omega=1}^n$ and the signals $\{\tilde{R}_{k'}(\omega;t)\vert k'=k,l\}_{\omega=1}^n$ derived from $\{Y_{k'}(\omega;t)\vert k'=k,l\}_{\omega=1}^n$ using the approach introduced in Section \ref{section: The Estimation of Reference Signals}, we propose the following estimator for $\mathcal{C}_{kl}(\xi)=\vert corr(\beta_k,\beta_l)\vert$.
\begin{align}\label{eq: main formula for data}
\begin{aligned}
& \mathcal{C}_{kl}^{est,n}(\xi) \overset{\operatorname{def}}{=} \frac{ \left\vert \widehat{\mathcal{A}_{kl}^{est,n}}(\xi)\right\vert }{\sqrt{  \left\vert \widehat{\mathcal{A}_{kk}^{est,n}}(\xi) \,\, \widehat{\mathcal{A}_{ll}^{est,n}}(\xi)  \right\vert  }},\ \ \ \mbox{for all } \xi\in\mathbb{R}, \mbox{ where } \\
& \mathcal{A}_{kl}^{est,n}(s) \overset{\operatorname{def}}{=} \underline{Y}(s)-\underline{\tilde{R}}(s), \\
& \underline{Y}(s) \overset{\operatorname{def}}{=} \frac{1}{T+1}\sum_{t\in\{\tau\Delta\}_{\tau=0}^T}\left[  \frac{1}{n}\sum_{\omega=1}^n  \left\{Y_k\left(\omega;t-U(\omega)\right)Y_l(\omega;t+s-U(\omega))\right\} \right], \\
& \underline{\tilde{R}}(s) \overset{\operatorname{def}}{=}  \frac{1}{T+1}\sum_{t\in\{\tau\Delta\}_{\tau=0}^T}\left[  \frac{1}{n}\sum_{\omega=1}^n  \{\tilde{R}_k(\omega;t-U(\omega)) \tilde{R}_l(\omega;t+s-U(\omega))\} \right].
\end{aligned}
\end{align}
The periodicity property of the Fourier transform implies that $\mathcal{C}_{kl}^{est,n}(\xi)$ is a periodic function of $\xi$ with period $1/\Delta$. Supplementary Lemma B.1 implies that the estimator $\mathcal{C}_{kl}^{est,n}(\xi)$ in Eq.~\eqref{eq: main formula for data} is well-defined for all $\xi\in[0,1/\Delta]$ except for at most finitely many points $\xi$ such that $\widehat{\mathcal{A}_{kk}^{*est}}(\xi) \widehat{\mathcal{A}_{ll}^{*est}}(\xi)=0$.

$\mathcal{A}_{kl}^{est,n}(s)$ has a method of moments form except $R_k$ in (\ref{eq: expected difference}) is replaced with the estimated $\tilde{R}_k$. This has an effect on the consistency of $\mathcal{C}_{kl}^{est,n}(\xi)$ as an estimator for $\mathcal{C}_{kl}(\xi)$ and results in the asymptotic bias $\vert\lim_{n\rightarrow\infty}\mathcal{C}_{kl}^{est,n}(\xi)-\mathcal{C}_{kl}(\xi)\vert$. The bias can be small for properly chosen $\xi$. Specifically, we derive a formula of the bias as a function of $\xi$ and choose $\xi$ that minimizes the bias. Since $\tilde{R}_k(\omega;t)$ approximates $R_k(\omega;t)$, we model the difference $W_k(\omega;t) \overset{\operatorname{def}}{=} R_k(\omega;t)-\tilde{R}_k(\omega;t)$ as random noise and assume that $\pmb{W}(\omega;t) \overset{\operatorname{def}}{=} (W_1(\omega;t), \cdots,W_K(\omega;t))^T$ satisfies the following
\begin{enumerate}[label=\roman*]
    \item $\pmb{W}(\omega;t_1)$ is independent of $\pmb{W}(\omega;t_2)$ if $t_1\ne t_2$,
    \item $\pmb{W}(\omega;t)$ is WSMZ, and
    \item $\Sigma_{kl}:=\mathbb{E}[W_k(t)W_l(t)]$ for $t\in\mathcal{T}$.
\end{enumerate}
Because of $\Sigma_{kl}\le\sqrt{\Sigma_{kk}\Sigma_{ll}}$ and $\mathbb{E}\beta_k=\mathbb{E}\beta_l=0$, under the assumptions on $\pmb{W}(\omega;t)$, Supplementary Theorem B.2 implies 
\begin{align*}
    \lim_{n\rightarrow\infty}\mathcal{C}^{est,n}_{kl}(\xi) &= \frac{ \left\vert \mathbb{E}(\beta_k\beta_l)\frac{\overline{\widehat{h}_k(\xi)}\widehat{h}_l(\xi)}{\left\vert \widehat{h}_k(\xi) \widehat{h}_l(\xi) \right\vert}  e^{2\pi i \xi (t_{0,k}-t_{0,l})} + \frac{\Sigma_{kl}}{(T+1)\left\vert\widehat{N}(\xi)\right\vert^2\left\vert\widehat{h}_k(\xi) \widehat{h}_l(\xi)\right\vert} \right\vert }{ \sqrt{ \left(\mathbb{E}(\beta_k^2)+\frac{\Sigma_{kk}}{(T+1)\left\vert\widehat{N}(\xi)\widehat{h}_k(\xi)\right\vert^2}\right)\cdot \left(\mathbb{E}(\beta_l^2)+\frac{\Sigma_{ll}}{(T+1)\left\vert\widehat{N}(\xi)\widehat{h}_l(\xi)\right\vert^2}\right) } } \\
    & \approx
    \frac{ \left\vert \mathbb{E}(\beta_k\beta_l)\frac{\overline{\widehat{h}_k(\xi)}\widehat{h}_l(\xi)}{\left\vert \widehat{h}_k(\xi) \widehat{h}_l(\xi) \right\vert}  e^{2\pi i \xi (t_{0,k}-t_{0,l})} \right\vert }{ \sqrt{ \mathbb{E}(\beta_k^2) \cdot\mathbb{E}(\beta_l^2) } }  = \frac{\vert\mathbb{E}(\beta_k \beta_l)\vert}{ \sqrt{ \mathbb{E}(\beta_k^2) \mathbb{E}(\beta_l^2) } } = \vert corr(\beta_k \beta_l)\vert
\end{align*}
almost surely at frequencies $\xi$ such that
\begin{align}\label{eq: zero contamination}
    \Sigma_{kk}\Big/[(T+1)\vert\widehat{N}(\xi)\widehat{h}_k(\xi)\vert^2]\approx 0,\ \ \ \Sigma_{ll}\Big/[(T+1)\vert\widehat{N}(\xi)\widehat{h}_l(\xi)\vert^2]\approx 0.
\end{align}
The $\xi$ resulting in sufficiently large $\vert\widehat{h}_k(\xi)\vert$ and $\vert\widehat{h}_l(\xi)\vert$ satisfying (\ref{eq: zero contamination}) reduces the estimation bias. From the signal processing perspective, large $\vert \widehat{h}_k(\xi)\vert$ and $\vert \widehat{h}_l(\xi)\vert$ filter out the random noise in $W_k(\omega;t)$, and HRFs act as band-pass filters \citep[typically $0-0.15$ Hz, see][]{aguirre1997empirical}. Hence, we are interested in $\xi\in(0,0.15)$. Therefore, motivated by (\ref{eq: zero contamination}), $\vert corr(\beta_k,\beta_l)\vert$ is approximated by the median of $\mathcal{C}_{kl}^{est,n}(\xi)$ across $\xi\in(0,0.15)$, since the median is more stable than mean. One typical curve of $\{\mathcal{C}_{kl}^{est,n}(\xi)\vert\xi\in\mathbb{R}\}$ is in Figure \ref{fig: curve of C as a function of xi}. The estimator $\mathcal{C}_{kl}^{est,n}(\xi)$ in (\ref{eq: main formula for data}) and the choice of $\xi$ complete the estimation of $\mathcal{C}_{kl}(\xi)=\vert corr(\beta_k \beta_l)\vert$. The estimation procedure is concluded in the ptFCE algorithm (Algorithm \ref{algorithm: ptFCE algorithm}). 

\begin{algorithm}[h]
	\caption{ptFCE Algorithm}\label{algorithm: ptFCE algorithm}
	\begin{algorithmic}[1]
		\INPUT
        \noindent(i) Task-fMRI BOLD signals  $\{(Y_k(\omega;\tau\Delta), Y_l(\omega;\tau\Delta))\}_{\tau=0}^T$ for sampled subjects $\omega\in\{1,\cdots,n\}$; (ii) repetition time $\Delta$; (iii) the stimulus signal $\{N(\tau\Delta)\}_{\tau=0}^T$; (iv) latency $t_{0,k'}=\tau_{0,k'}\Delta$ with some $\tau_{0,k'}\in\mathbb{Z}$ for $k'\in\{k,l\}$, and their default values are $t_{0,k'}=0$ for $k'\in\{k,l\}$; (v) HRF $h_{k'}$ for $k'\in\{k,l\}$. The default HRFs are $h_k=h_l=$ the \texttt{R} function \texttt{canonicalHRF} with its default parameters.

		\OUTPUT An estimation of the ptFC $\vert corr(\beta_k, \beta_l)\vert$ between the $k^{th}$ and $l^{th}$ nodes.
		\STATE Generate i.i.d. $\{U(\omega)\}_{\omega=1}^n$ from the uniform distribution on $\mathcal{T}=\{\tau\Delta\}_{\tau=0}^T$.
		\STATE For each $\omega\in\{1,\cdots,n\}$ and $k'\in\{k,l\}$, apply the AMUSE algorithm to input (\ref{eq: AMUSE input}) and obtain estimated reference signals $\{\tilde{R}_k(\omega;\tau\Delta)\}_{\tau=0}^T$.
		\STATE Compute the estimator $\mathcal{C}_{kl}^{est,n}(\xi)$ in (\ref{eq: main formula for data}) and the median of $\{\mathcal{C}_{kl}^{est,n}(\xi)\vert\xi\in(0, 0.15)\}$. The median is the output of this algorithm. 
		\end{algorithmic}
\end{algorithm}

\begin{figure}[ht]
    \centering
    \includegraphics[scale=0.6]{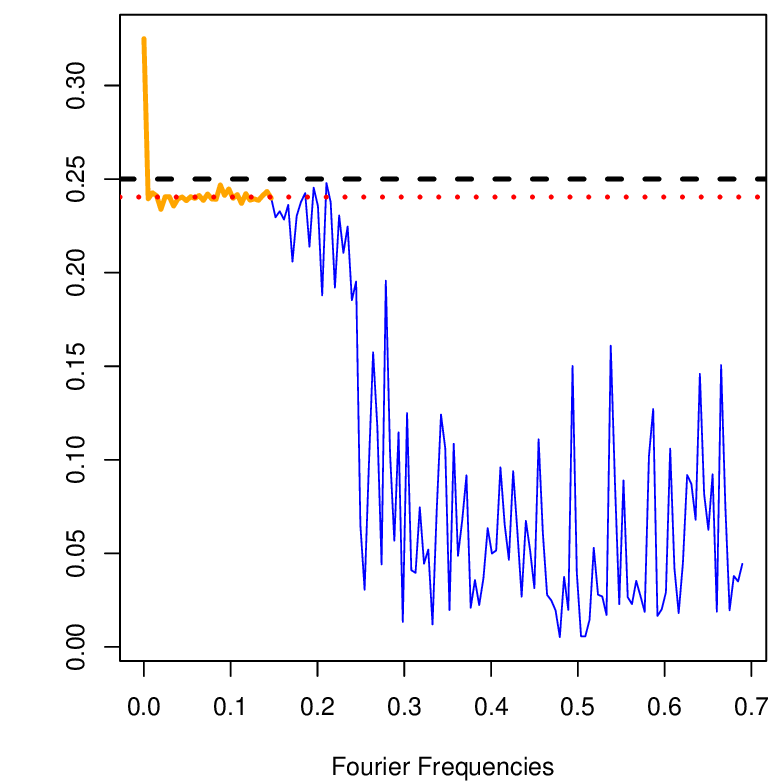}
    \caption{To generate this figure, we first generate synthetic BOLD signals using Mechanism 0 in Appendix E with sample size $n=308$ and underlying $\rho=0.25$ (presented by the black dashed line). Then we apply our proposed ptFCE algorithm to the synthetic BOLD signals generated by Mechanism 0. The solid blue curve presents the estimator $\mathcal{C}_{kl}^{est,n}(\xi)$  as a function of Fourier frequencies $\xi\in(0, \frac{1}{2\Delta})$, and the solid orange curve presents $\{\mathcal{C}_{kl}^{est,n}(\xi)\vert\xi\in(0,0.15)\}$ for taking the desired median. The unreasonably large value part of this orange curve motivates us to implement the median instead of the mean of $\{\mathcal{C}_{kl}^{est,n}(\xi)\vert\xi\in(0,0.15)\}$. The dotted red line presents the median of $\mathcal{C}_{kl}^{est,n}(\xi)$ across $\xi\in(0, 0.15)$, i.e., the output of the ptFCE algorithm for the synthetic BOLD signals.}
    \label{fig: curve of C as a function of xi}
\end{figure}


\section{Simulations}\label{section: simulations}

When comparing methods for FC estimation, direct comparisons of methods are difficult as the estimates are defined on different scales \citep{cisler2014comparison}. In the context of FC analysis, we may compare the methods by measuring their accuracy in correctly identifying the relative values of FC among pairs of nodes estimated by each approach, i.e., the performance of the algorithm in terms of detecting underlying ``weak vs. strong" patterns instead of the estimated values. Specifically, suppose we investigate two pairs of nodes. The connectivity between one pair is weak while the connectivity between another pair is strong; we are interested in whether the estimated value for the weak connectivity is smaller than that of the strong connectivity. Therefore, in this section, we only compare ptFCE with other approaches in detecting underlying ``weak vs. strong" patterns.  

We compare the ptFCE algorithm with the following methods: beta-series regression, \textit{naive Pearson correlation}, \textit{task Pearson correlation}, and coherence analysis. For the comparative analysis, we use synthetic data generated by two different mechanisms designed to mimic the properties of HCP motor data --- the first is based on the model in Eq.~\eqref{eq: BOLD time-course model} for defining ptFCs, while the second is motivated by Pearson correlations. We use the second data-generating mechanism to illustrate the comparable performance of our proposed approach with others even when the data-generating mechanism is not compatible with the ptFC framework. We design the simulations using non-canonical HRFs in generating synthetic data to illustrate the robustness of our algorithm that uses the canonical HRF as input. The two data-generating mechanisms are briefly summarized as follows. Details of them are in Appendix E.

\noindent\textbf{Mechanism 1}: \textbf{Step 1}, generate coefficients $\left(\beta_{1}(\omega), \beta_{2}(\omega), \beta_{3}(\omega)\right)^T\sim N_3\left(\pmb{0}, (\sigma_{ij})_{1\le i,j \le 3}\right)$ for $\omega=1,\cdots,n$. Assume $\rho_{ij}:=\vert\sigma_{ij}/\sqrt{\sigma_{ii}\sigma_{jj}}\vert$, for $(i,j)\in\{(1,2), (2,3)\}$ are the true underlying ptFCs between nodes $1$ and $2$ and between nodes $2$ and $3$, respectively. \textbf{Step 2}, generate reference signals $\{R_{k'}(\omega;\tau\Delta)=Q_{k'}(\omega;\tau\Delta)+\epsilon_{k'}(\omega;\tau\Delta)\}_{\tau=0}^T$, for all $\omega$ and $k'\in\{1,2,3\}$, where the non-zero dependent signals $Q_{k'}(\omega;\tau\Delta)$ are constructed using Eq.~\eqref{eq: an example of Qk} (details are available in Appendix E), and the $3nT$ noise values $\{\epsilon_{k'}(\omega;\tau\Delta) \vert k'=1,2,3; \omega=1,\cdots,n; \tau=0,\cdots,T\}\sim_{iid} N(0,V)$. \textbf{Step 3}, compute synthetic BOLD signals $\{Y_{k'}(\omega;\tau\Delta)\vert k'=1,2,3\}_{\tau=0}^T$, for $\omega=1,\cdots,n$, by 
\begin{align*}
    Y_{k'}(\omega; \tau\Delta)=9000+\beta_{k'}(\omega)\times \left(N * h_{k'}\right)(\tau\Delta) + R_{k'}(\omega; \tau\Delta).
\end{align*}
Denote $\pmb{Y}(\omega;t):=\{Y_{k'}(\omega;\tau\Delta)\}_{k'=1}^3$.

\noindent\textbf{Mechanism 2}: \textbf{Step 1}, for each $\omega\in\{1, \cdots, n\}$, generate independent normal vectors $$\{\pmb{\varepsilon}(\omega;\tau\Delta)=(\varepsilon_1(\omega;\tau\Delta), \varepsilon_2(\omega;\tau\Delta), \varepsilon_3(\omega;\tau\Delta))^T\}_{\tau=0}^T,$$ where $\pmb{\varepsilon}(\omega;\tau\Delta)\sim N_3(\pmb{0}, \pmb{\Sigma}_{task})$ if $N(\tau\Delta)=1$, and $\pmb{\varepsilon}(\omega;\tau\Delta)\sim N_3(\pmb{0}, \pmb{\Sigma}_{resting})$ if $N(\tau\Delta)=0$; matrices $\pmb{\Sigma}_{task}$ and $\pmb{\Sigma}_{resting}$ share the same diagonals, the off-diagonal elements of $\pmb{\Sigma}_{resting}$ are $0$, and $\{\varrho_{ij}\}_{i,j=1}^3$ denote the correlations deduced from covariance matrix $\pmb{\Sigma}_{task}$. \textbf{Step 2}, compute $Y_{k'}(\omega;\tau\Delta)=N*h_{k'}(\tau\Delta)+\sum_{j=1}^4 N_{j}*h_{k',j}(\tau\Delta)+\varepsilon_{k'}(\tau\Delta)$, for $k'\in\{1,2,3\}$, where the tasks $\{N_j(t)\}_{j=1}^4$ are nuisance tasks, and $\{h_{k',j}(t)\}_{j=1}^4$ are the corresponding HRFs. Details of the covariance matrices $\pmb{\Sigma}_{task}$ and $\pmb{\Sigma}_{resting}$ and nuisance tasks are in Appendix E. \textbf{Step 3}, repeat Steps 1 and 2 for all $\omega=1,\cdots,n$. 

In Mechanism 2, the zero off-diagonal elements of $\pmb{\Sigma}_{resting}$ indicate no task-evoked connectivity when the task of interest is absent. When $N(\tau\Delta)=1$, the correlation $\varrho_{ij}$ in $\pmb{\Sigma}_{task}$ measures the connectivity between nodes $i$ and $j$ evoked by task $N(\tau\Delta)=1$. We are interested in estimating the connectivity between nodes $1,2$ and nodes $2,3$. With the synthetic signals $\{\pmb{Y}(\omega;t)\}_{\omega=1}^n$ generated by Mechanisms 1 or 2, the existing methods are implemented as follows.

\noindent\textbf{Naive Pearson correlation:} For each subject $\omega$ and underlying $\rho_{ij}$ or $\varrho_{ij}$, compute the Pearson correlation between $\{Y_i(\omega;\tau\Delta)\}_{\tau=0}^T$ and $ \{Y_j(\omega;\tau\Delta)\}_{\tau=0}^T$ across all $\tau$ and denote the absolute value of this correlation by $\hat{\rho}^{naiveCorr}_{ij,\omega}$. Let $\hat{\rho}^{naiveCorr}_{ij, mean}$ and $\hat{\rho}^{naiveCorr}_{ij, median}$ denote the mean and median, respectively, of $\{\hat{\rho}^{naiveCorr}_{ij, \omega}\}_{\omega=1}^{n}$ across all $\omega$. 

\noindent\textbf{Task Pearson correlation:} For each $\omega$ and underlying $\rho_{ij}$ or $\varrho_{ij}$, compute the Pearson correlation between $\{Y_i(\omega;\tau\Delta)\vert N(\tau\Delta)=1\}$ and $ \{Y_j(\omega;\tau\Delta)\vert N(\tau\Delta)=1\}$ across $\tau$ such that $N(\tau\Delta)=1$ and denote its absolute value by $\hat{\rho}^{taskCorr}_{ij, \omega}$. Compute the mean and median of $\{\hat{\rho}^{taskCorr}_{ij, \omega}\}_{\omega=1}^{n}$ across $\omega$ and denote them by $\hat{\rho}^{taskCorr}_{ij, mean}$ and $\hat{\rho}^{taskCorr}_{ij, median}$, respectively.

\noindent Details of implementing the beta-series regression and coherence analysis to obtain estimates $(\hat{\rho}^{betaS}_{ij,mean}, \hat{\rho}^{betaS}_{ij,median})$ and $(\hat{\rho}^{Coh}_{ij, mean}, \hat{\rho}^{Coh}_{ij, median})$ are presented in Appendix G.

Let $(\rho_{12}, \rho_{23})=(\varrho_{12}, \varrho_{23})=(0.4, 0.6)$, indicating  weak ($0.4$) or strong ($0.6$) connectivity between the two node pairs. We use sample sizes $n=50$ and $308$ to illustrate the performance in different sample sizes, where 308 is the size of the HCP data set. PtFCs estimated by the ptFCE algorithm are denoted by $(\widehat{\rho}_{12}, \widehat{\rho}_{23})$. For each simulated data set, applying all the referred FC estimation methods, we obtain estimates
\begin{align*}
    \left\{ \left( \, \widehat{\rho}_{ij}, \, \hat{\rho}^{naiveCorr}_{ij, mean}, \, \hat{\rho}^{naiveCorr}_{ij, median}, \, \hat{\rho}^{taskCorr}_{ij, mean}, \, \hat{\rho}^{taskCorr}_{ij, median}, \, \hat{\rho}^{betaS}_{ij, mean}, \, \hat{\rho}^{betaS}_{ij,median}, \, \hat{\rho}^{Coh}_{ij, mean}, \, \hat{\rho}^{Coh}_{ij, median} \,\right) \,\Big\vert\, i<j \right\}_{i,j=1}^3.
\end{align*}
For each method, we evaluate whether it can identify the ``weak vs. strong" pattern. For example, our proposed ptFCE algorithm is effective if $\widehat{\rho}_{12} < \widehat{\rho}_{23}$. We repeat this procedure 500 times. The rates of correct identification of the connectivity patterns across 500 simulations are presented in Table \ref{table: identification rate}. 

\begin{table}[h]
\centering
\caption{Identification rates $\widehat{p}:=m'/m$ of different methods for the ``weak vs strong" pattern, where $m'$ is the numer of correct identifications among all $m=500$ synthetic samples. The correct identification rates in the table are presented in the $95\%$-confidence Wald-interval form \citep{brown2002confidence}, i.e., $\widehat{p}\pm z_{0.975}\cdot\sqrt{\widehat{p}(1-\widehat{p})/m}$, where $z_{0.975}\approx 1.96$ is the $0.975$-quantile of the standard normal distribution $N(0,1)$.}\label{table: identification rate}
\begin{tabular}{lllllllll}
\hline
                                    &  & Rates       &  & ($n=50$)    &  & Rates       &  & ($n=308$)   \\ \cline{3-5} \cline{7-9} 
Methods                             &  & Mech. 1 &  & Mech. 2 &  & Mech. 1 &  & Mech. 2 \\ \hline
ptFCE                               &  & $82.8(\pm 3.3)\%$    &  & $67.4(\pm 4.1)\%$    &  & $99.8(\pm 0.4)\%$    &  & $88.4(\pm 2.8)\%$    \\
MLM-based approach &  &  $69.2(\pm 4.0)\%$ &  & $58.8(\pm 4.3)\%$ &  & $84.6(\pm 3.2)\%$ &  &  $66.6(\pm 4.1)\%$ \\
naive Pearson (mean)   &  & $58.2(\pm 4.3)\%$    &  & $93.8(\pm 2.1)\%$    &  & $69.0(\pm 4.1)\%$    &  & $100(\pm 0)\%$     \\
naive Pearson (median) &  & $55.6(\pm 4.4)\%$    &  & $88.0(\pm 2.8)\%$    &  & $62.2(\pm 4.3)\%$    &  & $100(\pm 0)\%$     \\
task Pearson (mean)    &  & $47.6(\pm 4.4)\%$    &  & $100(\pm 0)\%$     &  & $48.6(\pm 4.4)\%$    &  & $100(\pm 0)\%$     \\
task Pearson (median)  &  & $47.8(\pm 4.4)\%$    &  & $100(\pm 0)\%$     &  & $50.6(\pm 4.4)\%$    &  & $100(\pm 0)\%$     \\
beta-series (mean)       &  & $49.4(\pm 4.4)\%$    &  & $96.6(\pm 1.6)\%$    &  & $48.8(\pm 4.4)\%$    &  & $100(\pm 0)\%$     \\
beta-series (median)     &  & $48.2(\pm 4.4)\%$    &  & $97.8(\pm 1.3)\%$    &  & $47.6(\pm 4.4)\%$    &  & $100(\pm 0)\%$     \\
coherence (mean)           &  & $51.0(\pm 4.4)\%$    &  & $55.8(\pm 4.4)\%$    &  & $60.8(\pm 4.3)\%$    &  & $72.0(\pm 3.9)\%$    \\
coherence (median)         &  & $51.6(\pm 4.4)\%$    &  & $57.6(\pm 4.3)\%$    &  & $58.8(\pm 4.3)\%$    &  & $68.0(\pm 4.1)\%$    \\ \hline
\end{tabular}
\end{table}

When synthetic BOLD signals are from the mechanism compatible with the definition of ptFCs --- Mechanism 1, our ptFCE algorithm performs overwhelmingly better than other methods. If the synthetic signals are from the mechanism motivated by Pearson correlations --- Mechanism 2, our proposed algorithm still provides a good identification rate and is better than the coherence analysis. Importantly, it is unlikely that the true data-generating process in task-fMRI follows the second data-generating process. Mechanism 2 implies the effect of task disappears as soon as the task is absent. However, it is known that there is a delay in brain response associated with time needed by brain vasculature to respond to the decrease in oxygen \citep{lindquist2008statistical}. Even in this unrealistic scenario, our proposed method is comparable with others. Naive Pearson correlation and coherence analysis are designed to measure FC evoked by both the task of interest and nuisance tasks. In contrast, beta-series regression, task Pearson correlations, and ptFCE are developed for quantifying the FC evoked by specific tasks of interest. Hence, only  beta-series regression, task Pearson correlations, and ptFCE algorithms are useful when estimating the FC evoked by specific tasks of interest. Additionally, the task Pearson correlation model assumes that the effect of task of interest disappears immediately when the task is absent. Hence, because of the delay in brain response, task Pearson correlations may be biased. Lastly, Table \ref{table: identification rate} shows that ptFCE performs better than the beta-series regression for data from Mechanism 1. 

In addition, we compare the MLM-based approach described in Section \ref{section: Estimation} with our proposed ptFCE algorithm. The implementation of the MLM-based approach is as follows.

\textbf{MLM-based approach:} for each subject $\omega$, we apply the AMUSE algorithm described in Section \ref{section: The Estimation of Reference Signals} and obtain estimators $\{\tilde{P}_k(\omega;\tau\Delta)\}_{\tau=0}^T$ of the true task evoked terms $\{P_k(\omega;\tau\Delta)\}_{\tau=0}^T$; we repeat this procedure for all subjects $\omega\in\Omega$. Then, we implement the MLM in Eq.~\eqref{eq: MLM-based approach} and obtain estimators $\{(\tilde{\beta}_k(\omega), \tilde{\beta}_l(\omega))\}_{\omega\in\Omega}$ of the true random-effects $\{(\beta_k(\omega), \beta_l(\omega))\}_{\omega\in\Omega}$. Lastly, we estimate the ptFC $\vert corr(\beta_k,\beta_l)\vert$ by the sample correlation between $\tilde{\beta}_k(\omega)$ and $\tilde{\beta}_l(\omega)$ across $\omega\in\Omega$.

We apply the MLM-based approach to the 500 simulations generated from Mechanisms 1 and 2, respectively, and compare the approach to other methods. The performance of the MLM-based approach and the corresponding comparison are presented in Table \ref{table: identification rate}. The results in Table \ref{table: identification rate} show that when the synthetic data are generated from Mechanism 1, although the MLM-based approach performs overwhelmingly better than the existing methods, it is inferior to our ptFCE algorithm. When the synthetic data are generated from Mechanism 2, which violates the assumptions in model (\ref{eq: BOLD time-course model}), the MLM-based approach is hardly better than a random guess.

In the simulation studies above (presented in Table \ref{table: identification rate}), we considered only 3 synthetic nodes. To further illustrate our proposed ptFCE algorithm, we apply it to 50 synthetic nodes. The simulation results for the 50-node study are presented in Figure \ref{fig: 50_node_study_plot}. The heatmap in the left panel shows the true connectivity structure evoked by the task of interest in Mechanisms 1 and 2; the heatmap in the middle panel shows the connectivity structure estimated by the ptFCE algorithm when the synthetic data are generated from Mechanism 1; the heatmap in the right panel shows the connectivity structure estimated by the ptPCE algorithm when the synthetic data are generated from Mechanism 2. The details of the 50-node study are provided in Appendix E.4.The similarity between the true connectivity and each of the estimated connectivity maps shows the capacity of our proposed ptFCE approach in estimating the true underlying task-evoked functional connectivity. 

\begin{figure}[ht]
    \centering
    \includegraphics[scale=0.46]{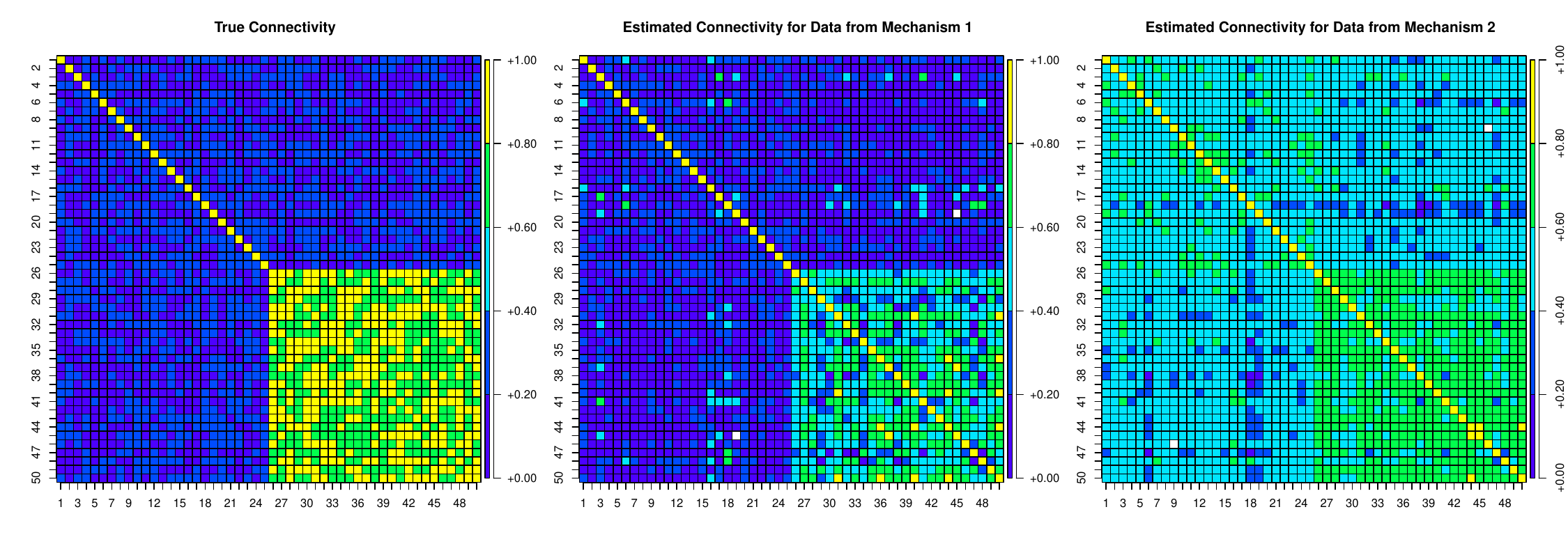}
    \caption{The simulation results of the 50-node study. The left panel shows the true connectivity structure for generating synthetic data using Mechanisms 1 and 2. The details of the data-generating procedure are provided in Appendix E.4. The yellow/green color in the bottom-right quadrant indicates that the connection between any two of nodes 26-50 is strong. For each of Mechanisms 1 and 2, we randomly generate only one synthetic subject. Then, we apply our proposed ptFCE algorithm to estimate the underlying connectivity. The estimated connectivity structures for data generated from Mechanisms 1 and 2 are presented in the middle and right panels, respectively.}
    \label{fig: 50_node_study_plot}
\end{figure}

Simulation studies showing the bias and variance of the ptFCE algorithm in estimating ptFC are provided in Appendix D. In addition, Appendix D also provides the simulation studies for different signal-to-noise ratios. PtFCE is computationally efficient and scalable. On a PC with \texttt{2.4 GHz 8-Core Intel Core i9} processor, the approximate computational time of ptFCE is 3.5 seconds for 50 subjects and 11.5 seconds for 1000 subjects.


\begin{figure}[h]
    \centering
    \includegraphics[scale=0.45]{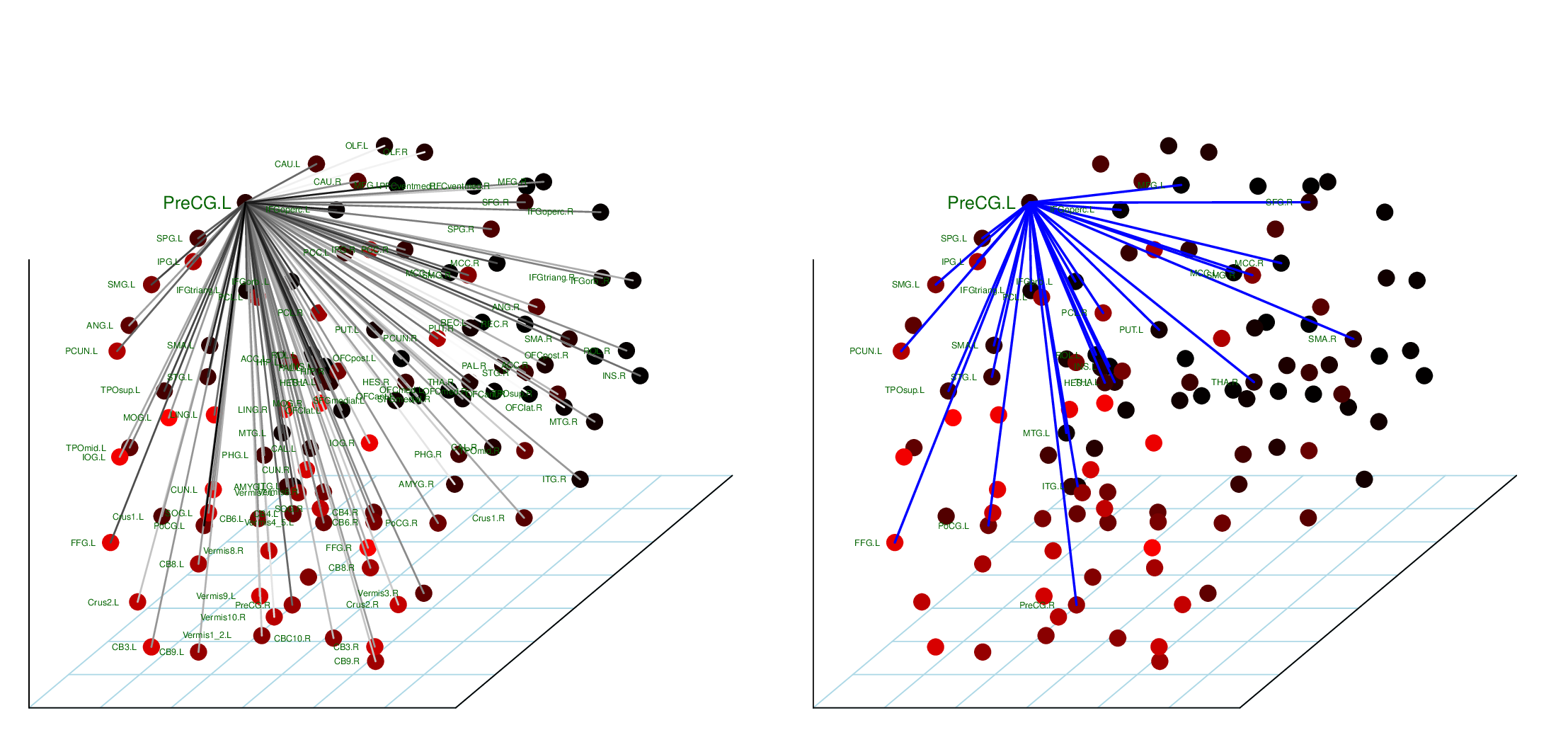}
    \caption{Each dot denotes a region from the AAL atlas, located using its corresponding MNI coordinates. The abbreviated region names are given next to each dot. We apply the ptFCE algorithm to estimate the ptFCs between region \texttt{PreCG.L} and each of the rest 116 regions. In the left panel, we use grayscale coloring of the edges to indicate the magnitude of ptFC between the corresponding two vertices; specifically, the larger a ptFC, the darker the line segment connecting the corresponding region pair. In the right panel, the presented blue line segments indicate the 30 largest ptFCs estimated by the ptFCE algorithm among all 116 regions.}
    \label{fig: 6plots12}
\end{figure}

\section{Analysis of HCP Motor Data}\label{section: applications}

In this section, we present an estimation of FC in a task-evoked functional MRI study using data from HCP. For comparison, we apply the ptFCE algorithm and existing methods to measure FC in the database of 308 subjects from HCP performing motor tasks. A detailed description of the HCP data is provided in Section \ref{section: Introduction}. We model squeezing right toes as the task of interest. Before the estimation step, we compute region-specific time courses using the AAL atlas \citep{tzourio2002automated} that consists of 120 brain regions. For each region, we extract the voxel-specific time series in that region and compute their spatial average for each time point. As a result, we obtain 120 time courses corresponding to 120 regions of interest. We select the left precentral gyrus (\texttt{PreCG.L}) as the seed region since it is located in the primary motor cortex and the motions of the right toes are associated with the left brain. We measure FC induced by $N(t)$ at a population level between the seed region and other regions using the following five approaches: the ptFCE algorithm, naive and task Pearson correlations, beta-series regression, and coherence analysis. In Section \ref{section: simulations}, the means $\hat{\rho}_{ij, mean}^{naiveCorr}$, $\hat{\rho}_{ij, mean}^{taskCorr}$, $\hat{\rho}_{ij, mean}^{betaS}$, and $\hat{\rho}^{Coh}_{ij, mean}$ tend to perform better than the corresponding medians. Hence, we show only the mean results in data analysis. Because of numerical issues, we omit regions \texttt{CB7.L}, \texttt{CB7.R}, and \texttt{CB10.L} and investigate the rest of the 117 regions. The detailed procedures for applying these approaches are described in Section \ref{section: simulations}.

Suppose $\mathcal{X}_{rj}$, for $r=1,\cdots,116$ and $j=1,\cdots,5$, are the estimated FC values between \texttt{PreCG.L} and other 116 regions computed using the referred five approaches indexed by $j$. Again, direct comparisons of methods are difficult as the estimations are defined on different scales \citep{cisler2014comparison}. Hence, we standardize the estimates by $\mathcal{X}_{r,j}^{(st)}=\frac{\mathcal{X}_{r,j}-\min\{\mathcal{X}_{r',j}\}_{r'=1}^{116}}{\max\{\mathcal{X}_{r',j}\}_{r'=1}^{116}-\min\{\mathcal{X}_{r',j}\}_{r'=1}^{116}}$,  for all $r$ and $j$, to enable comparisons. The standardized versions of the estimates are presented in Figure \ref{fig: Data Analysis Comparison}. To quantify the agreement between  estimation approaches implemented in this study, we use 0.5 as the threshold for standardized $\mathcal{X}_{r,j}^{(st)}$ to determine whether the node $r$ is connected to \texttt{PreCG.L} according to method $j$. Specifically, if $\mathcal{X}_{r,j}^{(st)}>0.5$, then method $j$ estimates that region $r$ is functionally connected to \texttt{PreCG.L}. Define $\mathbf{1}(\mathcal{X}_{r,j}^{(st)}>0.5)=:\vartheta_{r,j}$. As a result of applying each method $j$, we obtain the connectivity pattern $\{\vartheta_{r,j}\}_{r=1}^{116}$. Agreement between methods $j_1$ and $j_2$, for all $j_1,j_2=1,\cdots,5$, is measured by Cohen's kappa statistic denoted as $\kappa_{j_1, j_2}$ \citep{cohen1960coefficient} using classification results $\{\vartheta_{r,j_1}\}_{r=1}^{116}$ and $\{\vartheta_{r,j_2}\}_{r=1}^{116}$. The $\kappa_{j_1, j_2}$ along with p-values testing the null hypothesis indicating that the extent of agreement between each pair of methods is the same as random ($\kappa_{j_1, j_2}=0$) are presented in Table \ref{table: data comparison}. Using $\alpha=0.05$ and implementing multiple comparisons correction, we find that all methods have a significant agreement with each other.

The postcentral gyrus (\texttt{PoCG.L}) region is identified by all five methods as the region with strongest functional connectivity with the seed region. The estimated ptFCs indicate that the neural activity of the left precentral gyrus and that of left postcentral gyrus corresponding to squeezing right toes are highly correlated (either negatively or positively). Magnitudes of the neural activity corresponding to the task of interest in these two regions tend to be linearly dependent across the entire population. In addition to  modeling the advantages of our proposed ptFC approach described in Section \ref{section: advantages of the task-fc def}, we obtain differences between the estimates of ptFCE and those of competitor methods. We identify high connectivity between several regions and the seed region that is missed by competitor methods. Specifically, we obtain high task-evoked FC between the seed region and left/right thalamus (passing motor signals to the cerebral cortex), left/right paracentral lobule (motor nerve supply to the lower extremities), left superior temporal gyrus (containing the auditory cortex), and left Heschl gyrus (in the area of primary auditory cortex). These regions are related to the motor function or auditory cortex and can add to the existing results on functional connectivity between the precentral gyrus and the rest of the brain. 

To further visualize ptFCs induced by the task of interest, we apply MNI space coordinates of the 117 regions from the AAL atlas, where three-dimensional coordinates are obtained from \texttt{aal2.120} dataset in \texttt{R} package \texttt{brainGraph}. The regions are depicted by their MNI coordinates in Figure \ref{fig: 6plots12}. The grayscale shade of edges connecting each region to the seed region illustrates the estimated ptFC between the corresponding region and the seed region (Figure \ref{fig: 6plots12}, left panel), while the edges in the right panel of Figure \ref{fig: 6plots12} present the 30 highest ptFC values. Figure \ref{fig: 6plots12} shows that most of the large estimated ptFC values are in the left brain. This is expected since it is known that behaviors of extremities are functionally associated with contralateral brain regions \citep{nieuwenhuys2014central}.

\begin{table}
\centering
\caption{The Kappa statistics between each pair of methods are computed using the \texttt{R} function \texttt{Kappa.test} in package \texttt{fmsb}. The p-value of a Kappa statistic is presented in the parentheses below the statistic. }\label{table: data comparison}
\begin{tabular}{l|lllllllll}
\hline
Methods     & ptFCE &  & naive Persn             &  & task Persn             &  & beta series            &  & coherence              \\ \hline
ptFCE       & $*$   &  & $0.696$                 &  & $0.503$                &  & $0.260$                &  & $0.277$                \\
            &       &  & $(<0.0001)$ &  & $(<0.0001)$ &  & $(0.00636)$            &  & $(0.00402)$            \\
naive Persn & $*$   &  & $*$                     &  & $0.732$                &  & $0.391$                &  & $0.367$                \\
            &       &  &                         &  & $(<0.0001)$    &  & $(0.00022)$           &  & $(0.00051)$           \\
task Persn  & $*$   &  & $*$                     &  & $*$                    &  & $0.478$                &  & $0.497$                \\
            &       &  &                         &  &                        &  & $(<0.0001)$ &  & $(<0.0001)$ \\
beta series & $*$   &  & $*$                     &  & $*$                    &  & $*$                    &  & $0.961$                \\
            &       &  &                         &  &                        &  &                        &  & $(<0.0001)$ \\
coherence   & $*$   &  & $*$                     &  & $*$                    &  & $*$                    &  & $*$                    \\ \hline
\end{tabular}
\end{table}


\section{Conclusions}\label{section: Conclusions and Further Discussions}

In this paper, we introduce a model for task-evoked BOLD signals in fMRI studies. Building upon this model, we define a measure called ptFC, which characterizes task-evoked functional connectivity at the population level. We also address the limitations of several existing methods in the literature. We have developed the ptFCE algorithm for estimating the ptFC, demonstrating its computational efficiency. Furthermore, we provide comprehensive theoretical results regarding the properties of our proposed estimation procedure. Through simulation studies, we illustrate the superior performance of our proposed ptFCE algorithm when dealing with data generated by block-design task-fMRI data-generating mechanisms. We employ the ptFCE algorithm to estimate the ptFCs in a publicly available motor-task data set from the HCP. Notably, the results derived using our ptFCE algorithm exhibit functional connectivity patterns similar to widely used existing methods. In addition, ptFCE identifies functional connections between nodes that existing methods fail to detect. 

We plan to explore the inclusion of interaction terms such as $P_k(\omega;t)Q_k(\omega;t)$ (see Appendix H). In addition, we recognize the need for developing computationally efficient techniques to obtain variance estimates for the ptFCE algorithm. Lastly, persistent homology (PH) has been applied in brain network analyses for over a decade \citep{lee2011discriminative, li2023tree, wang2023topological}. The application of PH to our proposed ptFC is a topic for future research.

\begin{figure}[h]
    \centering
    \includegraphics[scale=0.52, angle=270]{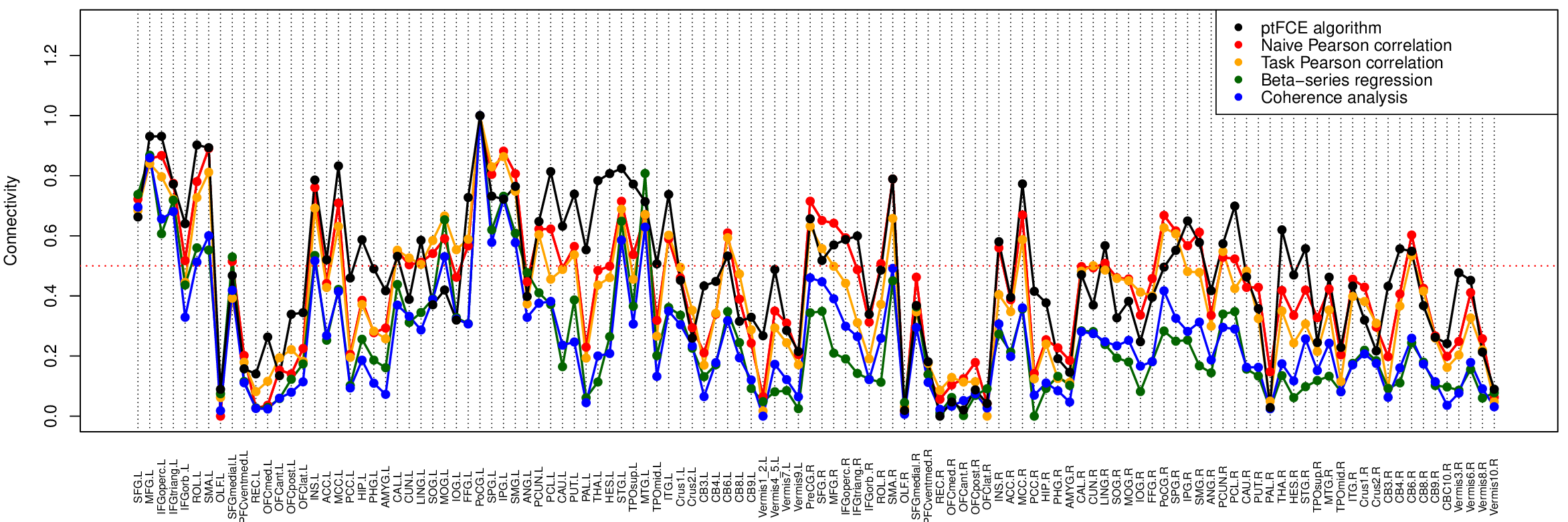}
    \caption{Illustration of estimation results from five FC estimation methods. The horizontal axis indicates the $116$ regions compared to \texttt{PreCG.L}. The abbreviations of region names are provided in the data set \texttt{aal2.120} in the \texttt{R} package \texttt{brainGraph}. The vertical axis presents the standardized connectivity measurements $\mathcal{X}_{r,j}^{(st)}$ between each region and the seed region \texttt{PreCG.R}. The red dotted horizontal line indicates the threshold 0.5 implemented for determining the connected/nonconnected relationships between \texttt{PreCG.L} and other regions.}
    \label{fig: Data Analysis Comparison}
\end{figure}

\section*{Data Availability Statement}

\noindent The data supporting the findings in Section \ref{section: applications} are publicly available on the HCP website: \\{\footnotesize\url{https://protocols.humanconnectome.org/HCP/3T/task-fMRI-protocol-details.html}} \citep{barch2013function}.\\ The \texttt{R} code for simulation studies and data analyses is available at \\ {\footnotesize \url{https://github.com/KMengBrown/Population-level-Task-evoked-Functional-Connectivity.git}}.


\clearpage
\newpage

\begin{center}
    \Large{\textbf{Appendix --- Supporting Materials for\\ ``Population-level Task-evoked Functional Connectivity\\ via Fourier Analysis"}}
\end{center}

\begin{appendix}
%
%

\section{Relationship between the proposed BOLD signal model and some existing models}\label{section: Relationship between the proposed BOLD signal model and some existing models}

In this paper, we implement the following model for observed task-fMRI BOLD signals $Y_k(\omega;t)$, for $\omega\in\Omega$ and $k\in\{1,\cdots,K\}$.
\begin{align}\label{app eq: BOLD model in our paper}
    & Y_k(\omega;t)=P_k(\omega;t)+Q_k(\omega;t)+\epsilon_k(\omega;t),\ \ \mbox{with}\\
    \notag & P_k(\omega;t)=\beta_k(\omega)\times N*h_k(t-t_{0,k}),
\end{align}
where $P_k(\omega;t)$ denotes the BOLD signal component stemming solely from the task $N(t)$ of interest, $Q_k(\omega;t)$ denotes spontaneous neural activity component and neural activity responding to nuisance tasks, and $\epsilon_k(\omega;t)$ is random error. Additionally, $t_{0,k}$ denotes the population shared latency and is usually smaller than the experimental unit \citep[e.g., see][p 138]{zhang2013semi}. The HRF shared by the whole population is $h_k(t)$. Here, we explore three existing models for BOLD signals. We show that these models are special cases of our proposed BOLD signal model (\ref{app eq: BOLD model in our paper}). 

\noindent {\bf Example 1:} Using the notations in our paper, the BOLD signal model implemented by \cite{zhang2013semi} can be represented as follows.
\begin{align}\label{app eq: BOLD in Zhang (2013)}
    Y_k(\omega;t)=X_k(\omega;t)^T d(\omega) + \sum_{\gamma=1}^\Gamma \beta_{k,\gamma}(\omega)\times \left(N_\gamma * h_{k,\gamma}\right)(t)+\epsilon_k(\omega;t),
\end{align}
where $N_\gamma(t)$ are task stimulus signals, $h_{k, \gamma}(t)$ are corresponding HRFs, and $X_k(\omega;t)^T d(\omega)$ characterizes the BOLD signal components stemming from other known sources, e.g., respiration and heartbeat. The BOLD signal model proposed by \cite{zhang2013semi} includes participant-dependent latency times. These latency times are essentially modeled as zero, since the values are usually much smaller than the corresponding experimental time unit. 

Suppose we are interested in the first experimental task $N_1(t)$. Define the following.
\begin{align}\label{app eq: notional changes}
    & \beta_k(\omega):=\beta_{k,1}(\omega),\ \ N(t):=N_1(t),\ \ P_k(\omega;t)=\beta_k(\omega)\times \left(N*h_k\right)(t), \\ 
    \notag & Q_k(\omega;t):=X_k(\omega;t)^T d(\omega) + \sum_{\gamma=2}^\Gamma \beta_{k,\gamma}(\omega)\times \left(N_\gamma * h_{k,\gamma}\right)(t).
\end{align}
Then (\ref{app eq: BOLD in Zhang (2013)}) is equivalent to the model (\ref{app eq: BOLD model in our paper}) with latency $t_{0,k}=0$, and there is no interaction term $P_k(\omega;t)Q_k(\omega;t)$.

\noindent {\bf Example 2:} Using the notations in our paper, the BOLD signal model implemented by \cite{warnick2018bayesian} is represented as follows \citep[see][Equations (1) and (2)]{warnick2018bayesian}.
\begin{align}\label{app eq: BOLD in Warnick (2018)}
    Y_k(\omega;t)=\mu_k(\omega)+\sum_{\gamma=1}^\Gamma \beta_{k,\gamma}(\omega)\times \left(N_\gamma*h_{k,\gamma}\right)(t)+\epsilon_k(\omega;t),
\end{align}
where $N_\gamma(t)$ are task stimulus signals, $h_{k,\gamma}(t)$ are corresponding HRFs, $\epsilon_k(\omega;t)$ are random error, and $\mu_k(\omega)$ present the baseline. Suppose we are interested in task $N_1(t)$. We apply the transform (\ref{app eq: notional changes}) and define $Q_k(\omega;t)=\mu_k(\omega)+\sum_{\gamma=2}^\Gamma \beta_{k,\gamma}(\omega)\times \left(N_\gamma * h_{k,\gamma}\right)(t)$. Then model (\ref{app eq: BOLD in Warnick (2018)}) is equivalent to our proposed model (\ref{app eq: BOLD model in our paper}) with latency $t_{0,k}=0$, without an interaction term $P_k(\omega;t)Q_k(\omega;t)$.

\noindent {\bf Example 3:} Motivated by the \textit{independent component analysis} framework, \cite{joel2011relationship} implemented the following model for task-fMRI BOLD signals.
\begin{align}\label{app eq: BOLD model Joel (2011)}
    \notag Y_k(\omega;t)=& M_{m,k}(\omega)\left\{\beta_{tm}(\omega)\times N*h_k(t)+\beta_{im}(\omega)\iota_{m}(t)\right\} \\
    \notag & + M_{v,k}(\omega)\left\{\beta_{tv}(\omega)\times N*h_k(t)+\beta_{iv}(\omega)\iota_{v}(t)\right\} \\
    & +M_b(\omega) \beta_b(\omega)\gamma(\omega;t),
\end{align}
where $M_{(\cdot),k}$ is the spatial mask at the $k^{th}$ node for visual cortex ($M_{v,k}$), motor cortex ($M_{m,k}$) or the whole brain ($M_{b}$), $N(t)$ is stimulus signal corresponding to the task of interest, $h_k(t)$ is an HRF, $\iota_{(\cdot)}$ is the intrinsic activity of the corresponding cortex, $\gamma(\omega;t)$ presents random noise, and $\beta_{tm}(\omega), \beta_{im}(\omega), \beta_{tv}(\omega), \beta_{iv}(\omega), \beta_{b}(\omega)$ are the weights of motor task, intrinsic motor activity, visual task, intrinsic visual activity and noise, respectively.

Since the task of interest in the HCP data is a motor task, we define the following notations. 
\begin{align*}
    & \beta_k(\omega):=M_{m,k}(\omega)\times \beta_{tm}(\omega), \\
    & Q_k(\omega):=M_{m,k}(\omega)\beta_{im}(\omega)\iota_m(t) + M_{v,k}(\omega)\left\{\beta_{tv}(\omega)\times N*h_k(t)+\beta_{iv}(\omega)\iota_{v}(t)\right\}, \\
    & \epsilon_k(\omega;t)= M_b(\omega) \beta_b(\omega)\gamma(\omega;t).
\end{align*}
Using the notations above, model (\ref{app eq: BOLD model Joel (2011)}) is the same as model (\ref{app eq: BOLD model in our paper}) with latency $t_{0,k}=0$, and there is no interaction term $P_k(\omega;t)Q_k(\omega;t)$.

\section{Lemmas, Theorems, and Their Proofs}\label{section: Lemmas, Theorems, and Their Proofs}

In this section, we provide proofs of theorems and lemmas presented in our paper. Throughout this section, the time index collection $\mathcal{T}$ denotes $\{\tau\Delta\}_{\tau=0}^T$, and $\Omega$ is a discrete and finite set. 

\begin{lemma}\label{lemma: zero points}
If $f: \mathcal{T}\rightarrow\R$ is not constant zero, its Fourier transform $\widehat{f}(\xi)$ has at most finitely many zero points (i.e., points $\xi\in\R$ such that $\widehat{f}(\xi)=0$) in any compact subset of $\R$.
\end{lemma}
\begin{proof}
Define the complex function $\Phi_f(z):=\frac{1}{T+1}\sum_{\tau=0}^T f(\tau\Delta)e^{-2\pi z (\tau\Delta)}$, for all $z=\eta+i\xi\in\mathbb{C}$. Since it is straightforward that $\Phi_f(z)$ satisfies the \textit{Cauchy-Riemann} equation $\frac{\partial}{\partial\overline{z}}\Phi_f(z)=0$, where $\frac{\partial}{\partial\overline{z}}=\frac{1}{2}(\frac{\partial}{\partial \eta}+ i \frac{\partial}{\partial \xi})$ is a Wirtinger derivative, the Looman–Menchoff theorem implies that $\Phi_f(z)$ is a holomorphic function. Then the zero points of $\Phi_f(z)$ are isolated, i.e., every zero point has a neighbourhood that does not contain any other zero point \citep[][Theorem 3.7]{conway1973functions}. Therefore, $\widehat{f}(\xi)=\Phi_f(i\xi)$ implies the desired result.   
\end{proof}

\begin{theorem}\label{thm: estimation theorem}
Suppose signals $\{Y_k(\omega;t)\vert t\in\mathcal{T}\}_{k=1}^K$, for $\omega\in\Omega$, are defined as in (\ref{app eq: BOLD model in our paper}), $\mathbb{E}\beta_k=\mathbb{E}R_k(t)=0$, and $t_{0,k}=\tau_{0,k}\Delta$ with some $\tau_{0,k}\in\mathbb{Z}$ for all $k=1,2,\cdots,K$ and $t\in\mathcal{T}$. Let the random variable $U:\Omega\rightarrow\mathcal{T}$ be uniformly distributed on $\mathcal{T}$. Furthermore, we assume that $U(\omega)$, $\{\beta_k(\omega)\}_{k=1}^K$, and $\{R_k(\omega;t)\vert t\in\mathcal{T}\}_{k=1}^K$ are independent. Then, the autocovariance differences $\mathcal{A}_{kl}(s) :=  \mathbb{E}\left\{ Y_k(t-U) Y_l(t+s-U) \right\} -\mathbb{E}\left\{ R_k(t-U)R_l(t+s-U)\right\}$ depend only on $s$, the Fourier transforms of $\mathcal{A}(s)$ are
\begin{align*}
    \widehat{\mathcal{A}}_{kl}(\xi)=\mathbb{E}\left( \beta_k \beta_l\right) \cdot \left\vert \widehat{N}(\xi) \right\vert^2 \cdot \overline{\widehat{h}_k(\xi)}\,\cdot\, \widehat{h}_l(\xi)\,\cdot\, e^{2\pi i (t_{0,k}-t_{0,l})\xi}.
\end{align*}
Furthermore, we have the following representation.
\begin{align*}
    \mathcal{C}_{kl}(\xi):= \vert \widehat{\mathcal{A}}_{kl}(\xi)\vert \Big/ \sqrt{  \vert \widehat{\mathcal{A}}_{kk}(\xi) \widehat{\mathcal{A}}_{ll}(\xi)  \vert  } = \left\vert corr\left(\beta_k, \beta_l\right)\right\vert, \mbox{ for all }\xi\in\R.
\end{align*}
\end{theorem}
\begin{proof}
The independence between $U(\omega)$, $\{\beta_k(\omega)\}_{k=1}^K$, and $\{R_k(\omega;t)\vert t\in\mathcal{T}\}_{k=1}^K$ implies
\begin{align}\label{app eq: convolution calculation}
\notag & \mathbb{E}[Y_k(t-U) Y_l(t+s-U)] - \mathbb{E}[R_k(t-U) R_l(t+s-U)]\\ 
\notag & = \mathbb{E}(\beta_k\beta_l)\times \mathbb{E}\left[ (N*h_k)(t-t_{0,k}-U) \times (N*h_l)(t+s-t_{0,l}-U)\right]\\ \notag
& = \mathbb{E}(\beta_k \beta_l) \times \frac{1}{T+1} \sum_{u=0}^T \Big\{(N*h_k)\left((\tau-\tau_{0,k}-u)\Delta\right)\times (N*h_l)\left((\underline{s}+\tau_{0,k}-\tau_{0,l})\Delta+(\tau-\tau_{0,k}-u)\Delta\right)\Big\}\\ \notag
& = \mathbb{E}(\beta_k \beta_l) \times \frac{1}{T+1} \sum_{v=-(\tau-\tau_{0,k})}^{T-(\tau-\tau_{0,k})} \Big\{(N*h_k)\left(-v\Delta\right)\times (N*h_l)\left((\underline{s}+\tau_{0,k}-\tau_{0,l})\Delta-v\Delta\right)\Big\}\\ 
& = \mathbb{E}(\beta_k \beta_l) \times \left[N*h_k(-\cdot)\right]*\left[N*h_l\right]\left(s+(t_{0,k}-t_{0,l})\right),
\end{align}
which only depends on $s$ and does not depend on $t$. Here, the last equality follows from the periodic extension and the definition of convolution. Then, we have the following Fourier transform,
\begin{align*}
\widehat{\mathcal{A}}_{kl}(\xi)& = \mathbb{E}(\beta_k \beta_l) \times \overline{\widehat{(N*h_k)}(\xi)} \widehat{(N*h_l)}(\xi)e^{2\pi i (t_{0,k}-t_{0,l})\xi}\\
& =\mathbb{E}(\beta_k \beta_l) \times \left\vert\widehat{N}(\xi)\right\vert^2\overline{\widehat{h_k}(\xi)}\widehat{h_l}(\xi) e^{2\pi i (t_{0,k}-t_{0,l})\xi},
\end{align*}
for all $\xi\in\R$, which implies $\mathcal{C}_{kl}(\xi) = \left\vert corr\left(\beta_k, \beta_l\right)\right\vert$ for all $\xi\in\R$.  
\end{proof}

\begin{lemma}\label{lemma: stationarity with U} 
The vector-valued stochastic process $\pmb{Z}(\omega;t)=\{(Z_1(\omega;t), Z_2(\omega;t), \cdots, Z_K(\omega;t))^T\}_{t\in\mathcal{T}}$ is defined on probability space $(\Omega, 2^{\Omega}, \mathbb{P})$ and is weakly stationary with mean zero, $U:\Omega\rightarrow\mathcal{T}$ is uniformly distributed, and $U(\omega)$ and $\pmb{Z}(\omega;t)$ are independent. Then the vector-valued stochastic process $\pmb{Z}(\omega;t-U(\omega))=\{(Z_1(\omega;t-U(\omega)), Z_2(\omega;t-U(\omega)), \cdots, Z_K(\omega;t-U(\omega)))^T\}_{t\in\mathcal{T}}$ is weakly stationary with mean zero as well. Additionally, $\mathbb{E}[Z_k(t-U)Z_l(t+s-U)]=\mathbb{E}[Z_k(0)Z_l(s)]$, for all $k,l=1,2,\cdots,K$.
\end{lemma}
\begin{proof}
Let $\mu$ be the probability measure on $\mathcal{T}$ associated with the uniform random variable $U:\Omega\rightarrow\mathcal{T}$, i.e., $\mu=\mathbb{P}\circ U^{-1}$, then we have $\mathbb{E}\{Z_k(t-U)\}=\int_{\mathcal{T}}\mathbb{E}\{Z_k(t-U)\vert U=u\}\mu(du)$. The independence between $U$ and $\pmb{Z}$ implies $\mathbb{E}\{Z_k(t-U)\vert U=u\}=\mathbb{E}\{Z_k(t-u)\}=0$ for all $t$, which results in $\mathbb{E}\{Z_k(t-U)\}=0$ for all $t\in\mathcal{T}$. Similarly, we have
\begin{align*}
    \mathbb{E}\left\{Z_k(t-U)Z_l(t+s-U)\right\}&=\int_{\mathcal{T}} \mathbb{E}\left\{Z_k(t-U)Z_l(t+s-U)\vert U=u\right\}\mu(du)\\
    &=\int_{\mathcal{T}} \mathbb{E}\left\{Z_k(t-u)Z_l(t+s-u)\right\}\mu(du)\\
    &=\int_{\mathcal{T}} \mathbb{E}\left\{Z_k(0)Z_l(s)\right\}\mu(du)\\
    &=\mathbb{E}\left\{Z_k(0)Z_l(s)\right\},
\end{align*}
This completes the proof of the desired result.  
\end{proof}

\begin{theorem}\label{theorem: asymptotic bias}
Suppose we have the following random vectors and stochastic processes.
\begin{itemize}
    \item $\mathfrak{B}:=\left\{(\beta_k(\omega), \beta_l(\omega))\right\}_{\omega}^n$ are $n$ i.i.d. random vectors.
    \item $\mathfrak{R}:=\left\{\left(\tilde{R}_k(\omega;t), \tilde{R}_l(\omega;t)\right)\vert t\in\mathcal{T}\right\}_{\omega=1}^n$ are $n$ i.i.d. stochastic processes.
    \item $\mathfrak{W}:=\left\{\left(W_k(\omega;t), W_l(\omega;t)\right)\vert t\in\mathcal{T}\right\}_{\omega=1}^n$ are $n$ i.i.d. stochastic processes satisfying (i) random vectors $\left(W_k(\omega;t_1), W_l(\omega;t_1)\right)$ and $\left(W_k(\omega;t_2), W_l(\omega;t_2)\right)$ are independent whenever $t_1\ne t_2$, (ii) the stochastic process $\{\left(W_k(\omega;t), W_l(\omega;t)\right)\vert t\in\mathcal{T}\}$ is weakly stationary with mean zero, and (iii) $\Sigma_{kl}=\mathbb{E}\left\{W_k(t)W_l(t)\right\}$ for all $t\in\mathcal{T}$.
\end{itemize}
The collections $\mathfrak{B}, \mathfrak{R}, \mathfrak{W}$ are independent. Furthermore, we have the following task-fMRI BOLD signals.$^\dagger$ \footnote{$\dagger$: Observed task-fMRI BOLD signals are $Y_k(\omega;t)=\beta_k(\omega)\times N*h_k(t-t_{0,k})+R_k(\omega;t)$, where $R_k(\omega;t)$ are true reference signals. However, the true reference signals are not observable in applications and should be estimated by $\tilde{R}_k(\omega;t)$ using the AMUSE algorithm. The corresponding estimation bias $R_k(\omega;t)-\tilde{R}_k(\omega;t)$ is denoted by $W_k(\omega;t)$. Hence, we have $Y_k(\omega;t)=\beta_k(\omega)\times N*h_k(t-t_{0,k})+\tilde{R}_k(\omega;t)+W_k(\omega;t)$, which is just a representation of $Y_k(\omega;t)=\beta_k(\omega)\times N*h_k(t-t_{0,k})+R_k(\omega;t)$.}
\begin{align*}
    Y_{k'}(\omega;t)=\beta_{k'}(\omega)\times N*h_{k'}(t-t_{0,{k'}})+\tilde{R}_{k'}(\omega;t)+W_{k'}(\omega;t),
\end{align*}
for $k'\in\{k,l\}$ and $\omega=1,\cdots,n$. We define the following estimator for $\mathcal{C}_{kl}(\xi)$.
\begin{align*}
\mathcal{C}_{kl}^{est,n}(\xi):= \left\vert \widehat{\mathcal{A}_{kl}^{est,n}}(\xi)\right\vert \Bigg/ \sqrt{  \left\vert \widehat{\mathcal{A}_{kk}^{est,n}}(\xi) \widehat{\mathcal{A}_{ll}^{est,n}}(\xi)  \right\vert  } ,\ \ \ \mbox{for all } \xi\in\mathbb{R},
\end{align*}
where hat $\widehat{(\cdot)}$ denotes the Fourier transform and 
\begin{align*}
    \notag \mathcal{A}_{kl}^{est,n}(s):=& \frac{1}{T+1}\sum_{t\in\{\tau\Delta\}_{\tau=0}^T}\left[  \frac{1}{n}\sum_{\omega=1}^n  \left\{Y_k\left(\omega;t-U(\omega)\right)Y_l(\omega;t+s-U(\omega))\right\} \right] \\
    & - \frac{1}{T+1}\sum_{t\in\{\tau\Delta\}_{\tau=0}^T}\left[  \frac{1}{n}\sum_{\omega=1}^n  \left\{\tilde{R}_k(\omega;t-U(\omega)) \tilde{R}_l(\omega;t+s-U(\omega))\right\} \right].
\end{align*}
Then we have the following asymptotic behavior of $\mathcal{C}_{kl}^{est,n}(\xi)$ as $n\rightarrow\infty$.
\begin{align}\label{app eq: asymptotic behavior}
    \mathbb{P}\left\{\lim_{n\rightarrow\infty}\mathcal{C}^{est,n}_{kl}(\xi) = 
    \frac{ \left\vert \mathbb{E}(\beta_k\beta_l)\frac{\overline{\widehat{h}_k(\xi)}\widehat{h}_l(\xi)}{\left\vert \widehat{h}_k(\xi) \widehat{h}_l(\xi) \right\vert}  e^{2\pi i \xi (t_{0,k}-t_{0,l})} + \frac{\Sigma_{kl}}{(T+1)\left\vert\widehat{N}(\xi)\right\vert^2\left\vert\widehat{h}_k(\xi) \widehat{h}_l(\xi)\right\vert} \right\vert }{ \sqrt{ \left(\mathbb{E}(\beta_k^2)+\frac{\Sigma_{kk}}{(T+1)\left\vert\widehat{N}(\xi)\widehat{h}_k(\xi)\right\vert^2}\right) \left(\mathbb{E}(\beta_l^2)+\frac{\Sigma_{ll}}{(T+1)\left\vert\widehat{N}(\xi)\widehat{h}_l(\xi)\right\vert^2}\right) } }\right\}=1.
\end{align}
\end{theorem}
\begin{proof}
For each pair $(s,t)\in\mathcal{T}^2$, the \textit{strong law of large numbers} implies that there exists $\mathcal{N}_{s,t}\in 2^{\Omega}$ depending on the pair $(s,t)$ such that $\mathbb{P}(\mathcal{N}_{s,t})=0$ and
\begin{align*}
& \frac{1}{n}\sum_{\omega=1}^n \left[ \left\{Y_k\left(\omega;t-U(\omega)\right)Y_l(\omega;t+s-U(\omega))\right\}-\left\{\tilde{R}_k\left(\omega;t-U(\omega)\right)\tilde{R}_l(\omega;t+s-U(\omega))\right\}\right]\\
   & \rightarrow\mathbb{E}\left\{Y_k(t-U)Y_l(t+s-U)\right\}-\mathbb{E}\left\{\tilde{R}_k(t-U)\tilde{R}_l(t+s-U)\right\}=:\mathcal{D}_{kl}(t,s),
\end{align*}
in $\Omega-\mathcal{N}_{s,t}$ as $n\rightarrow\infty$. Then we have 
\begin{align}\label{app eq: limit of A}
    \lim_{n\rightarrow\infty}\mathcal{A}_{kl}^{est,n}(s)=\frac{1}{T+1}\sum_{t\in\{\tau\Delta\}_{\tau=0}^T} \mathcal{D}_{kl}(t,s)
\end{align}
for all $s\in\mathcal{T}$ in $\Omega-\mathcal{N}$, where $\mathcal{N}:=\bigcup_{{s,t}\in\mathcal{T}}\mathcal{N}_{s,t}$. The limit (\ref{app eq: limit of A}) implies the following in $\Omega-\mathcal{N}$.
\begin{align*}
    \lim_{n\rightarrow\infty}\widehat{\mathcal{A}_{kl}^{est,n}}(\xi)&=\lim_{n\rightarrow\infty}\left\{\frac{1}{T+1}\sum_{s\in\{\tau\Delta\}_{\tau=0}^T} \mathcal{A}_{kl}^{est,n}(s) e^{2\pi i \xi s}\right\}\\
    &=\frac{1}{T+1}\sum_{s\in\{\tau\Delta\}_{\tau=0}^T} \lim_{n\rightarrow\infty}\left\{\mathcal{A}_{kl}^{est,n}(s)\right\} e^{2\pi i \xi s}\\
    &=\frac{1}{T+1}\sum_{t\in\{\tau\Delta\}_{\tau=0}^T} \left\{\frac{1}{T+1}\sum_{s\in\{\tau\Delta\}_{\tau=0}^T}  \mathcal{D}_{kl}(t,s) e^{2\pi i \xi s}\right\}
\end{align*}
for all $\xi\in\mathbb{R}$. In the derivation above, we change the order of summation and limit as the summation contains only finitely many terms. Since $\frac{1}{T+1}\sum_{s\in\{\tau\Delta\}_{\tau=0}^T}  \mathcal{D}_{kl}(t,s) e^{2\pi i \xi s}$ is the Fourier transform of $\mathcal{D}_{kl}(t,s)$ with respect to $s$, repeating the calculation strategy in (\ref{app eq: convolution calculation}), we have 
\begin{align*}
    \frac{1}{T+1}\sum_{s\in\{\tau\Delta\}_{\tau=0}^T}  \mathcal{D}_{kl}(t,s) e^{2\pi i \xi s}=\mathbb{E}(\beta_k\beta_l)\frac{\overline{\widehat{h}_k(\xi)}\widehat{h}_l(\xi)}{\left\vert \widehat{h}_k(\xi) \widehat{h}_l(\xi) \right\vert}  e^{2\pi i \xi (t_{0,k}-t_{0,l})} + \frac{\Sigma_{kl}}{(T+1)\left\vert\widehat{N}(\xi)\right\vert^2\left\vert\widehat{h}_k(\xi) \widehat{h}_l(\xi)\right\vert},
\end{align*}
which does not depend on $t$. Then we have 
\begin{align}\label{app eq: limit of A hat}
    \lim_{n\rightarrow\infty}\widehat{\mathcal{A}_{kl}^{est,n}}(\xi)=\mathbb{E}(\beta_k\beta_l)\frac{\overline{\widehat{h}_k(\xi)}\widehat{h}_l(\xi)}{\left\vert \widehat{h}_k(\xi) \widehat{h}_l(\xi) \right\vert}  e^{2\pi i \xi (t_{0,k}-t_{0,l})} + \frac{\Sigma_{kl}}{(T+1)\left\vert\widehat{N}(\xi)\right\vert^2\left\vert\widehat{h}_k(\xi) \widehat{h}_l(\xi)\right\vert}
\end{align}
for all $\xi\in\mathbb{R}$ in $\Omega-\mathcal{N}$. The limit (\ref{app eq: limit of A hat}) implies
\begin{align*}
\lim_{n\rightarrow\infty}\mathcal{C}_{kl}^{est,n}(\xi)&= \left\vert \lim_{n\rightarrow\infty} \widehat{\mathcal{A}_{kl}^{est,n}}(\xi)\right\vert \Bigg/ \sqrt{  \left\vert \lim_{n\rightarrow\infty} \widehat{\mathcal{A}_{kk}^{est,n}}(\xi)\times  \lim_{n\rightarrow\infty} \widehat{\mathcal{A}_{ll}^{est,n}}(\xi)  \right\vert  }\\
&=\frac{ \left\vert \mathbb{E}(\beta_k\beta_l)\frac{\overline{\widehat{h}_k(\xi)}\widehat{h}_l(\xi)}{\left\vert \widehat{h}_k(\xi) \widehat{h}_l(\xi) \right\vert}  e^{2\pi i \xi (t_{0,k}-t_{0,l})} + \frac{\Sigma_{kl}}{(T+1)\left\vert\widehat{N}(\xi)\right\vert^2\left\vert\widehat{h}_k(\xi) \widehat{h}_l(\xi)\right\vert} \right\vert }{ \sqrt{ \left(\mathbb{E}(\beta_k^2)+\frac{\Sigma_{kk}}{(T+1)\left\vert\widehat{N}(\xi)\widehat{h}_k(\xi)\right\vert^2}\right) \left(\mathbb{E}(\beta_l^2)+\frac{\Sigma_{ll}}{(T+1)\left\vert\widehat{N}(\xi)\widehat{h}_l(\xi)\right\vert^2}\right) } }
\end{align*}
for all $\xi\in\mathbb{R}$ in $\Omega-\mathcal{N}$. Since $\mathbb{P}(\mathcal{N})=\mathbb{P}\left(\bigcup_{s,t\in\mathcal{T}}\mathcal{N}_{s,t}\right)\le\sum_{s,t\in\mathcal{T}}\mathbb{P}(\mathcal{N}_{s,t})=0$, the desired result (\ref{app eq: asymptotic behavior}) follows.  
\end{proof}

\textbf{Proof of Theorem \ref{thm: 1st thm for AMUSE-ptFC} in the manuscript:} Since $\{\epsilon(\omega;t)\}_{t\in\mathcal{T}}$ is WSMZ, Lemma \ref{lemma: stationarity with U} implies that $\{\epsilon(\omega;t-U(\omega))\}_{t\in\mathcal{T}}$ is WSMZ as well. Because of $\mathbb{E}[\beta_{k,\gamma}]=0$ and the independence between $\beta_{k,\gamma}(\omega)$ and $U(\omega)$, we have $\mathbb{E}[R_k(t-U)]=0$ for all $t$. Additionally, the independence assumption in Theorem \ref{thm: 1st thm for AMUSE-ptFC} implies the following
\begin{align*}
\mathbb{E}[R_k(t-U)\cdot R_k(t+s-U)] = \mathbb{E}[Q_k(t-U)\cdot Q_k(t+s-U)] + \mathbb{E}[\epsilon_k(t-U)\cdot \epsilon_k(t+s-U)]
\end{align*}

\textbf{Proof of Theorem \ref{thm: 2nd thm for AMUSE-ptFC} in the manuscript:} In the following equation provided in Theorem \ref{thm: 1st thm for AMUSE-ptFC}, $\pmb{P}$ is a permutation matrix and $\pmb{\Lambda}$ is a diagonal matrix.
\begin{align}\label{app eq: thm 1 equation 1}
    \begin{pmatrix}
    1 & 1 \\
    \frac{1}{\beta^*} & 0
\end{pmatrix} = \pmb{A} \pmb{\Lambda}^{-1} \pmb{P}^{-1}.
\end{align}
Then equation (\ref{app eq: thm 1 equation 1}) implies $\beta^*=\beta_k(\omega)=a_{1i'}/a_{2i'}$ for some $i'\in\{1,2\}$ and $a_{2j}=0$ when $j\ne i'$. Furthermore, the pair $(\pmb{A}, \pmb{S}(\omega;t))$ also satisfies the following equation, which is provided in Theorem \ref{thm: 1st thm for AMUSE-ptFC} as well.
\begin{align}\label{app eq: thm 1 equation 2}
    \pmb{A}
    \begin{pmatrix}
    s_1(\omega;t) \\
    s_2(\omega;t)
    \end{pmatrix} =
    \begin{pmatrix}
    1 & 1 \\
    \frac{1}{\beta^*} & 0
    \end{pmatrix}
    \begin{pmatrix}
    J_k(\omega;t) \\
    R_k(\omega;t-U(\omega))
    \end{pmatrix} =
    \begin{pmatrix}
    Y_k(\omega;t-U(\omega)) - \beta^* C_k \\
    (N*h_k)(t-t_{0,k}-U(\omega))-C_k
    \end{pmatrix}, 
\end{align}
and equation (\ref{app eq: thm 1 equation 2}) implies $a_{2i'}s_{i'}(\omega;t)=(N*h_k)(t-t_{0,k}-U(\omega))-C_k$. Therefore, we have
\begin{align*}
    a_{1i'}s_{i'}(\omega;t)=\frac{a_{1i'}}{a_{2i'}}\times a_{2i'}s_{i'}(\omega;t)=\beta_k(\omega)\times\left\{(N*h_k)(t-t_{0,k}-U(\omega))-C_k\right\}=J_k(\omega;t).
\end{align*}  

The following theorem supports our Assumption 2.
\begin{theorem}
Let signals $\{Q_k(\omega;t)\}_{t\in\mathcal{T}}$ be of the following form
\begin{align}\label{eq: an example of Qk}
    Q_k(\omega;t)=\left\{\sum_{\gamma=1}^\Gamma \beta_{k,\gamma}(\omega)\times \tilde{N}_{\gamma}*\tilde{h}_{k,\gamma}(t)\right\}.
\end{align}
Let $\{\epsilon_k(\omega;t)\}_{t\in\mathcal{T}}$ be a random error signal. $U(\omega)$ is a random variable uniformly distributed on $\mathcal{T}$. Suppose the following conditions are satisfied:
\begin{enumerate}
    \item $\mathbb{E}\beta_{k,\gamma}=0$ for all $\gamma=1,\ldots,\Gamma$;
    \item $U(\omega)$, $\beta_{k,1}(\omega)$, $\beta_{k,2}(\omega), \ldots, \beta_{k,\Gamma}(\omega)$, and $\{\epsilon_k(\omega;t)\}_{t\in\mathcal{T}}$ are independent;
    \item $\{\epsilon_k(\omega;t)\}_{t\in\mathcal{T}}$ is WSMZ.
\end{enumerate}
Then, the reference term $\{R_k(\omega;t)=Q_k(\omega;t)+\epsilon_k(\omega;t)\}_{t\in\mathcal{T}}$ is WSMZ.
\end{theorem}
\noindent\textbf{Remark:} If the condition i) is not satisfied, we may apply the following centralization
\begin{align*}
    R_k(\omega;t)-\mathbb{R}[R_k(t)] = \left\{\sum_{\gamma=1}^\Gamma \left[\beta_{k,\gamma}(\omega)-\mathbb{E}(\beta_{k,\gamma})\right]\times \tilde{N}_{\gamma}*\tilde{h}_{k,\gamma}(t)\right\} + \left(\epsilon_k(\omega;t)-\mathbb{E}[\epsilon_k(t)] \right)
\end{align*}
and use $\left[\beta_{k,\gamma}(\omega)-\mathbb{E}(\beta_{k,\gamma})\right]$ as $\beta_{k,\gamma}(\omega)$.
\begin{proof}
The conditions i) and ii) imply
\begin{align*}
    \mathbb{E}[R_k(t-U)] = \left\{\sum_{\gamma=1}^\Gamma \mathbb{E}\beta_{k,\gamma}\times \mathbb{E}\left[\tilde{N}_{\gamma}*\tilde{h}_{k,\gamma}(t-U)
    \right]\right\} + \mathbb{E}[\epsilon_k(t-U)]=\mathbb{E}[\epsilon_k(t-U)].
\end{align*}
In addition, the condition iii), together with Lemma \ref{lemma: stationarity with U}, implies that $\{\epsilon_k(\omega;t-U(\omega))\}_{t\in\mathcal{T}}$ is WSMZ, hence, 
\begin{align*}
     \mathbb{E}[R_k(t-U)] = \mathbb{E}[\epsilon_k(t-U)] = 0.
\end{align*}
Condition ii) implies the following
\begin{align*}
    \mathbb{E}[R_k(t-U)\cdot R_k(t+s-U)] = \mathbb{E}[Q_k(t-U)\cdot Q_k(t+s-U)] + \mathbb{E}[\epsilon_k(t-U)\cdot \epsilon_k(t+s-U)].
\end{align*}
Since $\{\epsilon_k(\omega;t-U(\omega))\}_{t\in\mathcal{T}}$ is WSMZ, the second term above depends only on $s$. The condition ii) implies the following
\begin{align*}
    \mathbb{E}[Q_k(t-U)\cdot Q_k(t+s-U)] = \left\{\sum_{\gamma=1}^\Gamma \mathbb{E}\left(\beta_{k,\gamma}^2\right)\times \mathbb{E}\left[\tilde{N}_{\gamma}*\tilde{h}_{k,\gamma}(t-U)\cdot \tilde{N}_{\gamma}*\tilde{h}_{k,\gamma}(t+s-U)
    \right]\right\}.
\end{align*}
The calculation in Eq.~\eqref{app eq: convolution calculation} indicates that $\mathbb{E}\left[\tilde{N}_{\gamma}*\tilde{h}_{k,\gamma}(t-U)\cdot \tilde{N}_{\gamma}*\tilde{h}_{k,\gamma}(t+s-U)
    \right]$ depends only on $s$. Hence, $\mathbb{E}[Q_k(t-U)\cdot Q_k(t+s-U)]$ depends only on $s$. Therefore, the reference term $\{R_k(\omega;t)=Q_k(\omega;t)+\epsilon_k(\omega;t)\}_{t\in\mathcal{T}}$ is WSMZ.
\end{proof}

\section{Identifiability of Task-evoked Terms}\label{section: The Identifiability of Task-evoked Terms}

In this section, we prove the identifiability of task-evoked terms using the identifiability theory proposed by \cite{tong1991indeterminacy}. To implement this theory, we need a uniformly distributed random variable $U$ as an auxiliary variable. In our approach to estimating ptFC, the auxiliary random variable $U$ is artificially generated. Hence, all the distributional information of $U$, e.g., its independence, is viewed as available.

We generate $U: \Omega\rightarrow\mathcal{T}$ to be a uniformly distributed random variable. Identifying the task-evoked terms $P_k(\omega;t)$ from $Y_k(\omega;t)=P_k(\omega;t)+R_k(\omega;t)$ is equivalent to identifying $P_k(\omega;t-U(\omega))$ from $Y_k(\omega;t-U(\omega))$. Specifically, the identification of the task-evoked terms $P_k(\omega;t)$ consists of the following steps:
\begin{enumerate}
\item We first generate random variable $U(\omega)$;
    \item then, we identify $P_k(\omega;t-U(\omega))$ from $Y_k(\omega;t-U(\omega))=P_k(\omega;t-U(\omega))+R_k(\omega;t-U(\omega))$;
    \item lastly, we identify $P_k(\omega;t)$ by taking the transform $P_k(\omega;t-U(\omega))\mapsto P_k(\omega;t)$.
\end{enumerate}
To prove the identifiability of the task-evoked terms $P_k(\omega;t)$, it suffices to show the identifiability in the second step above. Our proposed BOLD signal model $Y_k(\omega;t-U(\omega))=P_k(\omega;t-U(\omega))+R_k(\omega;t-U(\omega))$ can be represented in the following ``mixing" form.
\begin{align}\label{app eq: our proposed BOLD signal model}
Y_k\left(\omega;t-U(\omega)\right)=\begin{pmatrix}
    1 & 0 \\
    0 & 1
    \end{pmatrix}\begin{pmatrix}
    P_k(\omega;t-U(\omega))\\
    R_k(\omega;t-U(\omega))
    \end{pmatrix},
\end{align}
where the identity matrix mixes the source signals $P_k(\omega;t-U(\omega))$ and $R_k(\omega;t-U(\omega))$. 

In this section, we provide the identifiability of the task-evoked terms $P_k(\omega;t-U(\omega))=\beta_k(\omega)N*h_k(t-t_{0,k}-U(\omega))$ as well as reference terms $R_k(\omega;t-U(\omega))$ in (\ref{app eq: our proposed BOLD signal model}), up to deterministic coefficients, among the family of \textit{mixing forms} defined as follows.
\begin{align}\label{app eq: linear mixing forms}
    \mathcal{M}:=\left\{ \begin{pmatrix}
    s_1(\omega;t)\\
    s_2(\omega;t)
    \end{pmatrix} \,\Bigg\vert\, \exists \pmb{A}\in\mathbb{R}^{2\times 2}\mbox{ such that }Y_k\left(\omega;t-U(\omega)\right)=\pmb{A}\begin{pmatrix}
    s_1(\omega;t)\\
    s_2(\omega;t)
    \end{pmatrix}\right\},
\end{align}
where all matrices $\pmb{A}$ are deterministic 2-by-2 matrices mixing \textit{source signals} $s_1(\omega;t)$ and $s_2(\omega;t)$. Specifically, we will show that, under some probabilistic conditions, all the stochastic signals $\left(s_1(\omega;t), s_2(\omega;t)\right)$ in (\ref{app eq: linear mixing forms}) are proportional to $(P_k(\omega;t-U(\omega)), R_k(\omega;t-U(\omega)))$ or $(R_k(\omega;t-U(\omega)), P_k(\omega;t-U(\omega)))$. Then the task-evoked terms $P_k(\omega;t-U(\omega))$ are identifiable in $\mathcal{M}$ in the following sense: there exist $i'\in\{1,2\}$ and $\lambda\ne0$ such that $s_{i'}(\omega;t)=\lambda P_k(\omega;t-U(\omega))=\{\lambda\beta_k(\omega)\}N*h_k(t-t_{0,k}-U(\omega))$ for all $t\in\mathcal{T}$. Since the deterministic coefficient $\lambda$ is canceled out in computing $corr(\beta_k,\beta_l)$, the identifiability of $P_k(\omega;t)$ up to a deterministic coefficient implies the exact identifiability of ptFC $\vert corr(\beta_k,\beta_l)\vert$.

The following theorem implies the identifiability of $P_k(\omega;t)$ and provides a formal foundation for the discussion above.
\begin{theorem}\label{thm: identifiability theorem}
Let $U:\Omega\rightarrow\mathcal{T}$ be uniformly distributed. Suppose $U(\omega)$, $\beta_k(\omega)$, and $\{R_k(\omega;t)\vert t\in\mathcal{T}\}$ are independent, and $\{R_k(\omega;t-U(\omega))\}_{t\in\mathcal{T}}$ is WSMZ. Furthermore, we assume that there exists a $t^*\in\mathcal{T}$ such that 
\begin{align*}
    \frac{ \mathbb{E}\left\{P_k(t-U)P_k(t-t^*-U)\right\} }{\mathbb{E}\left\{P_k(t-U)\right\}^2} \ne \frac{ \mathbb{E}\left\{R_k(t-U)R_k(t-t^*-U)\right\} }{\mathbb{E}\left\{R_k(t-U)\right\}^2}.
\end{align*}
If there exists a WSMZ stochastic process $(s_1(\omega;t), s_2(\omega;t))$ such that $\{s_1(\omega;t)\}_{t\in\mathcal{T}}$ and $\{s_2(\omega;t)\}_{t\in\mathcal{T}}$ are uncorrelated, 
\begin{align*}
& \frac{\mathbb{E}\left\{s_1(t)s_1(t-t^*)\right\}}{\mathbb{E}\left\{s_1(t)\right\}^2} \ne \frac{\mathbb{E}\left\{s_2(t)s_2(t-t^*)\right\}}{\mathbb{E}\left\{s_2(t)\right\}^2}, \mbox{ and}\\
&    Y_k(\omega;t-U(\omega))=P_k(\omega;t-U(\omega))+R_k(\omega;t-U(\omega))=\pmb{A}\begin{pmatrix}
    s_1(\omega;t)\\
    s_2(\omega;t)
    \end{pmatrix}
\end{align*}
for some deterministic 2-by-2 matrix $\pmb{A}$, then there exists a permutation matrix $\pmb{P}$ and a non-singular diagonal matrix $\pmb{\Lambda}$ such that
\begin{align}\label{app eq: the last identifiability equation}
    \begin{pmatrix}
    P_k(\omega;t-U(\omega))\\
    R_k(\omega;t-U(\omega))
    \end{pmatrix}=\pmb{P\Lambda}\begin{pmatrix}
    s_1(\omega;t)\\
    s_2(\omega;t)
    \end{pmatrix}.
\end{align}
\end{theorem}
\noindent\textbf{Remark}: (i) Our discussion of identifiability presented at the beginning of this section is based on (\ref{app eq: the last identifiability equation}). (ii) Since $\mathbb{E}\beta_k=\mathbb{E}R_k(t)=0$ for all $t$, we have $\mathbb{E}\left\{P_k(t_1)R_k(t_2)\right\}=0$ for all $t_1, t_2\in\mathcal{T}$, i.e., $P_k(\omega;t_1)$ and $R_k(\omega;t_2)$ are uncorrelated. (iii) One can verify that  $P_k(\omega;t-U(\omega))=\beta_k(\omega)N*h_k(t-t_{0,k}-U(\omega))$ is WSMZ. 

Based on these remarks, Theorem \ref{thm: identifiability theorem} is a straightforward result following from Theorem 2 of \cite{tong1991indeterminacy}.

\section{Bias and Variance --- Performance of ptFCE in Estimating ptFCs}\label{section: The Performance of ptFCE in Estimating ptFCs}

In this section, we analyze the performance of the ptFCE algorithm in terms of estimation bias and variance from both theoretical and simulational perspectives.

We first provide the estimation bias mechanism of the ptFCE algorithm motivated by
\begin{align}\label{app eq: estimation formula with noise}
    \mathbb{P}\left\{\lim_{n\rightarrow\infty}\mathcal{C}^{est,n}_{kl}(\xi) = 
    \frac{ \left\vert \mathbb{E}(\beta_k\beta_l)\frac{\overline{\widehat{h}_k(\xi)}\widehat{h}_l(\xi)}{\left\vert \widehat{h}_k(\xi) \widehat{h}_l(\xi) \right\vert}  e^{2\pi i \xi (t_{0,k}-t_{0,l})} + \frac{\Sigma_{kl}}{(T+1)\left\vert\widehat{N}(\xi)\right\vert^2\left\vert\widehat{h}_k(\xi) \widehat{h}_l(\xi)\right\vert} \right\vert }{ \sqrt{ \left(\mathbb{E}(\beta_k^2)+\frac{\Sigma_{kk}}{(T+1)\left\vert\widehat{N}(\xi)\widehat{h}_k(\xi)\right\vert^2}\right) \left(\mathbb{E}(\beta_l^2)+\frac{\Sigma_{ll}}{(T+1)\left\vert\widehat{N}(\xi)\widehat{h}_l(\xi)\right\vert^2}\right) } }\right\}=1, 
\end{align}
which is from Supplementary Theorem \ref{theorem: asymptotic bias}. Because the quantities in 
\begin{align*}
    \Sigma_{kk}\Big/[(T+1)\vert\widehat{N}(\xi)\widehat{h}_k(\xi)\vert^2]\approx 0,\ \ \ \Sigma_{ll}\Big/[(T+1)\vert\widehat{N}(\xi)\widehat{h}_l(\xi)\vert^2]\approx 0.
\end{align*}
are not exactly zeros, the asymptotic estimation bias
\begin{align*}
    \left\vert \lim_{n\rightarrow\infty}\mathcal{C}^{est,n}_{kl}(\xi) - \mathcal{C}_{kl}(\xi) \right\vert = \left\vert  \lim_{n\rightarrow\infty}\mathcal{C}^{est,n}_{kl}(\xi) - \frac{\left\vert\mathbb{E}(\beta_k \beta_l)\right\vert}{\sqrt{\mathbb{E}(\beta_k^2)\mathbb{E}(\beta_l^2)}} \right\vert
\end{align*}
always exists. We further model that approximation residual $W_{k'}(\omega;t)=R_{k'}(\omega;t)-\tilde{R}_{k'}(\omega;t)$ at different nodes are independent, implying $\Sigma_{kl}=0$ in (\ref{app eq: estimation formula with noise}). Then the following inequality, as a result of (\ref{app eq: estimation formula with noise}), indicates that our ptFCE algorithm tends to underestimate ptFC.
\begin{align}\label{app eq: underestimation inequality}
    \notag\lim_{n\rightarrow\infty}\mathcal{C}_{kl}^{est,n}(\xi)&=_{a.s.}\frac{ \left\vert \mathbb{E}(\beta_k\beta_l) \right\vert }{ \sqrt{ \left(\mathbb{E}(\beta_k^2)+\frac{\Sigma_{kk}}{(T+1)\left\vert\widehat{N}(\xi)\widehat{h}_k(\xi)\right\vert^2}\right) \left(\mathbb{E}(\beta_l^2)+\frac{\Sigma_{ll}}{(T+1)\left\vert\widehat{N}(\xi)\widehat{h}_l(\xi)\right\vert^2}\right) } }\\
    & \le
    \frac{ \left\vert \mathbb{E}(\beta_k\beta_l) \right\vert }{ \sqrt{ \mathbb{E}(\beta_k^2) \mathbb{E}(\beta_l^2) } } = \vert corr(\beta_k, \beta_l)\vert.
\end{align}
Furthermore, the larger the variance of the noise $\epsilon_k(\omega;t)$ in 
\begin{align*}
    Y_k(\omega;t)=\beta_k(\omega)\times \left(N*h_k\right)\left(t-t_{0,k}\right)+R_k(\omega;t),\ \ t\in\mathcal{T},\ \ k=1, \cdots, K,\ \ \omega\in\Omega,
\end{align*}
the less accurate is the approximation $R_{k'}(\omega;t)\approx\tilde{R}_{k'}(\omega;t)$ and the larger $\Sigma_{kk}$ and $\Sigma_{ll}$ in the inequality (\ref{app eq: underestimation inequality}). Hence, the larger the variance of noise $\epsilon_k(\omega;t)$, the larger the underestimation bias. The formula above also indicates that the larger the underlying ptFC, the larger the underestimation bias. Subsequent simulations confirm all the conclusions on bias presented herein. We also show that the underestimation bias is moderate in general in terms of estimating ptFCs.

We present the performance of ptFCE in estimating ptFC by investigating estimation bias and variance using simulated data. Specifically, suppose the true ptFC in our simulation mechanism is $\rho$, and the estimated ptFC from this algorithm is $\widehat{\rho}$, then we investigate bias $\widehat{\rho}-\rho$ and the variance of $\widehat{\rho}$. We implement the following data-generating mechanism, compatible with the definition of ptFC. The main steps of this mechanism are the following, while the details are provided in Appendix E. 

\noindent \textbf{Mechanism 0}: \textbf{Step 1}, we generate random coefficients $\left(\beta_{1}(\omega), \beta_{2}(\omega)\right)^T\sim N_2\left(\pmb{0}, (\sigma_{ij})_{1\le i,j \le 2}\right)$ for $\omega=1,\cdots,n$, and $\rho:=\vert\sigma_{12}/\sqrt{\sigma_{11}\sigma_{22}}\vert$ is the true underlying ptFC between synthetic nodes $1$ and $2$. \textbf{Step 2}, we generate reference signals $\{R_{k'}(\omega;\tau\Delta)=Q_{k'}(\omega;\tau\Delta)+\epsilon_{k'}(\omega;\tau\Delta)\}_{\tau=0}^T$ for $\omega=1,\cdots,n$ and $k'\in\{1,2\}$, where the $2nT$ noise values $\{\epsilon_{k'}(\omega;\tau\Delta) \vert k'=1,2; \omega=1,\cdots,n; \tau=0,\cdots,T\}\sim_{iid} N(0,V)$. \textbf{Step 3}, we compute the signals $\{Y_{k'}(\omega;\tau\Delta)\vert k'=1,2\}_{\tau=0}^T$, for $\omega=1,\cdots,n$, by $Y_{k'}(\omega; \tau\Delta)=9000+\beta_{k'}(\omega)\times \left(N * h_{k'}\right)(\tau\Delta) + R_{k'}(\omega; \tau\Delta)$. 

We investigate sample sizes $n\in\{50, 100,308, 1000\}$, where $308$ is the sample size of the HCP dataset used to obtain the results in our paper. We apply the ptFCE algorithm to estimate the underlying $\rho$ from synthetic signals. The corresponding estimate is denoted by $\widehat{\rho}$. For each $\rho\in\{0.25, 0.5, 0.75\}$, we repeat this procedure 500 times. The resulting estimates $\hat{\rho}$ are summarized in Supplementary Table \ref{table: simulation 1 summary} and Supplementary Figure \ref{fig: Boxplots_estimation}. 

\begin{table}
\centering
\caption{The summaries of the estimated $\widehat{\rho}$ in different underlying ptFC $\rho$ scenarios. The percentage in the parenthesis after each mean shows the corresponding relative bias $(\widehat{\rho}-\rho)/\rho$, where the minus signs indicate underestimation.}\label{table: simulation 1 summary}
\begin{tabular}{llllllllll}
\hline
         &  & $\rho=0.25$ &       &  & $\rho=0.5$ &       &  & $\rho=0.75$ &       \\ \cline{3-4} \cline{6-7} \cline{9-10} 
         &  & mean        & sd    &  & mean       & sd    &  & mean        & sd    \\ \hline
$n=50$   &  & 0.255 (2$\%$)      & 0.120 &  & 0.456 (-8.8$\%$)      & 0.114 &  & 0.670 (-10.7$\%$)       & 0.081 \\
$n=100$  &  & 0.246 (-1.6$\%$)       & 0.093 &  & 0.460 (-8$\%$)      & 0.079 &  & 0.671 (-10.5$\%$)       & 0.055 \\
$n=308$  &  & 0.248 (-0.8$\%$)       & 0.055 &  & 0.457 (-8.6$\%$)      & 0.047 &  & 0.670 (-10.7$\%$)      & 0.030 \\
$n=1000$ &  & 0.246 (-1.6$\%$)       & 0.028 &  & 0.462 (-7.6$\%$)      & 0.026 &  & 0.675 (-10$\%$)      & 0.018 \\ \hline
\end{tabular}
\end{table}
\begin{figure}[h]
    \centering
    \includegraphics[scale=0.65]{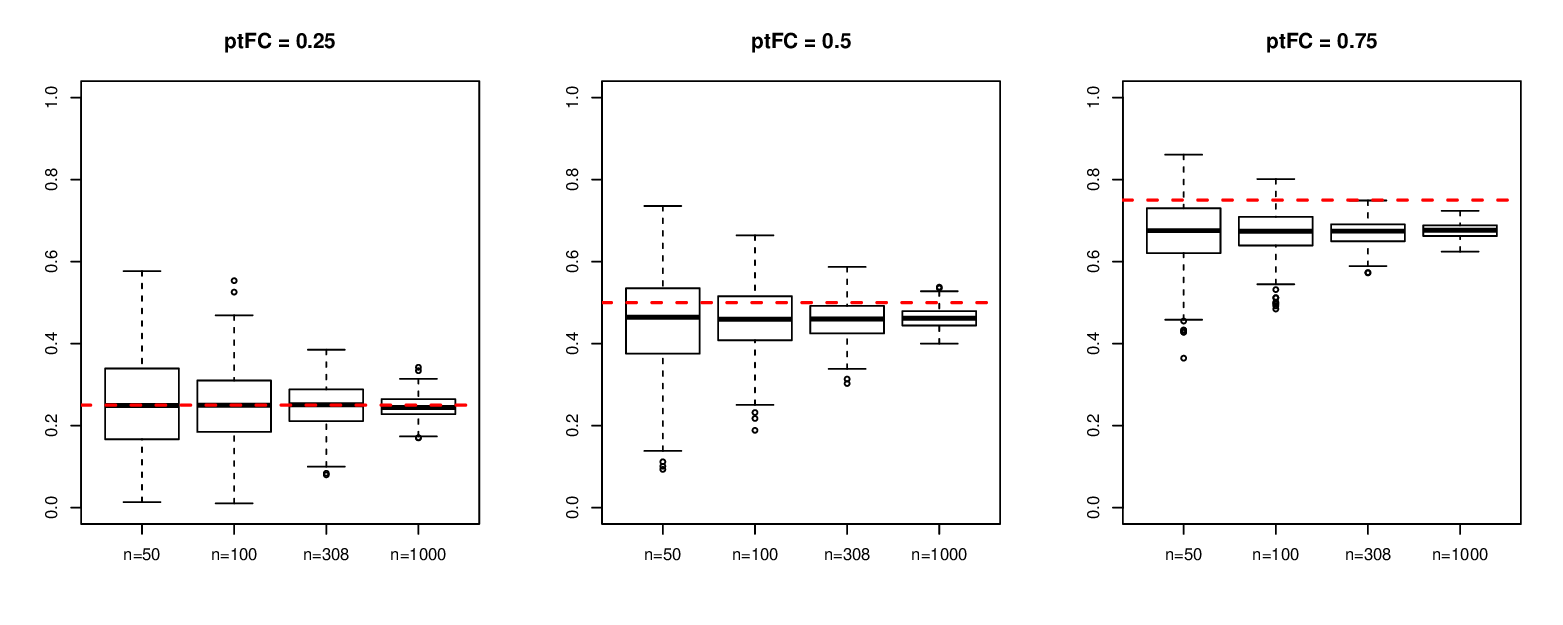}
    \caption{The boxplots summarizing the estimated $\widehat{\rho}$ in different underlying ptFC $\rho$ and different sample sizes $n$ scenarios.}
    \label{fig: Boxplots_estimation}
\end{figure}

In addition, we provide a simulation study for different signal-to-noise ratios. Specifically, we investigate the influence of random noise $\{\epsilon_{k'}(\omega;\tau\Delta)\}_{k',\tau,\omega}\sim_{iid} N(0, V)$ on the ptFCE algorithm. Specifically, we fix sample size $n=308$; for each $\lambda\in[0.5, 5]$, we conduct the 500-run simulation study described above except we change the random noise to $\{\epsilon_{k'}(\omega;\tau\Delta)\}_{k',\tau,\omega}\sim_{iid} N(0, \lambda V)$. For each $\lambda$ and simulation run $r$, the estimated ptFC between synthetic nodes is denoted as $\widehat{\rho}_{\lambda}^{(r)}$, then we obtain 500 curves $\{\widehat{\rho}_{\lambda}^{(r)}\vert\lambda\in[0.5,5]\}_{r=1}^{500}$ presented in Supplementary Figure \ref{fig: Influence of noise}. 

\begin{figure}
    \centering
    \includegraphics[scale=1, angle=270]{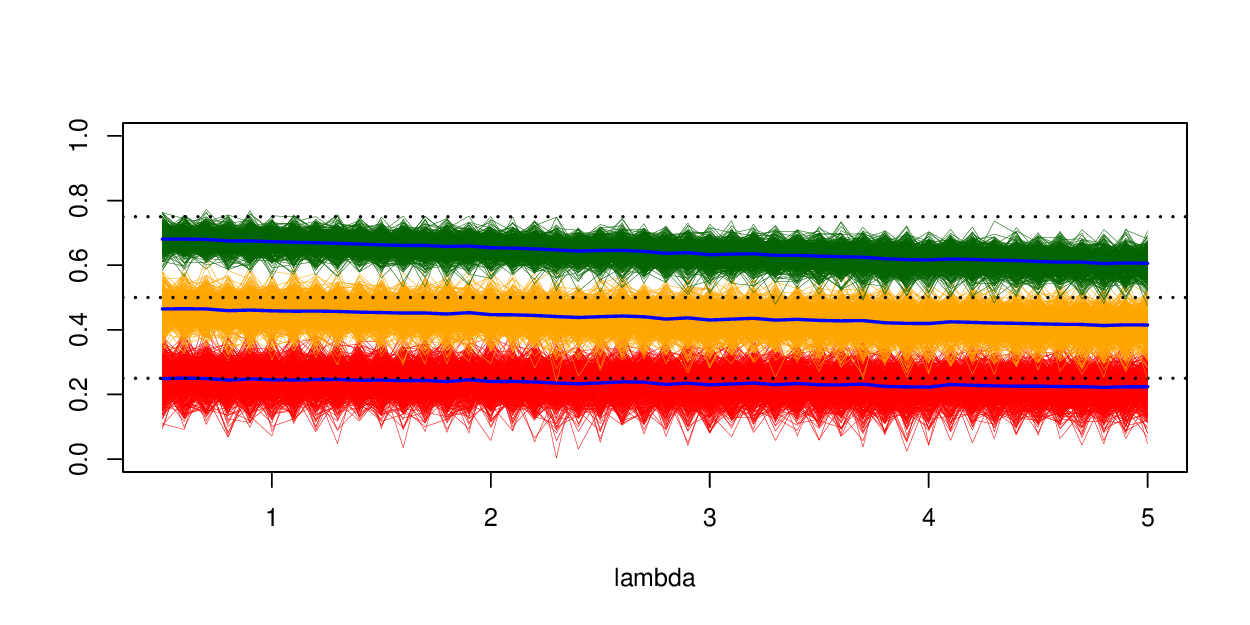}
    \caption{The red, orange, and green curves present the collections $\{\widehat{\rho}_{\lambda}^{(r)}\vert\lambda\in[0.5,5]\}_{r=1}^{500}$ corresponding to underlying ptFC $\rho=0.25, 0.5, 0.75$. The three blue solid curves present the mean curves $\{\frac{1}{500}\sum_{r=1}^{500} \widehat{\rho}_{\lambda}^{(r)}\vert\lambda\in[0.5,5]\}$ in the three underlying ptFC scenarios, and the three black dotted lines present the three underlying ptFC.}
    \label{fig: Influence of noise}
\end{figure}

As expected from the theoretical analysis of the bias mechanism, our proposed ptFCE algorithm tends to underestimate true ptFCs, see Supplementary Table \ref{table: simulation 1 summary} and Supplementary Figures \ref{fig: Boxplots_estimation} and \ref{fig: Influence of noise}. But Supplementary Table \ref{table: simulation 1 summary} shows that the bias is moderate. Simulations also confirm other conclusions on estimation bias presented in the theoretical analysis above. Additionally, Supplementary Figure \ref{table: simulation 1 summary} shows that the magnitude of noise $\epsilon_k(\omega;t)$ does not influence the variance of the ptFCE estimates. The increase in sample size reduces estimation variance. 

\section{Data-generating Mechanisms for Simulations}\label{section: data-generating mechanisms}

In this section, we provide the details of Mechanisms 0, 1, and 2 for generating synthetic signals in simulation studies in our paper. To mimic the HCP data of interest, we apply the following parameters in the two data-generating mechanisms.
\begin{itemize}
    \item Repeatition time $\Delta=0.72$ (seconds)
    \item Number of observation time points $T=283$.
    \item The task of interest (squeezing right toes) is presented by the stimulus signal $N(t)=\mathbf{1}_{[86.5, 98.5)}(t)+\mathbf{1}_{[162, 174)}(t)$.
    \item The tasks that are not of interest are presented by stimulus signals $\tilde{N}_1(t)=\mathbf{1}_{[71.35,83.35)}(t)+\mathbf{1}_{[177.125, 189.125)}(t)$,  $\tilde{N}_2(t)=\mathbf{1}_{[11, 23)}(t)+\mathbf{1}_{[116.63, 128.63)}(t)$, $\tilde{N}_3(t)=\mathbf{1}_{[26.13, 38.13)}(t)+\mathbf{1}_{[146.88, 158.88)}(t)$, and $\tilde{N}_4(t)=\mathbf{1}_{[56.26, 68.26)}(t)+\mathbf{1}_{[131.75, 143.75)}(t)$. They correspond to the tasks of squeezing left toes, squeezing left/right fingers, and moving tongue.
    \item HRF functions $h_k$ and $\tilde{h}_{k,\gamma}$, for $k\in\{1,3\}$ and $\gamma\in\{1,2,3,4\}$, are the double-gamma variate functions implemented in the \texttt{R} function \texttt{canonicalHRF} with parameters \texttt{a1}$=4$, \texttt{a2}$=10$, \texttt{b1}$=0.8$, \texttt{b2}$=0.8$, \texttt{c}$=0.4$; HRF functions $h_2$ and $\tilde{h}_{2,\gamma}$, for $\gamma\in\{1,2,3,4\}$, are the same function with parameters \texttt{a1}$=8$, \texttt{a2}$=14$, \texttt{b1}$=1$, \texttt{b2}$=1$, \texttt{c}$=0.3$. It is important to emphasize that none of the HRFs herein is canonical. The curves of all HRF functions used in our simulations studies are presented in Supplementary Figure \ref{fig: Different HRFs}.
\end{itemize}
\begin{figure}[h]
    \centering
    \includegraphics[scale=1]{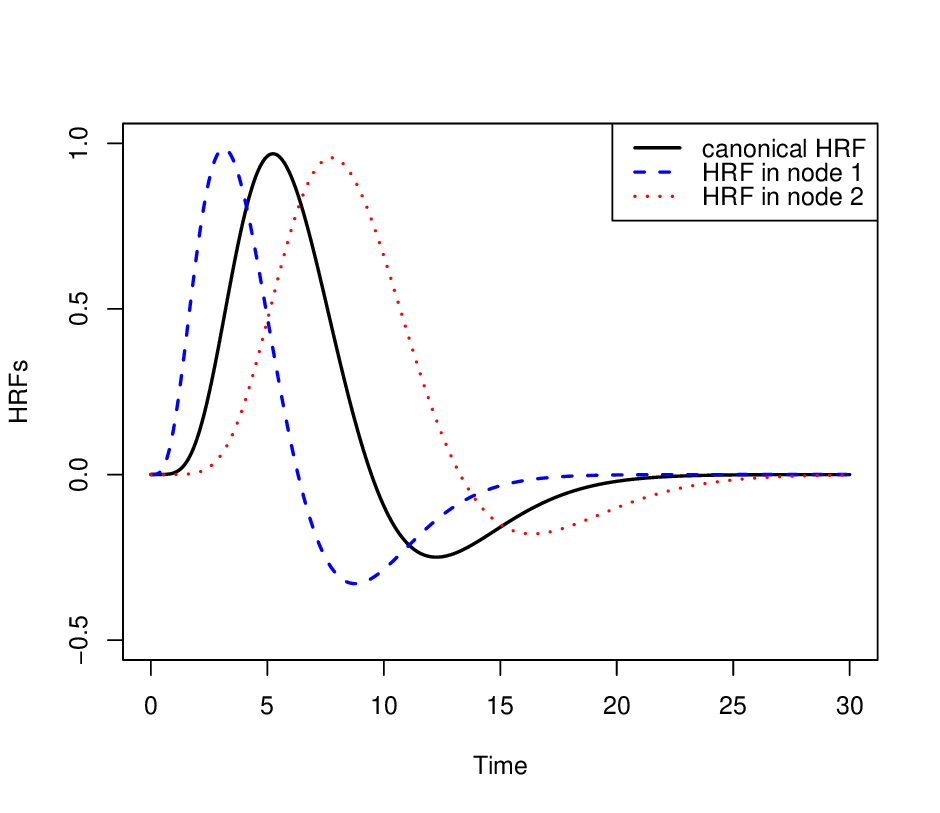}
    \caption{The black solid curve presents the canonical HRF (the \texttt{R} function \texttt{canonicalHRF} with default parameters). The blue dashed curve presents the \texttt{R} function \texttt{canonicalHRF} with parameters \texttt{a1}$=4$, \texttt{a2}$=10$, \texttt{b1}$=0.8$, \texttt{b2}$=0.8$, \texttt{c}$=0.4$. The red dotted curve presents the \texttt{R} function \texttt{canonicalHRF} with parameters \texttt{a1}$=8$, \texttt{a2}$=14$, \texttt{b1}$=1$, \texttt{b2}$=1$, \texttt{c}$=0.3$.}
    \label{fig: Different HRFs}
\end{figure}

\subsection{Mechanism 0}

This mechanism is based on the task-fMRI BOLD signal model proposed in our paper. Specifically, we implement the following model to generate synthetic data.
\begin{align}\label{app eq: data-generating model 1}
    & Y_{k'}(\omega;t)=9000+\beta_{k'}\times N*h_{k'}(t) + \left\{\sum_{\gamma=1}^4 \beta_{k',\gamma}(\omega)\times \tilde{N}_{\gamma}*\tilde{h}_{k',\gamma}(t)\right\}+\epsilon_{k'}(\omega;t),\\
    \notag& \mbox{where } t\in\mathcal{T}=\{\tau\Delta\}_{\tau=0}^T \mbox{ and }k'\in\{1,2\}.
\end{align}
We generate signals $\{(Y_1(\omega;\tau\Delta), Y_2(\omega;\tau\Delta)\}_{\tau=0}^T$, for $\omega=1,\cdots,n$, by the following steps.
\begin{itemize}
    \item \textbf{Step 1}: Generate bivariate normal random vectors $\left(\beta_{1}(\omega), \beta_{2}(\omega)\right)^T\sim_{i.i.d.} N_2\left((0,0)^T, (\sigma_{ij})_{1\le i,j \le2}\right)$ for $\omega=1,\cdots,n$, where $\rho:=\vert\sigma_{12}/\sqrt{\sigma_{11}\sigma_{22}}\vert$ is the true underlying ptFC, and 
    \begin{align}\label{app eq: underlying rho}
        \sigma_{11}=2,\ \ \sigma_{22}=3,\ \ \sigma_{12}=\rho\times\sqrt{\sigma_{11}\sigma_{22}}.
    \end{align}
    \item\textbf{Step 2}: For each $\gamma\in\{1,2,3,4\}$, generate bivariate normal random vector
    \begin{align*}
    \begin{pmatrix}
        \beta_{1,\gamma}(\omega)\\
        \beta_{2,\gamma}(\omega)
    \end{pmatrix}
        \sim_{i.i.d.} N_2\left(\begin{pmatrix}
        0\\
        0
        \end{pmatrix}, \begin{pmatrix}
        2 & 0.3\times\sqrt{2\times3}\\
        0.3\times\sqrt{2\times3} & 3
        \end{pmatrix}\right),
    \end{align*}
    for $\omega=1,\cdots,n$. The four collections $\left\{\left(\beta_{1,\gamma}(\omega), \beta_{2,\gamma}(\omega)\right)^T\right\}_{\omega=1}^n$, for $\gamma\in\{1,2,3,4\}$, are independently generated.
    \item\textbf{Step 3}: For each fixed $\omega\in\{1,\cdots,n\}$, generate (white) noise as follows.
    \begin{align*}
        \begin{pmatrix}
        \epsilon_1(\omega;\tau\Delta)\\
        \epsilon_2(\omega;\tau\Delta)
        \end{pmatrix}\sim_{i.i.d.} N_2\left(\begin{pmatrix}
        0\\
        0
        \end{pmatrix}, \begin{pmatrix}
        30 & 0\\
        0 & 30
        \end{pmatrix}\right), \mbox{ for }\tau=0,\cdots,T.
    \end{align*}
    The $n$ collections $\left\{\left(\epsilon_1(\omega;\tau\Delta), \epsilon_2(\omega;\tau\Delta)\right)\right\}_{\tau=0}^T$, for $\omega\in\{1,\cdots,n\}$, are independently generated.
    \item\textbf{Step 4}: Compute the signals $\{(Y_1(\omega;\tau\Delta), Y_2(\omega;\tau\Delta)\}_{\tau=0}^T$, for $\omega=1,\cdots,n$, by model (\ref{app eq: data-generating model 1}).
\end{itemize}

A pair of simulated $\{(Y_1(\omega;\tau\Delta), Y_2(\omega;\tau\Delta)\}_{\tau=0}^T$ using Mechanisms 0 is presented in Supplementary Figure \ref{fig: Signals from Mechanism 1}.
\begin{figure}
    \centering
    \includegraphics[scale=0.6]{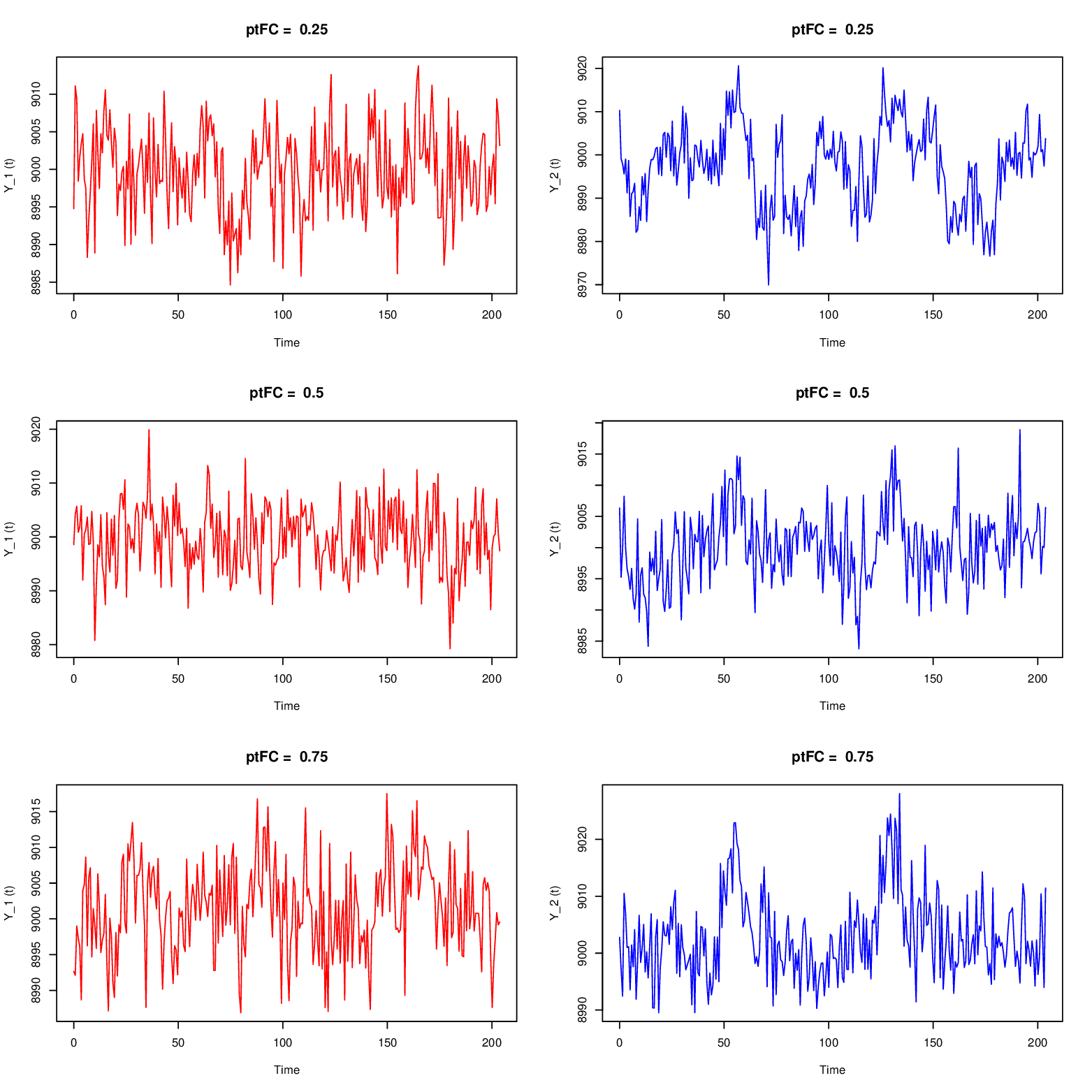}
    \caption{Left panels present signals generated for node 1, and right panels present signals generated for node 2. The underlying ptFCs for pairs in the upper, middle, and lower rows are 0.25, 0.5, and 0.75, respectively. }
    \label{fig: Signals from Mechanism 1}
\end{figure}

\subsection{Mechanism 1}

Mechanism 1 is very similar to Mechanism 0 and is based on (\ref{app eq: data-generating model 1}) as well, except for $k'\in\{1,2,3\}$. We generate signals $\{(Y_1(\omega;\tau\Delta), Y_2(\omega;\tau\Delta), Y_3(\omega;\tau\Delta)\}_{\tau=0}^T$, for $\omega=1,\cdots,n$, using the following steps.
\begin{itemize}
    \item \textbf{Step 1}: Generate normal random vectors $\left(\beta_{1}(\omega), \beta_{2}(\omega), \beta_{3}(\omega)\right)^T\sim_{i.i.d.} N_3\left(\pmb{0}, (\sigma_{ij})_{1\le i,j \le 3}\right)$ for $\omega=1,\cdots,n$, where  
    \begin{align}\label{app eq: underlying rho}
        \left(\sigma_{ij}\right)_{1\le i,j \le3}=\begin{pmatrix}
        2 & \rho_{12}\times\sqrt{2\times3} & 0\\
        \rho_{12}\times\sqrt{2\times3} & 3 & \rho_{23}\times\sqrt{2\times3} \\
        0 & \rho_{23}\times\sqrt{2\times3} & 2
        \end{pmatrix},
    \end{align}
    and $\rho_{ij}:=\vert\sigma_{ij}/\sqrt{\sigma_{ii}\sigma_{jj}}\vert$ for $(i,j)\in\{(1,2), (2,3)\}$ are the true underlying ptFCs.
    \item\textbf{Step 2}: For each $\gamma\in\{1,2,3,4\}$, generate bivariate normal random vector
    \begin{align*}
    \begin{pmatrix}
        \beta_{1,\gamma}(\omega)\\
        \beta_{2,\gamma}(\omega)\\
        \beta_{3,\gamma}(\omega)
    \end{pmatrix}
        \sim_{i.i.d.} N_3\left(\begin{pmatrix}
        0\\
        0\\
        0
        \end{pmatrix}, \begin{pmatrix}
        2 & 0.3\times\sqrt{2\times3} & 0\\
        0.3\times\sqrt{2\times3} & 3 & 0.3\times\sqrt{2\times3} \\
        0 & 0.3\times\sqrt{2\times3} & 2
        \end{pmatrix}\right),
    \end{align*}
    for $\omega=1,\cdots,n$. The four collections $\left\{\left(\beta_{1,\gamma}(\omega), \beta_{2,\gamma}(\omega), \beta_{3,\gamma}(\omega)\right)\right\}_{\omega=1}^n$, for $\gamma\in\{1,2,3,4\}$, are independently generated.
    \item\textbf{Step 3}: For each fixed $\omega\in\{1,\cdots,n\}$, generate (white) noise as follows.
    \begin{align*}
        \begin{pmatrix}
        \epsilon_1(\omega;\tau\Delta)\\
        \epsilon_2(\omega;\tau\Delta)\\
        \epsilon_3(\omega;\tau\Delta)
        \end{pmatrix}\sim_{i.i.d.} N_3\left(\begin{pmatrix}
        0\\
        0\\
        0
        \end{pmatrix}, \begin{pmatrix}
        30 & 0 & 0\\
        0 & 30 & 0\\
        0 & 0 & 30
        \end{pmatrix}\right), \mbox{ for }\tau=0,\cdots,T.
    \end{align*}
    The $n$ collections $\left\{\left(\epsilon_1(\omega;\tau\Delta), \epsilon_2(\omega;\tau\Delta), \epsilon_3(\omega;\tau\Delta)\right)\right\}_{\tau=0}^T$, for $\omega\in\{1,\cdots,n\}$, are independently generated.
    \item\textbf{Step 4}: Compute the signals $\{(Y_1(\omega;\tau\Delta), Y_2(\omega;\tau\Delta), Y_3(\omega;\tau\Delta)\}_{\tau=0}^T$, for $\omega=1,\cdots,n$, by model (\ref{app eq: data-generating model 1}).
\end{itemize}

\subsection{Mechanism 2}

This mechanism is motivated by the Pearson correlation approach and implemented by the following steps.
\begin{itemize}
    \item \textbf{Step 1}: For each fixed $\omega\in\{1,\cdots,n\}$, we independently generate $(\epsilon_1(\omega;\tau\Delta), \epsilon_2(\omega;\tau\Delta),\epsilon_2(\omega;\tau\Delta))^T$, for $\tau\in\{1,\cdots,T\}$, using the following distributions:
    
    \noindent (i) If $N(\tau\Delta)=1$, we apply
    \begin{align}\label{app eq: underlying varrho}
        \begin{pmatrix}
        \epsilon_1(\omega;\tau\Delta)\\
        \epsilon_2(\omega;\tau\Delta)\\
        \epsilon_3(\omega;\tau\Delta)
        \end{pmatrix}\sim N_3\left(
        \begin{pmatrix}
        0\\
        0\\
        0
        \end{pmatrix}, 
        \begin{pmatrix}
        30 & \varrho\times30 & 0\\
        \varrho\times30 & 30 & \varrho\times30\\
        0 & \varrho\times30 & 30
        \end{pmatrix}
        \right);
    \end{align}
    if $N(\tau\Delta)=0$, we implement
    \begin{align*}
        \begin{pmatrix}
        \epsilon_1(\omega;\tau\Delta)\\
        \epsilon_2(\omega;\tau\Delta)\\
        \epsilon_3(\omega;\tau\Delta)
        \end{pmatrix}\sim N_3\left(
        \begin{pmatrix}
        0\\
        0\\
        0
        \end{pmatrix}, 
        \begin{pmatrix}
        30 & 0 & 0\\
        0 & 30 & 0\\
        0 & 0 & 30
        \end{pmatrix}
        \right).
    \end{align*}
    The $n$ collections $\left\{\left(\epsilon_1(\omega;\tau\Delta), \epsilon_2(\omega;\tau\Delta), \epsilon_3(\omega;\tau\Delta)\right)\right\}_{\tau=0}^T$, for $\omega\in\{1,\cdots,n\}$, are independently generated.
    \item\textbf{Step 2}: Compute the signals $\{(Y_1(\omega;\tau\Delta), Y_2(\omega;\tau\Delta), Y_3(\omega;\tau\Delta)\}_{\tau=0}^T$, for $\omega=1,\cdots,n$, by the following.
    \begin{align}\label{app eq: Mechanism 2 model}
    & Y_{k'}(\omega;t)=9000+ N*h_{k'}(t) + \left\{\sum_{\gamma=1}^4  \tilde{N}_{\gamma}*\tilde{h}_{k',\gamma}(t)\right\}+\epsilon_{k'}(\omega;t),\\
    \notag& \mbox{where } t\in\mathcal{T}=\{\tau\Delta\}_{\tau=0}^T \mbox{ and }k'\in\{1,2,3\}.
\end{align}
\end{itemize}

Simulated $\{(Y_1(\omega;\tau\Delta), Y_2(\omega;\tau\Delta), Y_3(\omega;\tau\Delta)\}_{\tau=0}^T$ using Mechanisms 1 and 2, are presented in Supplementary Figure \ref{fig: Signals from Mechanism 2}.

\begin{figure}[h]
    \centering
    \includegraphics[scale=0.9, angle=270]{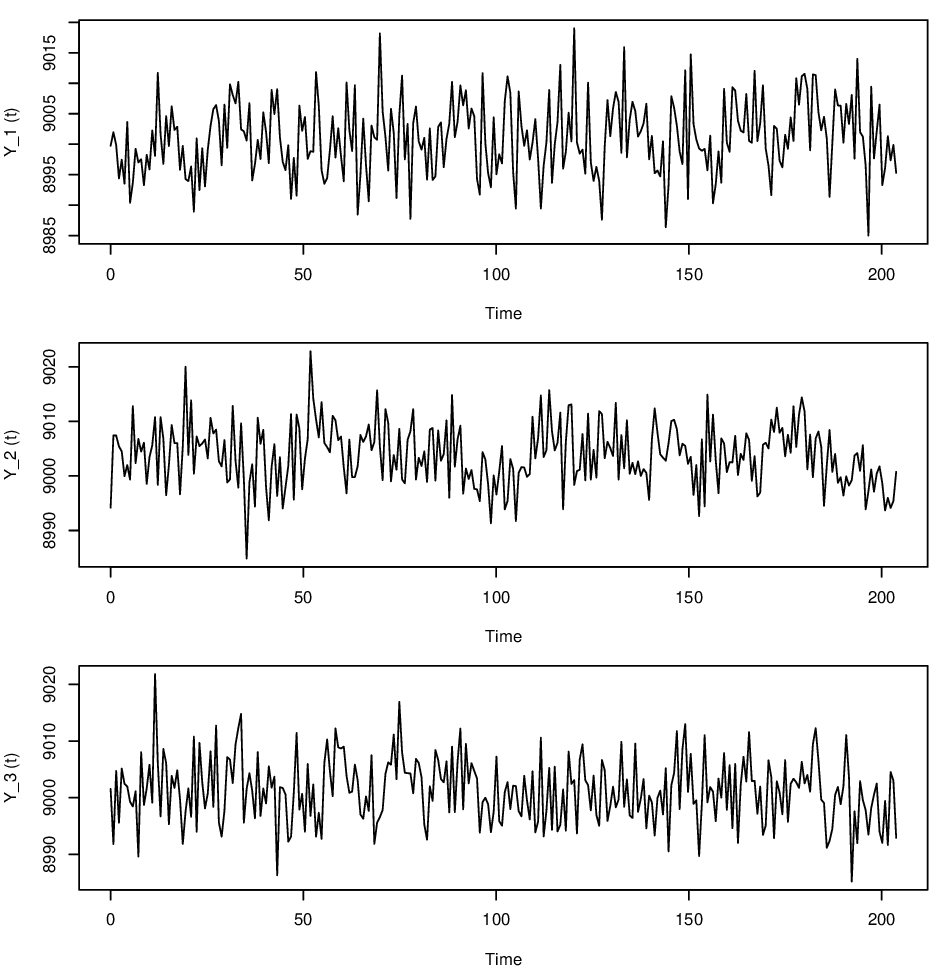}
    \caption{Signals generated for nodes 1, 2, and 3 using Mechanism 2.}
    \label{fig: Signals from Mechanism 2}
\end{figure}

\subsection{The Data-generating Mechanism for the 50-node Study}\label{section: The Data-generating Mechanism for the 50-node Study}

We provide two data-generating mechanisms based on Mechanisms 1 and 2 (see Appendix E.2 and E.3). 

First of all, we generate a $50\times 50$ correlation matrix $(\rho_{ij})_{1\le i,j\le 50}$ by the following as our ground truth connectivity structure
\begin{enumerate}
\item For all $i=1,2,\ldots,50$, we define $\rho_{ii}=1$.
\item For $i=1,\ldots,25$ and $j=i+1,\ldots,50$, randomly generate $\rho_{ij}=\rho_{ji}\sim \operatorname{Unif}(0,\,0.4)$.
\item For $i=26,\ldots,50$ and $j=i+1,\ldots,50$, randomly generate $\rho_{ij}=\rho_{ji}\sim \operatorname{Unif}(0.6,\,1)$.
\end{enumerate}

\subsubsection{Mechanism 1-based Approach}

With the correlation matrix $(\rho_{ij})_{1\le i,j\le 50}$ generated above, we generate signals $\{(Y_1(\omega;\tau\Delta),\ldots, Y_{50}(\omega;\tau\Delta)\}_{\tau=0}^T$ for only one synthetic subject $\omega$ using the following steps
\begin{itemize}
    \item \textbf{Step 1}: Generate normal random vector $\left(\beta_{1}(\omega), \beta_{2}(\omega), \ldots, \beta_{50}(\omega)\right)^T\sim_{i.i.d.} N\left(\pmb{0}, (3\rho_{ij})_{1\le i,j \le 50}\right)$.
    \item\textbf{Step 2}: For each $\gamma\in\{1,2,3,4\}$, generate normal random vector $(\beta_{1,\gamma}(\omega), \beta_{2,\gamma}(\omega),\ldots, \beta_{50,\gamma}(\omega))^T\sim N(\boldsymbol{0}, (3\Tilde{\rho}_{ij})_{1\le i,j\le 50})$, where $\Tilde{\rho}_{ij}=1$ if $i=j$ and $\Tilde{\rho}_{ij}=0.3$ if $i\ne j$.
    \item\textbf{Step 3}: For each $\tau=0,1,2,\ldots,T=283$, generate white noise as follows
    \begin{align*}
        \left(\epsilon_1(\omega;\tau\Delta), \epsilon_2(\omega;\tau\Delta),\ldots, \epsilon_{50}(\omega;\tau\Delta)\right) \sim N(\boldsymbol{0}, \boldsymbol{I}_{50\times 50})
    \end{align*}
    where the $\boldsymbol{I}_{50\times 50}$ denotes the $50\times 50$ identity matrix. 
    \item\textbf{Step 4}: Compute the signals $\{(Y_1(\omega;\tau\Delta), Y_2(\omega;\tau\Delta), \ldots, Y_{50}(\omega;\tau\Delta)\}_{\tau=0}^T$ by model (\ref{app eq: data-generating model 1}).
\end{itemize}

\subsubsection{Mechanism 2-based Approach}

With the correlation matrix $(\rho_{ij})_{1\le i,j\le 50}$ generated at the beginning of Section \ref{section: The Data-generating Mechanism for the 50-node Study}, we generate signals $\{(Y_1(\omega;\tau\Delta),\ldots, Y_{50}(\omega;\tau\Delta)\}_{\tau=0}^T$ for only one synthetic subject $\omega$ using the following steps
\begin{itemize}
    \item \textbf{Step 1}: We independently generate $(\epsilon_1(\omega;\tau\Delta), \epsilon_2(\omega;\tau\Delta),\ldots,\epsilon_{50}(\omega;\tau\Delta))^T$, for $\tau\in\{1,\cdots,T\}$, using the following distributions:
    
    \noindent (i) If $N(\tau\Delta)=1$, we apply $\left( \epsilon_1(\omega;\tau\Delta), \epsilon_2(\omega;\tau\Delta),\ldots,\epsilon_{50}(\omega;\tau\Delta) \right)^T \sim N\left(\boldsymbol{0}, (30\rho_{ij})_{1\le i,j\le 50} \right)$;
    if $N(\tau\Delta)=0$, we implement $\left( \epsilon_1(\omega;\tau\Delta), \epsilon_2(\omega;\tau\Delta),\ldots,\epsilon_{50}(\omega;\tau\Delta) \right)^T \sim N\left(\boldsymbol{0}, \boldsymbol{I}_{50\times 50} \right)$.
    \item\textbf{Step 2}: Compute the signals $\{(Y_1(\omega;\tau\Delta), Y_2(\omega;\tau\Delta),\ldots, Y_{50}(\omega;\tau\Delta)\}_{\tau=0}^T$ by the following.
    \begin{align*}
    & Y_{k'}(\omega;t)=9000+ N*h_{k'}(t) + \left\{\sum_{\gamma=1}^4  \tilde{N}_{\gamma}*\tilde{h}_{k',\gamma}(t)\right\}+\epsilon_{k'}(\omega;t),\\
    \notag& \mbox{where } t\in\mathcal{T}=\{\tau\Delta\}_{\tau=0}^T \mbox{ and }k'\in\{1,2,\ldots,50\}.
\end{align*}
\end{itemize}





\section{Limitations of the Pearson correlation approach}\label{section: Supplementary Figures}

Supplementary Figure \ref{fig: HRF Pearson} illustrates the limitations of the Pearson correlation approach defined as follows in terms of the influence of latency/HRF variance and noise.
\begin{align}\label{app eq: Pearson correlation with the convolution model}
    \left\vert corr(P_k, P_l)\right\vert=\left\vert \int_{\mathcal{T}} \phi_k(t) \times \phi_l(t)\mu(dt)\Bigg/\sqrt{ \int_{\mathcal{T}}\left\vert \phi_k(t)\right\vert^2 \mu(dt)\times \int_{\mathcal{T}}\left\vert \phi_l(t)\right\vert^2 \mu(dt)}\right\vert.
\end{align}

\begin{figure}
    \centering
    \includegraphics[scale=0.55]{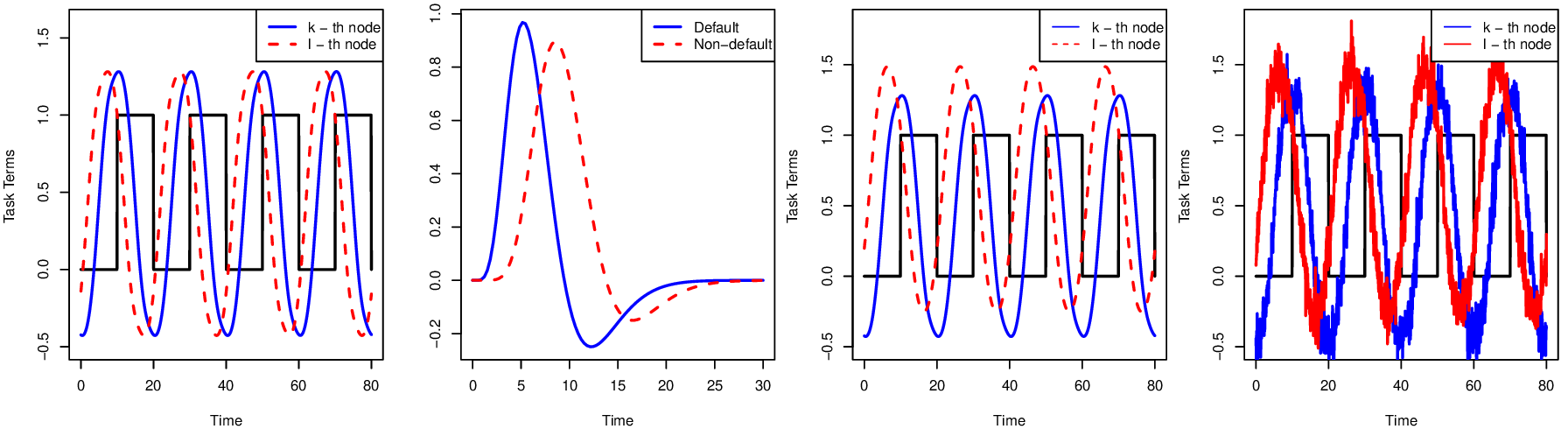}
    \put(-460, -10){\small{(a): $\vert corr(P_k, P_l)\vert=0.571$}}
    \put(-295, -10){\small{(b)}}
    \put(-225, -10){\small{(c): $\vert corr(P_k, P_l)\vert=0.445$}}
    \put(-105, -10){\small{(d): $\vert corr(P_k, P_l)\vert=0.422$}}
    \caption{An example illustrating the limitations of the Pearson correlation approach in (\ref{app eq: Pearson correlation with the convolution model}). Let $N(t)=\sum_{m=1}^4\mathbf{1}_{(20m-10, 20m]}(t)$ be the stimulus signal of a block design, $\mathcal{T}=[0,80]$, and $h_k$ be the \texttt{canonicalHRF} function with default parameters in the \texttt{R} package \texttt{neuRosim}. $h_k$ is illustrated by the solid blue curve in (b). Panel (a) shows the influence of the variation in latency on (\ref{app eq: Pearson correlation with the convolution model}), where $h_l(t)=h_k(t+3)$, i.e., $t_{0,k}=-3$; the task-evoked terms at the $k^{th}$ and $l^{th}$ nodes are presented by blue and red curves, respectively. In (b), $h_l$ is replaced by the \texttt{canonicalHRF} function with parameters $\texttt{a1}=10, \texttt{a2}=15, \texttt{b1}=\texttt{b2}=0.9, \texttt{c}=0.35$, and presented by the dashed red curve. Panel (c) shows the influence of variation in HRF on (\ref{app eq: Pearson correlation with the convolution model}), and the task-evoked terms at the $k^{th}$ and $l^{th}$ nodes are presented by blue and red curves, respectively, where $h_l$ is defined in (b) as the red curve. Panel (d) is a noise-contaminated version of (c), i.e., curves of $P_k(t)+\epsilon_k(t)$ (blue) and $P_l(t)+\epsilon_l(t)$ (red) with $\epsilon_k(t), \epsilon_l(t)\sim_{iid} N(0, 0.1)$ for each $t$. Panel (d) shows that random noise can further influence (\ref{app eq: Pearson correlation with the convolution model}). In panels (a,c,d), the presented $\vert corr(P_k, P_l)\vert$ is computed by (\ref{app eq: Pearson correlation with the convolution model}).}
    \label{fig: HRF Pearson}
\end{figure}

\section{Details of Beta-series Regression and Coherence Analysis}\label{section: Details of Beta-series Regression and Coherence Analysis}

\noindent\textbf{Beta-series regression:} We apply the procedure described in Chapter 9.2 of \cite{ashby2019statistical} to estimate task-evoked FC using beta-series regression, except we ignore the ``nuisance term" therein. For each subject $\omega$ and underlying $\rho_{ij}$ or $\varrho_{ij}$ with $i<j$, we apply beta-series regression to signals $\{Y_i(\omega;\tau\Delta), Y_j(\omega;\tau\Delta)\}_{\tau=0}^T$, estimate the FC evoked by the task of interest $N(t)$, and denote the estimated quantity by $\hat{\rho}^{betaS}_{ij,\omega}$. Then we compute the mean and median of $\{\hat{\rho}^{betaS}_{ij,\omega}\}_{\omega=1}^{n}$ across all $\omega=1,\cdots,n$ and denote them by $\hat{\rho}^{betaS}_{ij,mean}$ and $\hat{\rho}^{betaS}_{ij,median}$. 

\noindent\textbf{Coherence analysis:} For each $\omega$ and $\rho_{ij}$ or $\varrho_{ij}$, compute the coherence between signals $\{Y_i(\omega;\tau\Delta)\}_\tau$ and $ \{Y_j(\omega;\tau\Delta)\}_{\tau}$ by the \texttt{R} function \texttt{coh} in package \texttt{seewave}. The coherence is a function $coh_{ij, \omega}(\xi)$ of $\xi$. Since HRFs act as band-pass filters \citep[$0-0.15$ Hz,][]{aguirre1997empirical}, compute the median of $coh_{ij, \omega}(\xi)$ across all $\xi\in(0,0.15)$ and denote it by $\hat{\rho}^{Coh}_{ij, \omega}$. The mean and median of $\{\hat{\rho}^{Coh}_{ij, \omega}\}_{\omega=1}^{n}$ across all $\omega$ are denoted by $\hat{\rho}^{Coh}_{ij, mean}$ and $\hat{\rho}^{Coh}_{ij, median}$, respectively.

\section{Future Research}\label{section: Future Research}

An interesting extension of our proposed model is the inclusion of an interaction term between the $P_k(\omega;t)$ and $Q_k(\omega;t)$ in model $Y_k(\omega;t)=P_k(\omega;t) + Q_k(\omega;t) + \epsilon_k(\omega;t)$. In order to experimentally validate whether the inclusion of such an interaction term would be biologically relevant, novel experimental designs would be needed. Particularly, we may consider an experiment where the subjects are at rest for a certain period during the scanning session followed by the performance of the task. Using this design, we may obtain estimates of $P_k$ and $Q_k$ considering the model with and without an interaction term and discuss the biological relevance of the results. Future work may investigate this issue and develop theoretical conditions for the identifiability of terms in a model with an interaction term and estimation algorithms.

In many applications, task-evoked FC at the subject level instead of the population level is of interest. In our subsequent research, we will propose a framework parallel to the ptFC one at the subject level. Additionally, implementing the \textit{persistent homology} (PH) framework to FC estimates is an effective way of circumventing the choice of threshold for FC measurements, e.g., see \cite{lee2001relative}. Applying the PH approach to our proposed ptFCs is left for future research as well.

\end{appendix}


\bibliography{sample}

\end{document}